%% file: RepOrExp.tex
\documentclass[12pt]{article}
\usepackage{hyperref,natbib,amsmath,amssymb, amsthm}
\usepackage{import, color,graphicx}

\usepackage{fullpage}
\usepackage{epsf}
\usepackage{epsfig}
\usepackage{amsfonts}
\usepackage{setspace}
\usepackage[boxed]{algorithm}
\usepackage{algpseudocode}
\usepackage{algorithmicx}
\usepackage[usenames, dvipsnames]{xcolor}

\usepackage{caption}
\usepackage{subcaption}
\usepackage[normalem]{ulem}

\newtheorem{Proposition}{Proposition}[section]

\DeclareMathOperator{\Diag}{Diag}
\DeclareMathOperator{\tr}{tr}
\DeclareMathOperator\argmin{arg min}

\newcommand{\bm}[1]{\mbox{\boldmath $#1$}}

\newcommand{\bE}[0]{\mathbb{E}}
\newcommand{\bV}[0]{\mathbb{V}\mathrm{ar}}
\newcommand{\mN}[0]{\mathcal{N}}
\newcommand{\R}{\mathbb{R}}
\newcommand{\x}{\mathbf{x}}
\newcommand{\xu}{\bar{\mathbf{x}}} 
\newcommand{\yu}{\bar{y}} 
\newcommand{\z}{\mathbf{z}}
\newcommand{\veck}{\mathbf{k}}
\newcommand{\veckN}{\mathbf{k}_N}
\newcommand{\veckn}{\mathbf{k}_n}

\newcommand{\vecX}{\mathbf{X}}
\newcommand{\vecY}{\mathbf{Y}}
\newcommand{\vecYu}{\bar{\mathbf{Y}}} 

\newcommand{\An}{\mathbf{A}_n}
\newcommand{\B}{\mathbf{B}}
\newcommand{\Cg}{\mathbf{C}_{(g)}}

\newcommand{\Cn}{\mathbf{C}_n}

\newcommand{\Deltan}{\boldsymbol{\Delta}_n}

\newcommand{\K}{\mathbf{K}}
\newcommand{\KN}{\mathbf{K}_N}
\newcommand{\Kn}{\mathbf{K}_n}

\newcommand{\hnug}{\hat{\nu}_{(g)}}

\newcommand{\ones}{\bm{1}}
\newcommand{\Lan}{\boldsymbol{\Lambda}_n}

\newcommand{\nuN}{\hat{\nu}_N}
\newcommand{\Or}{\mathcal{O}}

\newcommand{\xnew}{\tilde{\x}}
\newcommand{\xnewp}{\tilde{\x}_{(p)}}
\newcommand{\cs}{\check{\sigma}}
\newcommand{\vecc}{\mathbf{c}}
\newcommand{\vecd}{\mathbf{d}}
\newcommand{\vecg}{\mathbf{g}}
\newcommand{\vech}{\mathbf{h}}
\newcommand{\veckg}{\mathbf{k}_{(g)}}
\newcommand{\vecu}{\mathbf{u}}
\newcommand{\vecw}{\mathbf{w}}
\newcommand{\Hmat}{\mathbf{H}}
\newcommand{\W}{\mathbf{W}}
\newcommand{\Wn}{\mathbf{W}_n}
\newcommand{\M}{\mathbf{M}}
\newcommand{\m}{\mathbf{m}}
\newcommand{\erf}{\mathrm{erf}}

\newtheorem{lemma}{Lemma}[section]

\title{\vspace{-1cm}  Replication or exploration? Sequential design for stochastic simulation experiments}
\author{Micka\"{e}l Binois\thanks{Corresponding author: The University of Chicago Booth School of Business, 
Chicago IL, 60637;
\href{mailto:mbinois@mcs.anl.gov}{\tt mbinois@mcs.anl.gov}}
\and
Jiangeng Huang\thanks{Department of Statistics, Virginia Tech}
\and
Robert B.~Gramacy\footnotemark[2]
\and Mike Ludkovski\thanks{Department of Statistics and Applied Probability,
University of California Santa Barbara}
}

\begin{document}
\maketitle

\begin{abstract}
 	
We investigate the merits of replication, and provide methods for optimal
design (including replicates), with the goal of obtaining globally
accurate emulation of {\em noisy} computer simulation experiments. We first
show that replication can be beneficial from both design and computational
perspectives, in the context of Gaussian process surrogate modeling. We then
develop a lookahead based sequential design scheme that can determine if a new
run should be at an existing input location (i.e., replicate) or at a new one
(explore). When paired with a newly developed heteroskedastic Gaussian process
model, our dynamic design scheme facilitates learning of signal and noise
relationships which can vary throughout the input space. We show that it does
so efficiently, on both computational and statistical grounds. In addition to
illustrative synthetic examples, we demonstrate performance on two challenging
real-data simulation experiments, from inventory management and epidemiology.

\bigskip
  \noindent {\bf Keywords:}
  computer experiment, Gaussian process, surrogate model, input-dependent
  noise, replicated observations, lookahead

\end{abstract}


\section{Introduction}

Historically, design and analysis of computer experiments focused on
deterministic solvers from the physical sciences via Gaussian process (GP)
interpolation \citep{Sacks1989}. But nowadays computer modeling is common in
the social \citep[][Chapter 8]{cioffi2014introduction}, management
\citep{Law2015} and biological \citep{johnson:2008} sciences, where
stochastic simulations abound. Noisier simulations demand bigger experiments
to isolate signal from noise, and more sophisticated GP models---not just
adding nuggets to smooth the noise, but variance processes to track changes in
noise throughout the input space in the face of heteroskedasticity
\citep{Binois2016}. In this context there
are not many simple tools: most add to, rather than reduce,
modeling and computational complexity.

Replication in the experiment design is an important exception, offering a
pure look at noise, not obfuscated by signal.  Since usually the signal is of
primary interest, a natural question is: How much replication should be
performed in a simulation experiment? The answer to that question depends on a
number of factors. In this paper the focus is on
global surrogate model prediction accuracy and computational efficiency, and we show
that replication can be a great benefit to both, especially for
heteroskedastic systems.

There is evidence to support this in the literature. \cite{Ankenman2010}
demonstrated how replicates could facilitate signal isolation, via stochastic
kriging (SK), and that accuracy could be improved without much extra
computation by augmenting existing degrees of replication in stages
\citep[also see][]{Liu2010,Quan2013,Mehdad2018}. \cite{Wang2017} showed
that replicates have an important role in characterizing sources of inaccuracy
in SK. \citet{Boukouvalas2014} demonstrated the value of replication in
(Fisher) information metrics, and \citet{Plumlee2014} provided asymptotic
results favoring replication in quantile regression.  Finally, replication has
proved helpful in the surrogate-assisted (i.e., Bayesian) optimization of
noisy blackbox functions \citep{Horn2017,Jalali2017}.

However, none of these studies address what we see as the main decision
problem for design of GP surrogates in the face of noisy simulations. That is:
how to allocate a set of unique locations, and the degree of replication
thereon, to obtain the best overall fit to the data.  That sentiment has been echoed
independently in several recent publications
\citep{Kleijnen2015,Weaver2016,Jalali2017,Horn2017}.  The standard approach of
allocating a uniform number of replicates leaves plenty of room for
improvement. One exception is \cite{Chen2014,Chen2017} who proposed several
criteria to explore the replication/exploration trade-off, but only for a
finite set of candidate designs.

Here we tackle the issue sequentially, one new design element at a time.  We
study the conditions under which the new element should be a replicate, or
rather explore a new location, under an integrated mean-square prediction
error (IMSPE) criterion.  We also highlight how replicates offer computational
savings in surrogate model fitting and prediction with GPs, 
augmenting results of \cite{Binois2016}
with fast updates as new data arrives.
Inspired by those findings, we develop a new IMSPE-based criterion that
offers lookahead over future replicates.  This criterion is the first to
acknowledge that exploring now offers a new site for replication later, and
conversely that replicating first offers the potential to learn a little more
(cheaply, in terms of surrogate modeling computation) before committing to a
new design location.  A key component in solving this sequential decision
problem in an efficient manner is a closed form expression for IMSPE, and its
derivatives, allowing for fast numerical optimization.

While our IMSPE criterion corrects for myopia in replication, it is important
to note that it is not a full lookahead scheme.  Rather, we illustrate that it
is biased toward replication: longer lookahead horizons tend to tilt toward
more replication in the design. In our experience, full lookahead, even when
approximated, is impractical for all but the most expensive simulations. Even
the cleverest dynamic programming-like schemes
\citep[e.g.,][]{Ginsbourger2010,Gonzalez2016,Lam2016,Huan2016} require
approximation to remain tractable or otherwise only manage to glimpse a few
steps into the future despite enormous computational cost. Our more thrifty
scheme can search dozens of iterations ahead.  That flexibility
allows us to treat the horizon as a tuning parameter that can be adjusted,
online, to meet design and/or surrogate modeling goals. When simulations are
cheap and noisy, we provide an adaptive horizon scheme that favors replicates
to keep surrogate modeling costs down; when surrogate modeling costs are less
of a concern, we provide a scheme that optimizes out-of-sample RMSE, which
might or might not favor longer horizons (i.e., higher replication).

The structure of the remainder of the paper is as follows. First we summarize
relevant elements of GPs, sequential design and the computational savings
enjoyed through replication in Section \ref{sec:concept}. Then in Section
\ref{sec:imspe} we detail IMSPE, with emphasis on sequential applications and
computational enhancements (e.g., fast GP updating) essential for the
tractability of our framework. Section \ref{sec:ahead} discusses our lookahead
scheme, while  Section \ref{sec:implement}
provides practical elements for the implementation, including tuning the
horizon of the lookahead scheme. Finally, in Section \ref{sec:experiment}
results are presented from several simulation experiments, including
illustrative test problems, and real simulations from epidemiology and
inventory management, which benefit from disparate design strategies.

\section{Background and proof of concept}
\label{sec:concept}

Here we introduce relevant surrogate modeling and design elements while at the
same time illustrating proof-of-concept for our main methodological
contributions.  Namely that replication can be valuable computationally, as
well as for accuracy in surrogate modeling.

\subsection{Gaussian process regression with replication}

We consider Gaussian process (GP) surrogate models for an unknown function over a fixed domain
$f: D \subset \R^d \to \R$ based on noisy observations $\vecY = (y_1, \ldots,
y_N)^\top$ at design locations $\vecX = (\x_1^\top, \ldots, \x_N^\top)$. For
simplicity, we assume a zero-mean GP prior, completely specified by covariance
kernel $k(\cdot, \cdot)$, a positive definite function. Many different choices
of kernel are possible, while in the computer experiments literature the power
exponential and Mat\'ern families are the most common.
Often the families are parametererized by unknown quantities such as
lengthscales, scales, etc., which are inferred from data
\citep[see, e.g.,][]{Rasmussen2006,Santner2013}.  The noise is presumed to be
zero-mean i.i.d.\ Gaussian, with variance $r(\x) = \bV[ Y(\x) | f(\x)]$. While
we discuss our preferred modeling and inference apparatus in Section
\ref{sec:hetGP}, for now we make the (unrealistic) assumption that kernel
hyperparameters, along with the potentially non-constant
$r(\x)$, are known.
Altogether, the data-generating mechanism follows
a multivariate normal distribution, $\vecY
\sim \mN_N(0,
\KN)$, where $\KN$ is an $N \times N$ matrix comprised of $k(\x_i, \x_j) +
\delta_{ij} r(\x_i)$, for $1
\leq i,j \leq N$ and with $\delta_{ij}$ being the Kronecker delta function.

Conditional properties of multivariate normal
(MVN) distributions yield that the predictive distribution $Y(\x) | \vecY$ is Gaussian with
\begin{align}
\mu_N(\x) &= \bE(Y(\x)| \vecY) =  \veckN(\x)^\top \KN^{-1} \vecY, \quad
\mbox{ with }\; \veckN(\x) = (k(\x, \x_1), \ldots, k(\x, \x_N))^\top;  \nonumber \\
\sigma_N^2(\x) &= \bV(Y(\x)| \vecY) =
k(\x, \x) + r(\x) - \veckN^\top(\x) \KN^{-1} \veckN(\x).
\label{eq:Npred} 
\end{align}

It can be shown that $\mu(\x)$
is a best linear unbiased predictor (BLUP) for $Y(\x)$ (and $f(\x)$). Although
testaments to the high accuracy and attractive uncertainty quantification
features abound in the literature, one notable drawback is that when $N$ is
large the computational expense of $\Or(N^3)$ due to decomposing $\KN$ (e.g.,
to solve for $\KN^{-1}$) can be prohibitive.

When the observations $y(\x)$ are deterministic (i.e., $r(\x) = 0$), often $N$
can be kept to a manageable size.  When data are noisy, with potentially
varying noise level, many samples may be needed to separate signal from noise.
Indeed in our motivating applications, the signal-to-noise ratios can be very
low, so even for a relatively small input space, thousands of training
observations are necessary. In that context replication can offer significant
computational gains.  To illustrate, let $\xu_i$, $i= 1, \ldots, n$
denote the $n \leq N$ unique input locations, and $y_i^{(j)}$ be the
$j^\mathrm{th}$ out of $a_i \ge 1$ replicates observed at $\xu_i$, i.e.,
$j=1,\ldots, a_i$, where $\sum_{i = 1}^n a_i = N$. Also, let $\vecYu_{(N,n)} =
(\yu_1, \ldots,
\yu_n)^\top$ store averages over replicates, $\yu_i =
\frac{1}{a_i}\sum_{j =1}^{a_i} y_i^{(j)}$.
Then \citet{Binois2016} show that predictive equations based on
this ``unique-$n$'' formulation, i.e., following Eq.~(\ref{eq:Npred})
except with $\vecYu_{(N,n)}$ and $\K_{(N,n)} = \left( k(\xu_i, \xu_j) +
\delta_{ij}
\frac{r(\xu_i)}{a_i}
\right)_{1 \leq i,j \leq n}$, are identical. Compared to
 the ``full-$N$'' formulation, the respective costs are reduced from
 $\mathcal{O}(N^3)$ to just $\mathcal{O}(n^3)$, without any approximations.

\subsection{Sequential design for GPs}

Although there are many criteria dedicated to design for GP regression \citep[see,
e.g.,][]{Pronzato2012}, our focus here is on global predictive accuracy defined
via integrated mean-squared prediction error (IMSPE). Fixing $\vecX$, the
IMSPE integrates the ``de-noised'' posterior variance $\cs^2_N(\x) = \sigma_N^2(\x) -
r(\x)$ over $D$,
\begin{equation}
\mathrm{IMSPE}(\x_1, \ldots, \x_N) = \int \limits_{\x \in D}
\cs^2_N(\x) \, d\x =: I_N.
\label{eq:IMSPE}
\end{equation}
Note that although this definition removes $r(\x)$, it is still present in
$\KN$ and therefore affects $\cs^2_N(\x)$. Removing $r(\x)$ is not required,
but since $\int r(\x) \, d\x$ is constant over $\x_1, \dots, \x_N$, it simplifies
future expressions.

Even in the highly idealized case were all covariance
$k(\cdot, \cdot)$ and noise $r(\cdot)$ relationships are presumed known, one-shot
design---i.e., choosing all $N$ locations $\vecX$ at once to minimize
\eqref{eq:IMSPE}---is an extraordinarily difficult task owing to the $(N
\times d)$-dimensional search space. Only in very specific cases, such as
$d=1$ and a exponential kernel \citep{Antognini2010}, or with the
simpler task of allocating $N$ replicates to a fixed set of $n$ unique sites
$\xu_1,
\ldots, \xu_n$ \citep{Ankenman2010}, is a computationally tractable solution
known.

Therefore, we consider here the simpler case of a purely sequential design, building up
a big design greedily, one simulation at a time. Note that this means that $N$
grows by 1 after each iteration.  While $n$ is also evolving, the precise
change is dependent on whether a replicate or a new location is selected. In the
generic step, we condition on existing $\x_1,\ldots, \x_N$ locations and
optimize $\mathrm{IMSPE}( \x_1, \ldots, \x_N, \x_{N+1})$ over $\x_{N+1}$.
Recall that the posterior variance $\cs^2_N$ only depends on the geometry of
$\vecX$, i.e., it is independent of the outputs $\vecY$ and hence we can view the
above as minimizing $I_{N+1}(\x_{N+1}) :=
\mathrm{IMSPE}(\x_{N+1} | \x_1, \ldots, \x_N)$.  Later we establish specific
closed-form expressions both for $I_{N+1}$ and its gradient which
enable fast optimization via library-based numerical schemes. Foreshadowing
these developments, and utilizing the calculations detailed therein, we
illustrate here the possibility that $\x_{N+1} = \argmin_{\x}
I_{N+1}(\x)$ is a replicate.  The conditions under which replication is
advantageous, which we describe shortly in Section \ref{ssec:closed}, have to
our knowledge only been illustrated empirically \citep{Boukouvalas2010}, or
conceptually (e.g., \citet{Wang2017} highlight that replication is more
beneficial as the signal-to-noise ratio decreases, via upper bounds on the
MSPE), or to bolster technical results (e.g., \cite{Plumlee2014} demand a
sufficient degree of replication to ensure asymptotic efficiency).

\begin{figure}[ht!]
  \centering
  \vspace{-0.5cm}
  \begin{subfigure}[t]{0.4\textwidth}%
	\centering%
	\includegraphics[width=\textwidth, trim = 0 10 0 30, clip = TRUE]{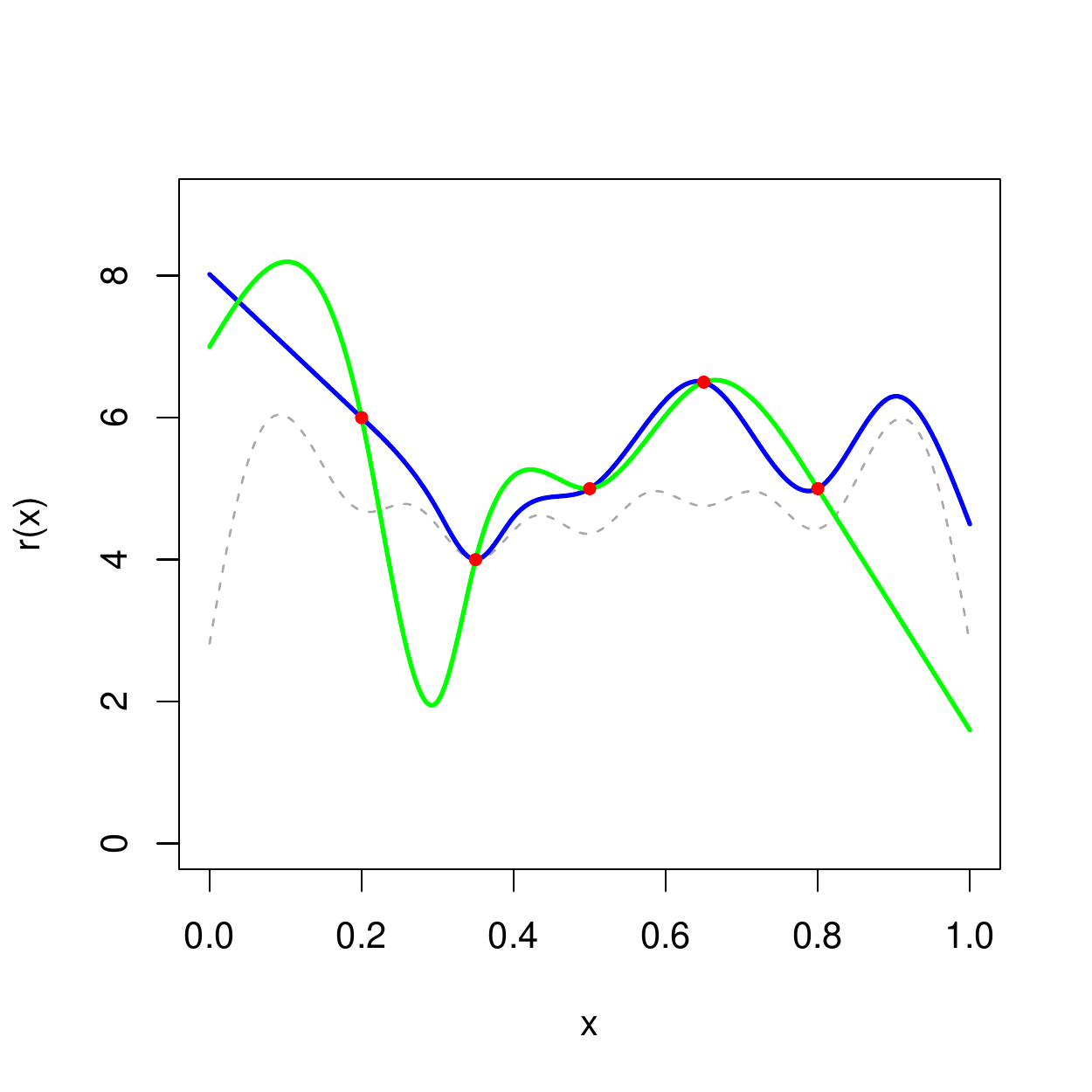}%
	\end{subfigure}%
	  \begin{subfigure}[t]{0.4\textwidth}%
	\centering%
	\includegraphics[width=\textwidth, trim = 0 10 0 30, clip = TRUE]{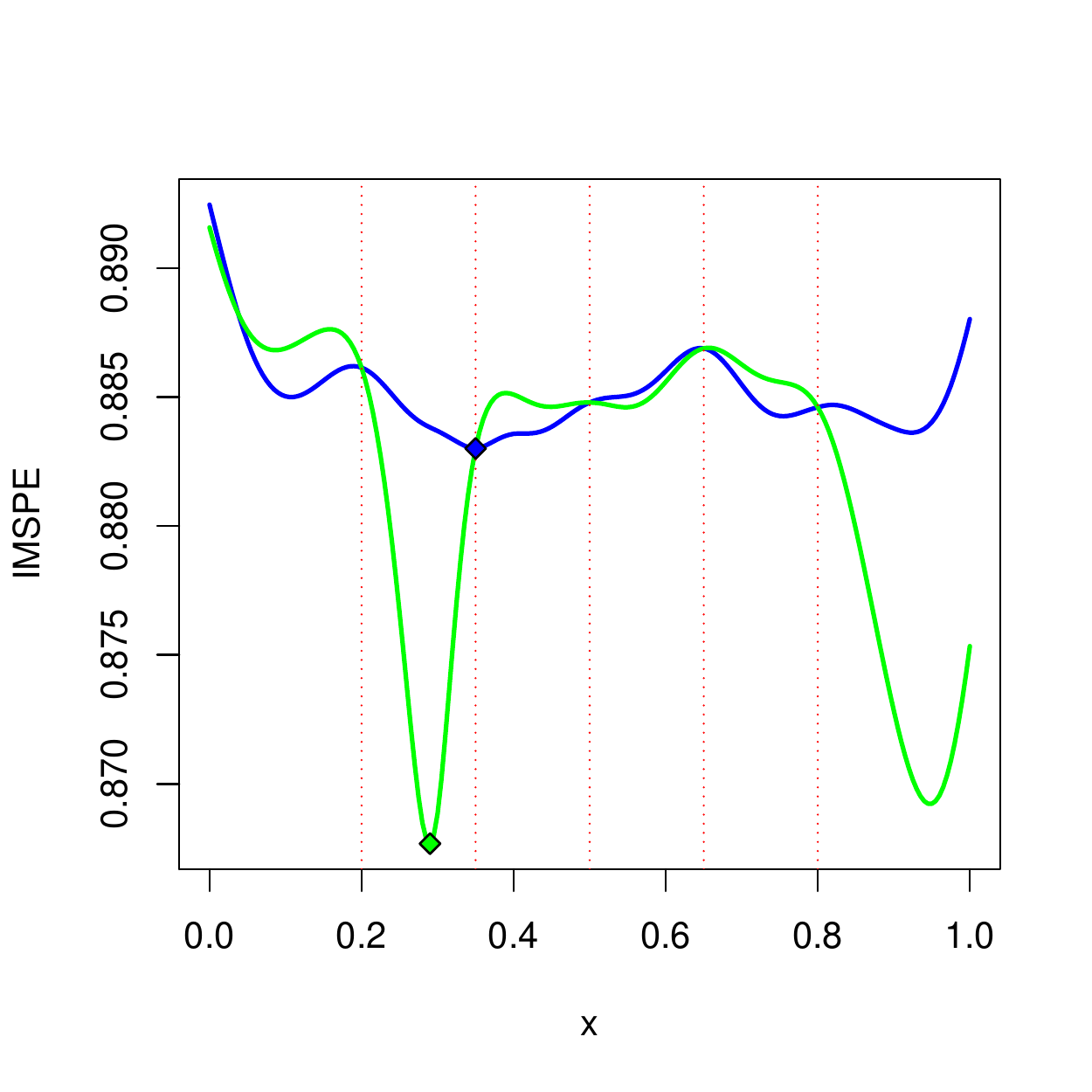}%
	\end{subfigure}%
	\vspace{-0.25cm}
  \caption{Illustration of the effect of noise variance on IMSPE
  optimization. Left: two examples of noise variance functions $r(\cdot)$
  (blue solid and green dashed lines), with observations at $\vecX$ (five red points). The
  grey dotted line represents the minimum $r(\x)$ that guarantees that
  replicating is optimal. Right:  $I_{N+1}(\x)$ for the two respective
  $r(\cdot)$. Diamonds highlight minimum values, and red dotted lines the
  existing designs $\x_1, \ldots, \x_5$. }
  \label{fig:rep}%
\end{figure}

The {\em left} panel of Figure \ref{fig:rep} shows two different noise levels,
$r(\x)$, for a stylized heteroskedastic GP predictor trained at $N=5$
locations whose $\x_1, \ldots, \x_5$ values are shown as red dots.  The fact
that the two $r(\x)$ curves coincide at these locations is not material to
this discussion. Later in Section \ref{ssec:closed} this feature and a
description of the gray-dotted curve will be provided.  The right panel
in the figure shows the predicted IMSPE, $I_{N+1}(\x) = \mathrm{IMSPE}(\x_1, \ldots,
\x_5, \x)$ derived from $\cs_N^2$ calculations using those $\x_1, \ldots,
\x_5$ values combined with the two $r(\x)$ settings.  With smaller IMSPE being
better, we see that the solid blue regime calls for $\x_6$ being a replicate
$(\argmin I_{N+1}(\x)= \x_2)$, whereas the dashed green regime wants to explore at a
new unique location ($\argmin I_{N+1}(\x) \simeq 0.32$).  Also note that the
IMSPE surfaces are multi-modal, which may pose a challenge to numerical
optimizers, and that even for the dashed green curve there are replicates (e.g., at
$\x_2$) with lower IMSPE than some local minima, meaning that augmenting with
a cheap discrete search over replicates may be more effective than deploying a
multi-start optimization scheme.

Ultimately, we entertain the far more realistic setup of unknown kernel
hyperparameters and noise processes.  In this context, sequential design to
``learn-as-you-go'' is essential.  We take this approach not simply to avoid
pathologies in hyperparameter mis-specification, as discussed in homoskedastic
setups \citep[e.g.,][]{seo00,Krause2007}, but explicitly to gain the
flexibility to sample non-uniformly in a manner that can only be adapted after
a degree of initial sampling allows a fit of the noise process $\hat{r}(\x)$
to be obtained, and further refined.  Our empirical results illustrate that
reasonable, yet inaccurate, {\em a priori} simplifications such as constant
$r(\x)$ may---even if just for the purposes of design, not subsequent
fitting---lead to inferior prediction. Previously such adaptive behavior and
non-uniform sampling was only available via more cumbersome fully
nonstationary methods, say involving treed partitioning \citep{gra:lee:2009}.

\section{IMSPE through the lens of replication}
\label{sec:imspe}

Over the years several authors \citep[e.g.,][]{Ankenman2010,anagnos:gramacy:2013,Burnaev2015,Leatherman2017}
have provided closed form expressions
for IMSPE (i.e., for the integral in \eqref{eq:IMSPE}) via variations on the
criterion's definition (i.e., versions somewhat different than our preferred
version in (\ref{eq:IMSPE})), or via simplifications to the GP specification
or to the argument $\x_1,\ldots,\x_N$, obtained by constraining the search
set.  Others have argued in more general contexts that $d$-dimensional
numerical integration, usually via sums over a (potentially random) reference
set, is the only viable option in their setting
\citep{seo00,gra:lee:2009,Gauthier2014,Gorodetsky2016,Pratola2017}.

Here we provide a new closed-form expression for the IMSPE which, despite
being intimately connected to earlier versions, is quite general and, we
think, could replace many of the prevailing numerical schemes.  This
development uses the ``unique-$n$'' representation for efficient calculation
under replication, however the analogue ``full-$N$'' version is immediate.  We
then consider an ``add one'' variation, $I_{N+1}(\xnew) = \mathrm{IMSPE}(\xnew
| \x_1, \ldots, \x_N)$, for efficient calculation in the sequential design
setting and derive a condition under which replication is preferred for the
next sample.  Here we use $\xnew$ for a potential new location, while
$\x_{N+1}$ will ultimately be chosen as the best candidate (i.e., minimizing
IMSPE over $\xnew$). Note that if $\x_{N+1}$ turns out to be a replicate, $n$
would not increase.

One important reason to have a closed-form IMSPE is the
calculation of gradients, also in closed form, to aid in optimization.  We
provide the first such derivative expressions of which we are aware.  Finally,
acknowledging the dual role of replication (to speed calculations and
separate signal from noise) we describe two new lookahead IMSPE heuristics for
tuning the lookahead horizon in an online fashion, depending on whether speed
or accuracy is more important.

\subsection{IMSPE closed-formed expressions}
\label{ssec:closed}

We start by writing the IMSPE, shorthanded as $I_N$ in
Eq.~(\ref{eq:IMSPE}), as an expectation: \[ I_N  = \int \limits_{\x \in D}
\cs^2_n(\x) \; d\x = \bE[\cs^2_n(X)] = \bE[k(X, X)] - \bE[\veckn(X)^\top
\Kn^{-1} \veckn(X)] \] with $X$ uniformly sampled in $D$, and using the
linearity of the expectation. Notice that $\K_n$ depends on the number of
replicates per unique design, so this representation includes a tacit
dependence on the noise and replication counts $a_1, \ldots, a_n$. Then,
as shown in Lemma \ref{lem:IMSPE_closed}, the integration of $\cs_n^2$ over
$D$ may be reduced to integrations of the covariance function.

\begin{lemma}
Let $\Wn$ be an $n \times n$ matrix with entries comprising integrals of kernel
products $w(\x_i, \x_j) =
\int_{\x \in D} k(\x_i, \x) k(\x_j, \x) \; d\x$ for $1 \leq i,j \leq n$, and
let $E = \int_{\x \in D} k(\x,\x) \; d\x$. Then
\vspace{-0.25cm}
\begin{equation}
I_N = E - \tr(\Kn^{-1}\Wn). \label{eq:Ifull}
\vspace{-1.0cm}
\end{equation}
\label{lem:IMSPE_closed}
\end{lemma}
\begin{proof}
The first part, involving $E$, follows simply by definition.  For the second
part, let $\z$ be a random vector of size $n$ with mean $\m$ and covariance
$\M$. \cite{Petersen2008} provides that $\bE \left[\z^\top \Kn^{-1} \z \right]
= \tr \left( \Kn^{-1} \M \right) + \m^\top \Kn^{-1} \m$. Therefore
using $\m^\top
\Kn^{-1} \m = \tr(\Kn^{-1} \m \m^\top)$, we have
$\bE[\veckn(X)^\top \Kn^{-1} \veckn(X)] = \tr( \Kn^{-1}(\M + \m \m^\top) )$
where $\m = \bE[ \veckn(X)]$ and $\M = \mathbb{C}ov( \veckn(X)^\top, \veckn(X))$. Observing that
$\Wn = \M + \m\m^\top$ gives the desired result.
\end{proof}

Our interest in the re-characterization in \eqref{eq:Ifull} is three-fold.
First and foremost, some of the most commonly used kernels enjoy closed form
expressions of $E$ and $w(\cdot, \cdot)$. In Appendix \ref{ap:kernel_exp} we
provide $w(\cdot, \cdot)$ for (i) Gaussian, (ii) Mat\'ern-5/2, (iii)
Mat\'ern-3/2, and (iv) Mat\'ern-1/2 families. For those families,
$E$ further reduces to their scale hyperparameter.  Section
\ref{sec:hetGP} offers specific forms for the generic expression
(\ref{eq:Ifull}) under our {\tt hetGP} model. Second, note that even when
closed forms are not available, as may arise when the kernel $k(X, X)$ cannot
be analytically integrated over $D$, this formulation may still be
advantageous. Numerically integrating $k(\x, \cdot)$ inside $\Wn$ will likely
be far easier than the alternative of integrating $\cs_n^2$, which can be
highly multi-modal. Third, we remark that $\tr(\Kn^{-1} \Wn) = \ones^\top(
\Kn^{-1} \circ \Wn) \ones$ where $\circ$ stands for the Hadamard (i.e.,
element-wise) product. Once $\Kn^{-1}$ and $\W_n$ are computed, the cost is in
$\Or(n^2)$, whereas the na\"ive alternative is $\Or(n^3)$.

Now, in sequential application the goal is to choose a new $\x_{N+1}$ by
optimizing $I_{N+1}(\xnew)$ over candidates $\xnew$.  Fixing the first $n$ unique
design elements simplifies calculations substantially if we assume that
$\K_n^{-1}$ and $\Wn$ are previously available. In that case, write
\[
\K_{n+1}=\begin{bmatrix}
\K_{n} & \veckn(\xnew) \\
\veckn(\xnew) ^\top & k(\xnew,\xnew) + r(\xnew)
 \end{bmatrix}, \quad
\W_{n+1}= \begin{bmatrix}
 \Wn & \vecw(\xnew)\\
\vecw(\xnew)^\top & w(\xnew, \xnew)
 \end{bmatrix}
\]
with $\vecw(\xnew) = \left(w(\xnew, \xu_i) \right)_{1 \leq i \leq n}$.
The partition inverse equations
\citep{barnett:1979} give
\begin{equation}
\K_{n+1}^{-1}=\begin{bmatrix}
\K_{n}^{-1}+ \vecg(\xnew) \vecg(\xnew)^ \top \sigma_n^2(\xnew) & \vecg(\xnew) \\
\vecg(\xnew) ^\top & \sigma_n^2(\xnew)^{-1}
 \end{bmatrix},
 \label{eq:Knp1i_up}
 \end{equation}
where $\vecg(\xnew)= -\sigma_n^2(\xnew)^{-1}\K_n^{-1} \veckn(\xnew)$ and $\sigma_n^2(\xnew)
= \cs^2_n(\xnew) + r(\xnew)$ as in \eqref{eq:Npred}.
Combining those two
results together leads to
\vspace{-0.5cm}
\begin{align} \notag
I_{N+1}(\xnew) &
 = E - \ones^\top [\K_{n+1}^{-1} \circ \W_{n+1}] \ones  \\
&= E - \left(\ones^\top [\K_{n}^{-1} \circ \W_{n}] \ones
+ \sigma_n^2(\xnew) \vecg(\xnew)^\top \Wn \vecg(\xnew) + 2 \vecw(\xnew)^\top \vecg(\xnew)
+ \sigma_n^2(\xnew)^{-1} w(\xnew, \xnew) \right)  \nonumber \\
&= I_N - \left( \sigma_n^2(\xnew) \vecg(\xnew)^\top \Wn \vecg(\xnew)
+ 2 \vecw(\xnew)^\top \vecg(\xnew) + \sigma_n^2(\xnew)^{-1}  w(\xnew, \xnew) \right). \label{eq:I}
\end{align}
Both (\ref{eq:Knp1i_up}--\ref{eq:I}) only require $\Or(n^2)$ computation.

After optimizing the latter part of \eqref{eq:I} over $\xnew$ and choosing the
best new design $\x_{N+1}$, one may utilize those inverse equations again to
update the GP fit.   Although similar identities have been provided in the
literature \citep[e.g.,][]{gramacy:polson:2011,Chevalier2014}, the ones we
provide here are the first to exploit the thrifty ``unique-$n$''
representation, and to tailor to the setting where $\x_{N+1}$ is a
replicate, i.e., an $\xu_k$, for $k \in \{1,\dots,n\}$, versus a new distinct
$\xu_{n+1}$ location.

\begin{lemma}
Suppose $\x_{N+1} = \xu_k$. Then the updated predictive mean and variance (increasing $N$
but not $n$) are given by
\begin{align*}
\mu_{(N+1,n)}(\x) & := \mu_{(N,n)}(\x) + \veckn(\x)^\top(\K_{(N,n)}^{-1} (\vecYu_{(N+1,n)} - \vecYu_{(N,n)}) - \B_k \vecYu_{(N+1,n)}), \\
\sigma^2_{(N+1,n)} (\x) &= \sigma^2_{(N,n)}(\x) - \veckn(\x)^\top\B_k \veckn(\x),
\end{align*}
with  $\B_k = \frac{\left(\K_{(N,n)}^{-1} \right)_{.,k} \left(\K_{(N,n)}^{-1}\right)_{k,.} }{a_k (a_k + 1)/r(\xu_k) - \left(\K_{(N,n)}\right)^{-1}_{k,k}}$, a rank-one matrix.
\label{lem:rep}
\end{lemma}
\begin{proof}
By adding a replicate at $\xu_k$, the only change is to augment $a_k$ by
one in $\K_{(N+1,n)}$, namely $\K_{(N+1,n)} - \K_{(N,n)} = - \Diag\left(0, \ldots, 0, \frac{r(\xu_k)}{a_k (a_k +
1)}, 0, \ldots, 0 \right) =: -r(\xu_k) \vecu \vecu^\top = r(\xu_k)\vecu' \vecu^\top$ with $\vecu'
= -\vecu$. Similarly, $\vecYu_{(N+1, n)} -
\vecYu_{(N,n)} = \left( 0, \ldots, 0, \frac{1}{a_k+1}( y_k^{(a_{k+1})} - \bar{y}_k^{(N)}), \ldots, 0 \right)$ has only one non-zero element, residing in position $k$.

The Sherman-Morrison (i.e., rank-one Woodbury) formula gives
\begin{equation}
\K_{(N+1,n)}^{-1} = (\K_{(N,n)} + r(\xu_k) \vecu' \vecu^\top)^{-1} = \K_{(N,n)}^{-1} + \frac{\left(\!\K_{(N,n)}^{-1}\!\right)_{.,k} \left(\!\K_{(N,n)}^{-1}\!\right)_{k,.} }{\left(r(\xu_k) u_k^2 \right)^{-1} \!- \left(\!\K^{-1}_{(N,n)}\!\right)_{k,k}} = \K_{(N,n)}^{-1} + \B_k.
\label{eq:Kni_up}
\end{equation}
This enables us to write $\mu_{(N+1,n)}(\x) - \mu_{(N,n)}(\x) =
\veckn(\x)^\top\left( \K_{(N+1,n)}^{-1} \vecYu_{n+1} - \K_{(N,n)}^{-1} \vecYu_{(N,n)} \right)$
and $\sigma^2_{(N+1,n)} (\x) - \sigma^2_{(N,n)}(\x) = \veckn(\x)^\top\left(
\K_{(N+1,n)}^{-1} - \K_{(N,n)}^{-1} \right) \veckn (\x)$ and substitute $\K_{(N+1,n)}^{-1} - \K_{(N,n)}^{-1} = \B_k$ from \eqref{eq:Kni_up}. From the proof we also see that adding a replicate $\x_{N+1}$ incurs $\Or(n)$ rather than the usual $\Or(n^2)$
cost.
\end{proof}

As a corollary we obtain the following formula for one-step-ahead IMSPE at existing designs $I_{N+1}(\xu_k)$
(relying on the fact that $\Wn$ is unchanged when replicating):
\begin{equation}
I_{N+1}(\xu_k) = E - \tr(\K^{-1}_{(N+1,n)} \Wn) = E -
\tr( (\K^{-1}_{(N,n)} + \B_k) \Wn) = I_N - \tr(\B_k \Wn). \label{eq:Ir}
\end{equation}
Besides enabling a ``quick check'' (with cost $\Or(n^2)$) for finding the best replicate, perhaps a
more important application of this result is that \eqref{eq:Ir} yields an explicit condition
under which replication is optimal.

\begin{Proposition}
Given $n$ unique design locations $\xu_1, \ldots, \xu_n$, replicating is optimal (with respect to $I_{N+1}$)
if
\begin{equation}
r(\xnew) \geq \frac{\veck(\xnew)^\top \K_n^{-1} \Wn \K_n^{-1}
\veck(\xnew) - 2 \vecw(\xnew)^\top \K_n^{-1} \veck(\xnew) + w(\xnew,
\xnew)}{\tr(\B_{k^*} \Wn)} - \cs_n^2(\xnew), \quad \forall \xnew \in D, \label{eq:r}
\end{equation}
where $k^* \in \argmin_{1 \leq k \leq n} I_{N+1}(\xu_k)$.
\label{prop:rep}
\end{Proposition}
\begin{proof}
We proceed by comparing $I_{N+1}(\xnew)$ values when $\xnew$ is a
replicate vis-\`a-vis a new design.  Summarizing our results from above, we have $I_{N+1}(\xu_k^*) = I_N - \tr(\B_{k^*} \Wn)$ for the (best) replicate and $I_{N+1}(\xnew)
= I_N - \left( \sigma_n^2(\xnew) \vecg(\xnew)^\top \Wn \vecg(\xnew) + 2
\vecw(\xnew)^\top \vecg(\xnew)
+ \sigma_n^2(\xnew)^{-1} w(\xnew, \xnew) \right)$ for a new design.  Replicating is better if
$I_{N+1}(\xu_k^*) \leq I_{N+1}(\xnew)$ for all $\xnew$, or when
\[
\tr(\B_{k^*} \Wn) \geq
\sigma_n^2(\xnew)^{-1} \left(\veck(\xnew)^\top \K_n^{-1} \Wn \K_n^{-1} \veck(\xnew)
   - 2 \vecw(\xnew)^\top \K_n^{-1} \veck(\xnew) + w(\xnew, \xnew) \right).
\]
Using the fact that $\sigma_n^2(\xnew) = \cs_n^2(\xnew) + r(\xnew)$ establishes the desired result.
\end{proof}
Referring back to Figure \ref{fig:rep}, the gray-dotted line in the {\em left}
panel represents the right hand side of Eq.~(\ref{eq:r}). Thus, any noise
surfaces with $r(\x)$ above this line will lead to the $I_{N+1}$ minimizer
being a replicate, cf.~the solid blue $r(\x)$ case in the figure.  Although this
illustration involves a heteroskedastic example, the inequality in
\eqref{eq:r} can also hold in the homoskedastic case.  In practice,
replication in homoskedastic processes is most often at the edges of the input
space, however particular behavior is highly sensitive to the settings of the
$n$ design locations, and their degrees of replication, $a_i$.

\subsection{Gradient expressions}
\label{ss:grad}

To facilitate the optimization of $I_{N+1}(\xnew)$ with respect to $\xnew$, we
provide closed-form expressions for its gradient, via partial derivatives.
Below the subscript $(p)$ denotes the $p$-th coordinate of the $d$-dimensional
design $\xnew \in D$. As a starting point, the chain rule gives
\begin{equation}
\dfrac{ \partial I_{N+1}(\xnew)}{\partial \xnewp}
= -\dfrac{ \partial \tr(\K_{n+1}^{-1} \W_{n+1})}{\partial \xnewp}
=- \tr \left(\K^{-1}_{n+1}\dfrac{ \partial \W_{n+1}}{\partial \xnewp}\right)
- \tr\left(\dfrac{ \partial \K^{-1}_{n+1}}{\partial \xnewp} \W_{n+1} \right)
\label{eq:dIfirst} .
\end{equation}
To manage the computational costs, we notate below how the partial derivatives are distributed
in another application of the partition inverse equations:

\noindent\begin{minipage}{7cm}
\begin{equation}
\dfrac{ \partial \K^{-1}_{n+1}}{\partial \xnew}
 = \begin{bmatrix}
 \Hmat(\xnew) & \vech(\xnew) \\
\vech(\xnew)^\top & v_1(\xnew)
 \end{bmatrix} \label{eq:dKi}
 \end{equation}
 \end{minipage}%
 \hfill
 \begin{minipage}{7cm}
\begin{equation}
\dfrac{ \partial \W_{n+1}}{\partial \xnewp}= \begin{bmatrix}
\mathbf{0}_{n \times n} & \vecc_1(\xnew) \\
\vecc_1(\xnew)^\top & c_2(\xnew)
 \end{bmatrix} \label{eq:dW}
 \end{equation}
 \end{minipage}\vskip1em
 \noindent where the detailed expressions and derivations are given in Appendix \ref{ap:grad}.

The expressions above are collected into the following lemma.

\begin{lemma}
The $p^\mathrm{th}$ component of the gradient for sequential ISMPE is 
\begin{align}
-\frac{\partial I_{N+1}}{ \partial \xnewp}
 & = 2 \vecc_1(\xnew)^\top \vecg(\xnew) + \vecc_2 \sigma_n^2(\xnew)^{-1}
 	+ \ones^\top_n [\Hmat(\xnew) \sigma_n^2(\xnew) \circ \W_{n}] \ones_n  \label{eq:dI} \\
 & \quad + 2 \vecw(\xnew)^\top \vech(\xnew)
 	+ v_1(\xnew) w(\xnew, \xnew). \nonumber 
\end{align}

\label{lem:grad_IMSPE_closed}
\end{lemma}
\begin{proof}
Beginning with Eq.~(\ref{eq:dIfirst}), substitute (\ref{eq:dKi})
for the partial derivative of $\K_{n+1}^{-1}$, and (\ref{eq:dW}) for that of $\W_{n+1}$.
Then, note that
$\ones^\top_n [\Hmat(\xnew) \sigma_n^2(\xnew) \circ \W_{n}] \ones_n$ can be
rewritten as $v_2(\xnew) \vecg(\xnew)^\top \Wn \vecg(\xnew)$ $+ 2 \sigma_n^2(\xnew)
\vecg(\xnew)^\top \Wn \vech(\xnew)$.
\end{proof}

\noindent Since no further matrix decompositions are required, note that
calculating the gradient of $I_{N+1}(\xnew)$ in this way incurs computational
costs in $\Or(n^2)$.

\subsection{Looking ahead over replication}
\label{sec:ahead}

Under certain conditions, sequential design via IMSPE, i.e., greedily
minimizing $I_{N+1}$ to choose $\x_{N+1}$, can well-approximate a one-shot
batch design of size $N_{\mathrm{max}}$ because the criterion is monotone
supermodular \citep{Das2008,Krause2008}.   However, these results assume a
known kernel hyperparameterization $k(\cdot, \cdot)$ and constant noise level
$r(\cdot)$. In the more realistic case where those quantities must be
estimated from data, and potentially with non-constant variance, there is
ample evidence in the literature suggesting that sequential design can be {\em
much better} than a batch design, e.g., based on a poorly-chosen
parameterization, and no worse than an idealistic one
\citep{seo00,gra:lee:2009}. However, that does not mean that greedy, myopic,
selection is optimal.  By accounting for potential future selections in
choosing the very next one, it is possible to obtain substantially improved
final designs.  However, the calculations involved, especially to ``look
ahead'' from early sequential decisions to a far-away horizon
$N_{\mathrm{max}}$, require expensive dynamic programming techniques to search
an enormous decision space.

Approximating that full search, by limiting the lookahead horizon or
otherwise reducing the scope of the decision space, has become an active area
in Bayesian optimization via expected improvement
\citep{Ginsbourger2010,Gonzalez2016,Lam2016}. Targeting overall accuracy has
seen rather less development, the work by \cite{Huan2016} being an important
exception.  Here we aim to port many of these ideas to our setting of IMSPE
optimization, where the nature of our approximation involves a weak bias
towards replication which we have shown can be doubly beneficial in design.

The essential decision boils down to either choosing an $\x_{N+1}$ to explore,
i.e., a new design element $\xu_{n+1}$, or choosing to replicate with
$\x_{N+1}$ taken to be some $\xu_k$, for $k
\in \{1,\ldots,n\}$.  However, rather than directly minimizing (\ref{eq:I}) or
(\ref{eq:Ir}), respectively, we perform a ``rollout'' lookahead procedure
similar to \cite{Lam2016} in order to explore the impact of those choices on a
limited space of future design decisions.  The updating equations in the
previous subsections make this tractable.

In particular we consider a horizon $h \in \{0, 1, 2, \ldots \}$ determining
the number of design iterations to look ahead, with $h=0$ representing
ordinary (myopic) IMSPE search.  Although larger values of $h$ entertain
future sequential design decisions, the goal (for any $h$) is to determine
what to do \emph{now}. Toward that end, we evaluate $h+1$ ``decision paths''
spanning alternatives between exploring sooner and replicating later, or vice
versa. During each iteration along a given path, either (\ref{eq:I}) or
(\ref{eq:Ir}) (but not simultaneously) is taken up as the hypothetical action.
On the first iteration, if a new $\xu_{n+1}$ is chosen by optimizing
Eq.~\eqref{eq:I}, that location (along with the existing $\xu_1, \ldots,
\xu_n$) are considered as candidates for future replication over the remaining
$h$ lookahead iterations (when $h \geq 1$). If instead a replicate is chosen
in the first iteration, the lookahead recursively searches over the choice of
which of the remaining $h$ iterations will pick a new $\xu_{n+1}$, with the
others optimizing over replicates. This recursion is resolved by moving to the
second iteration and again splitting the decision path into the choice between
replicate-explore-replicate-... and replicate-replicate-..., etc. After
recursively optimizing up to horizon $h$ along the $h+1$ paths, the ultimate
IMSPE for the respective hypothetical design with size $N + 1 +h$ is computed,
and the decision path yielding the smallest IMSPE is noted. Finally, the next
$\xu_{N+1}$ is a new location if the explore-first path was optimal, and is a
replicate otherwise.

\begin{figure}[ht!]
  \centering
  \def\svgwidth{0.75\textwidth}%
  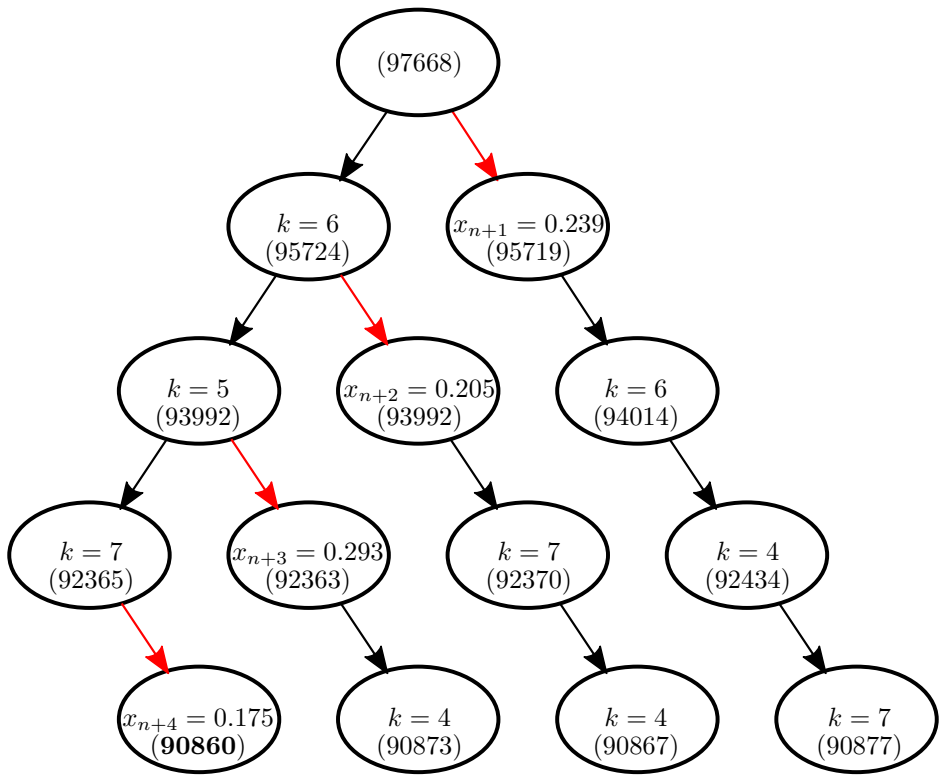%
  \caption{Lookahead strategy for $h = 3$, starting with $n$ unique designs.
  Each 
  ellipse is a state, with a specific training set. Black dashed arrows
  represent the action of adding the best replicate, red solid arrows represent
  adding the best new design. This example considers augmenting the design for
  the example shown in Figure \ref{fig:alloc}. Numbers in parenthesis
  indicates the IMSPE at each stage (values have been multiplied by $10^8$).}
  \label{fig:rollout}
\end{figure}

A diagram depicting the decision space is shown in Figure \ref{fig:rollout}.
In this example we attempt to augment the design from Figure~\ref{fig:alloc}
using $h=3$. The path yielding the lowest IMSPE at the
horizon involves here replicating three times (adding copies of design
elements 6, 5, and 7 respectively), with exploration at $x_{n+4} = 0.175$ in
the final stage. Consequently, replication is
preferred over exploration and the next design element will be a replicate,
duplicating $\xu_6$.  The figure also illustrates that the cost of searching
for the best replicate over adding a new design element involves at most $h+1$
global optimizations of Eq.~(\ref{eq:I}), using the gradient.  Although
$(h+1)(h+2)/2-1$ discrete searches over (\ref{eq:Ir}) are required, the
diagram indicates that $(h+1)$ searches of mixed continuous and discrete type
may be performed in parallel. In practice, global optimization
with a budget of at least the same order as $n$ is an order of magnitude more
expensive than looking for the best replicate.

In this scheme the horizon, $h$, determines the extent to which
replicates are entertained in the lookahead, and therefore larger $h$ somewhat
inflates the likelihood of replication. Indeed, as $h$ grows, there are more and
more decision paths that delay exploration to a later iteration; if any of
them yield a smaller IMSPE than the explore-first path, the
immediate action is to replicate. However, note that although larger $h$
allows more replication before committing to a new, unique $\xu_{n+1}$,  it
also magnifies the value of an $\xu_{n+1}$ chosen in the first iteration, as
it could potentially accrue its own replicates in subsequent rollout
iterations.  Therefore, although we do find in practice that larger $h$ leads
to more replication in the final design, this association is weak. Indeed, we
frequently encounter situations where exploration is (temporarily) preferred
for arbitrarily large horizons.

\section{Modeling, inference and implementation}
\label{sec:implement}

Here we consider inference and implementation details, in particular for
learning the noise process $r(\cdot)$.  Our presumption is that little is
known about the noise, however it is worth noting that this assumption may not
be well aligned to some data-generating mechanisms, e.g., as arising from
Monte Carlo simulations with known convergence rates \citep{Picheny:2013}.
After reviewing a promising new framework called {\tt hetGP}, for
heteroskedastic GP surrogate modeling \citep{Binois2016}, we provide
extensions facilitating fast sequential updating of that predictor along with
its (hyper-) parameterization. We conclude with schemes for adjusting the
lookahead horizon introduced in Section \ref{sec:ahead}.

\subsection{Heteroskedastic modeling}
\label{sec:hetGP}

One way of dealing with heteroskedasticity in GP regression
is to use empirical estimates of the variance as in SK \citep{Ankenman2010},
described briefly in Section \ref{sec:concept}.  Although this has the
downside of requiring a minimum amount of replication, the calculations are
straightforward and computations are thrifty. However, sequential design
requires predicting the variance at new locations, and to accommodate that
within SK \citeauthor{Ankenman2010}, recommend fitting a second,
independent, GP for $\hat{r}(\x)$ to smooth the empirical variances.

An alternative is to model the (log) variance as a latent process, jointly
with the original ``mean process''
\citep{goldberg:williams:bishop:1998,kersting:etal:2007}.  However these
methods can be computationally cumbersome, and are not tailored to leverage
the computational savings that come with replication.  Here we rely on the
hybrid approach detailed by \cite{Binois2016}, leveraging replication and
learning the latent log-variance GP based on a joint log-likelihood with the
mean GP. We offer the following by way of a brief review.

For common choices of stationary kernel $k(\x, \x') = \nu c(\x -
\x')$, the covariance matrix for the ``mean GP'' may be characterized as $\Kn =
\nu(\Cn + \Lan)$ with $\Cn = \left(c(\xu_i -
\xu_j) \right)_{1 \leq i,j \leq n}$; and for the ``noise GP'' we take the
analog $\log \Lan = \Cg (\Cg + g
\An^{-1})^{-1} \Deltan$ where $\Cg$ is the equivalent of $\Cn$ for the second
GP with kernel $k_{(g)}$. That is, $\log \Lan$ is the prediction given by a GP
based on latent variables $\Deltan = (\delta_1, \ldots, \delta_n)$ that can be
learned as additional parameters, alongside hyperparameters
of $k_{(g)}$ and nugget $g$.

Based on this representation, the MLE of $\nu$ is $$\nuN := N^{-1}
\left(N^{-1} \sum_{i=1}^n \frac{a_i}{\lambda_i} s_i^2 + \vecYu^\top (\Cn +
\An^{-1}\Lan)^{-1} \vecYu \right)$$ with $s_i^2 = \frac{1}{a_i}
\sum_{j=1}^{a_i} (y_i^{(j)} - \bar{y}_i)^2$ whereas the rest of the
parameters and hyperparameters can be optimized based on the concentrated
joint log-likelihood:
\begin{align*}
\log \tilde L =~&   - \frac{N}{2} \log \nuN  \nonumber
- \frac{1}{2} \sum\limits_{i=1}^n \left[(a_i - 1)\log \lambda_i + \log a_i \right] -
\frac{1}{2} \log |\Cn + \An^{-1}\Lan|  \\
& - \frac{n}{2} \log \hnug  - \frac{1}{2} \log |\Cg + g \An^{-1}| + \mbox{Const}, \label{eq:jllik}
\end{align*}
with $\hnug = n^{-1} \Deltan^\top (\Cn + g \An^{-1})^{-1} \Deltan$. Closed
form derivatives are given in \cite{Binois2016}, while an {\sf R}
\citep{cran:R}  package with embedded {\sf C++} subroutines is available as
\texttt{hetGP} on CRAN.

Notice that for stationary kernels, the Eq.~(\ref{eq:Ifull}) reduces to
$\mathrm{IMSPE}(\x_1, \ldots, \x_N) = \nu( 1 - \mathrm{tr}(\Cn^{-1} \Wn))$.
The look-ahead IMSPE over replicates (\ref{eq:Ir}) becomes
$I_{N+1}(\xu_k) = \nu( 1 - \mathrm{tr}(\B_k' \Wn))$ with $\B_k' =
\frac{\left(\left(\Cn + \An^{-1}\Lan\right)^{-1}\right)_{.,k} \left(\left(\Cn + \An^{-1}\Lan\right)^{-1}\right)_{k,.} }{a_k (a_k + 1)/\lambda_k -
\left(\Cn + \An^{-1}\Lan \right)^{-1}_{k,k}}$. Also, the gradient of $I_{N+1}(\xnew)$ from (\ref{eq:I})
involves $\partial r(\xnew)/ \partial \xnewp$, which for \texttt{hetGP}
reduces to
\[
\frac{\partial \veckg(\xnew) (\Cg + g \An^{-1})^{-1} \Deltan}{ \partial
\xnewp} = \frac{\partial \veckg(\xnew) }{\partial \xnewp} (\Cg + g
\An^{-1})^{-1} \Deltan.
\]

\subsection{Sequential heteroskedastic modeling}\label{sec:seq-hetGP}

Optimizing IMSPE with lookahead over replication [Section \ref{sec:ahead}] is
only practical if the \texttt{hetGP} model can be updated efficiently when new
simulations are performed. Two different update schemes are necessary: one for
potential new designs, considered during the process of evaluating
alternatives under the criteria [Eqs.~(\ref{eq:I}--\ref{eq:Ir})]; and another
for the actual update with new simulation $y(\x_{N+1})$.

When looking-ahead, no new $y$-value is entertained, so hyperparameters of both
GPs stay fixed and only the latents may need to be augmented. Updating $\K_n$
follows (\ref{eq:Knp1i_up}) or (\ref{eq:Kni_up}), depending on whether the
candidate $\xnew$ is new or a replicate. In the latter case, only $\An$ is
updated for the ``noise GP''. Conversely, if a new location is added, an estimate of
$r(\xu_{n+1})$ is required, which can come from the noise GP via exponentiating
the usual GP predictive equations. That is, the new latent $\delta_{n+1}$ is
taken as the predicted value by the noise GP.

The second update scheme---using the $y(\x_{N+1})$ observation---will require
updating all the GPs' hyperparameters (including latents). Optimizing all
hyperparameters of our heteroskedastic GP model is a potentially costly
$\Or(n^3)$ procedure. Instead of starting from scratch, a warm start of the
MLE optimization is performed. Where they exist, previous values can be
re-used as starting values, leaving only the latent $\tilde{\delta}$ at the
newest design point, that is $\tilde{\delta} = \delta_{n+1}$ for a new location
or $\tilde{\delta} = \delta_k$ for a replicate,
requiring special attention.

As in the first case, $\tilde{\delta}$ may be
initialized at its predicted value. But taking into account the new
$y(\x_{N+1})$ makes it possible to combine information from the latent noise GP
with results from empirical estimation of the log-variances.
\cite{Kaminski2015} explores this for updating SK models when new observations
are added---a special case of the typical GP update formulas.
The resulting combination of two predictions is via the geometric mean and can be summarized by the Gaussian
$\mathcal{N}(\tilde{\delta}, V_{\tilde{\delta}})$ with
$$ \tilde{\delta} = \left(\frac{\mu_{(g)} (\x_{N+1})}{\cs^2_{(g)}(\x_{N+1})} + \frac{\hat{\delta}}{ V_{\hat{\delta}} } \right) \left( \frac{1}{\cs^2_{(g)}(\x_{N+1})} + \frac{1}{V_{\hat{\delta}}} \right)^{-1}, \quad
V_{\tilde{\delta}} = \left(\frac{1}{\cs^2_{(g)}(\x_{N+1})} + \frac{1}{V_{\hat{\delta}}} \right)^{-1},
$$
where $\mu_{(g)} (\x_{N+1})$ and $\cs^2_{(g)}(\x_{N+1})$ are the prediction
from the noise GP\footnote{To avoid predictive variances close to zero for
replicates, i.e., $\tilde \delta = \delta_k$, such that
$\sigma^2_{(g)}(\x_{N+1}) \approx 0$ ($g$ should be small), the variance is
given by the ``downdated'' GP instead (i.e., the predicted variance if
removing the replicated design), that are usually used for Leave-One-Out
estimations and can be found, e.g., in \cite{Bachoc2013}, giving
$\sigma^2_{(g)}(\x_{N+1}) =
\left( \left(\Cg + g \An^{-1} \right)^{-1}_{k,k} \right)^{-1}$. 
} while $\hat{\delta}$ is the empirical estimate
of the log variance at $\x_{N+1}$, itself with variance $V_{\hat{\delta}}$.
We take
$\hat{\sigma}^2 = \nuN^{-1} \frac{1}{\tilde{a}} \sum_{j=1}^{\tilde{a}}
\left(y^{(j)}(\x_{N+1}) -
\mu_n(\x_{N+1})\right)^2$, i.e., the uncorrected sample variance estimator
that exists even for $\tilde{a} = 1$, i.e., the number of observations at
$\x_{N+1}$. Supposing that the $y^{(j)}(\x_{N+1})$'s are i.i.d.\ Gaussian, we have $\tilde{a}
\hat{\sigma}^2/\sigma^2 \sim \chi^2_{\tilde{a}}$. Accounting for the log-transformation, as in \cite{Boukouvalas2010}, we get $\hat{\delta} =
\log(\hat{\sigma}^2) - \Psi((\tilde{a})/2) - \log(2) + \log(\tilde{a})$ and
$V_{\hat{\delta}} = \Psi_2(\tilde{a}/2)$ with $\Psi$ and $\Psi_2$ the digamma
and trigamma functions.

Finally, the quick updates described above are predicated on improving local
searches, and are thus not guaranteed to globally optimize the likelihood,
which is always a challenge in MLE settings.  The risk of becoming trapped in
an inferior local mode is greater at earlier stages in the sequential design,
i.e., when $n$ is small. In practice, we find it beneficial to periodically
restart the optimization with conservative (potentially random)
initializations, which is cheap in that (small $n$) setting.  As $n$ increases,
and the likelihood becomes more peaked, we find that costly restarts are
of limited practical value.  Local refinements, as described above, are fast
and reliable.

\subsection{Defining the horizon}
\label{sec:horiz}

Although the horizon $h$ in the lookahead criteria in Section \ref{sec:ahead}
could be fixed arbitrarily, or chosen based on computational considerations
(smaller being faster), here we propose two heuristics to set it adaptively
based on design goals. The adaptiveness means that $h \equiv h_N$ is now
indexed by the current design size.

 The first heuristic involves managing surrogate modeling costs, targeting a
fixed ratio $\rho = n/N$ of unique to full design size.  The goal is to ensure
that each new unique location is, ``worth its weight'' in replicates
from a computational perspective. The choice of $n/N$ is arbitrary
---other targets will do---but we focus on this particular one because its
magnitude is easy to intuit.  The \emph{Target} heuristic we use to ``maintain $\rho$'' as
sequential design steps progress is as simple as it is effective:
\begin{equation}
h_{N+1} \leftarrow \left\{
\begin{array}{lll}
h_N + 1 & \mbox{if } n/N > \rho & \mbox{and a new point $\xu_{n+1}$ is chosen}; \\
\max\{h_N-1,-1\} & \mbox{if } n/N < \rho & \mbox{and a replicate is chosen}; \\
h_N & \mbox{otherwise}.
\end{array}
\right. \label{eq:htarget}
\end{equation}
If the current ratio is too high and a new point $\xu_{n+1}$ was recently
added, making the ratio even higher, the horizon is increased to encourage
future replication. If rather a replicate has been added while the current
ratio was too low, then the horizon is decreased, encouraging exploration.
Otherwise the evolution is on the right trajectory and the horizon is
unchanged.  Observe that (\ref{eq:htarget}) allows a horizon of $-1$, which
is described shortly.

To implement the continuous search (\ref{eq:I}), we
deploy a limited multistart scheme over {\tt optim} searches in {\sf R} with
\verb!method="lbfgsb"! and closed form gradients (\ref{eq:dI}).  In parallel,
a discrete search over $\xu_1,\ldots, \xu_n$ is carried out via (\ref{eq:Ir}).
The two solutions thus obtained are compared against thresholds, and if the
$\xnew$ found via continuous search (or its relative objective value) is
within (say) $\varepsilon = 10^{-6}$ of that of the best $\xu_{k^*}$, the
replicate is preferred on computational grounds. The horizon $h_{N}=-1$ is an
exception, adding a new $\xnew$ no matter how close it is to the replicate
candidate. Thus, $h \equiv -1$ can be roughly thought of as incrementing $n$
by 1 along with $N$ at each iteration; in practice it still occasionally
generates replicates, primarily at the corners of the input space, if the
corresponding multistart scheme determines that the $I_{N+1}$-minimizer lies
at the boundary of $D$. On the other hand,  $h \equiv 0$ obtains many
replicates due to thresholding, which yields a ``soft'' clustering mechanism
for $(\xu_1, \ldots, \xu_n)$. Indeed, every iteration where we have a
situation resembling Figure \ref{fig:rep},
the $h = 0$ rule will select a replicate and not increment $n$. In
contrast, $h = -1$ will only ``stumble'' into a replicate if the optimizer
finds a global minimum at the edge of the domain.  It does not explicitly
entertain replicates via \eqref{eq:Ir}.

Our empirical work [Section \ref{sec:experiment}] illustrates how horizon
targeting effectively manages computational costs. Although at times the
horizon $h_N$ can reach quite high values (upwards of $h_N=20$), the
computational cost of search is negligible compared to updating the GP fits.
Meanwhile high horizons represent a ``light touch'' preference for
replication: they do not preclude exploration, rather they somewhat discourage
it. Thus, while the ultimate number of unique locations $n$ is dependent on
the entire history of the simulations, and hence comes with a sampling
distribution, the corresponding search heuristic is much simpler than one that
would impose a hard constraint on the final $n$.

When accuracy is the ultimate goal we prefer a different adaptation of $h$,
making a more explicit link between $\rho$ and the signal-to-noise ratio in
the data. In {\em linear} regression contexts, one way to deal with
heterogeneity is to allocate replications on unique designs such that the
ratio of the empirical variance over number of replicates are close to each
other, i.e., to enforce homogeneity of $\hat{\sigma}_i^2/a_i$
\citep{Kleijnen2015}. This approach captures the basic idea that more
replicates are needed where $r(\x)$ is high, but applicability to
our setup is not direct because such a scheme does not factor in correlations
estimated by GPs. \citet{Ankenman2010} address this within SK by
considering the allocation of the remaining budget of evaluations over
existing designs, i.e., to determine where to augment with additional
replicates. In particular, they show that the optimal allocation of the $N$
simulations across $n$ unique designs is summarized by
$\An^*$, a diagonal matrix with components
\begin{align}\label{eq:sk-imspe}
a_i^* \approx N
\frac{\sqrt{r(\xu_i)K_i}}{\sum\limits_{j=1}^n
\sqrt{r(\xu_j)K_j}}, \quad \mbox{where} \quad K_i = (\K_n^{-1} \Wn \K_n^{-1})_{i,i}.
\end{align}
We emphasize that \eqref{eq:sk-imspe} only addresses the replication
aspect---the designs $\xu_1, \dots, \xu_n$ must be entered {\em a
priori} by the user. Thus, this recipe is not directly implementable in a
sequential design setting. One solution could be to generate (e.g., by
space-filling) a candidate design of pre-determined size $n$ and
$\underbar{r}$ replicates per design and then, after learning $\K_n$ and
$r(\xu_i)$'s, apply \eqref{eq:sk-imspe}. However, in that case one may end up
with $a_i^* < \underbar{r}$, as is the case in Figure \ref{fig:alloc}. This
illustrative 1d example highlights that in areas with low noise, a lower
number of replicates would have been better, while in more noisy areas, more
points are necessary. The right panel shows the $a_i^*$ at this stage
(referred to as \emph{batch}) compared to the greedy sequential allocation of
105 replicates. The latter is more realistic because it acknowledges that
design decisions cannot be undone\footnote{The batch scheme recommends fewer
than five replicates after five replicates where already used.}.

\begin{figure}[ht!]
  \centering
  \vspace{-0.75cm}
    \begin{subfigure}[t]{0.4\textwidth}%
	\centering%
	\includegraphics[width=\textwidth, trim = 0 5 0 10, clip = TRUE]{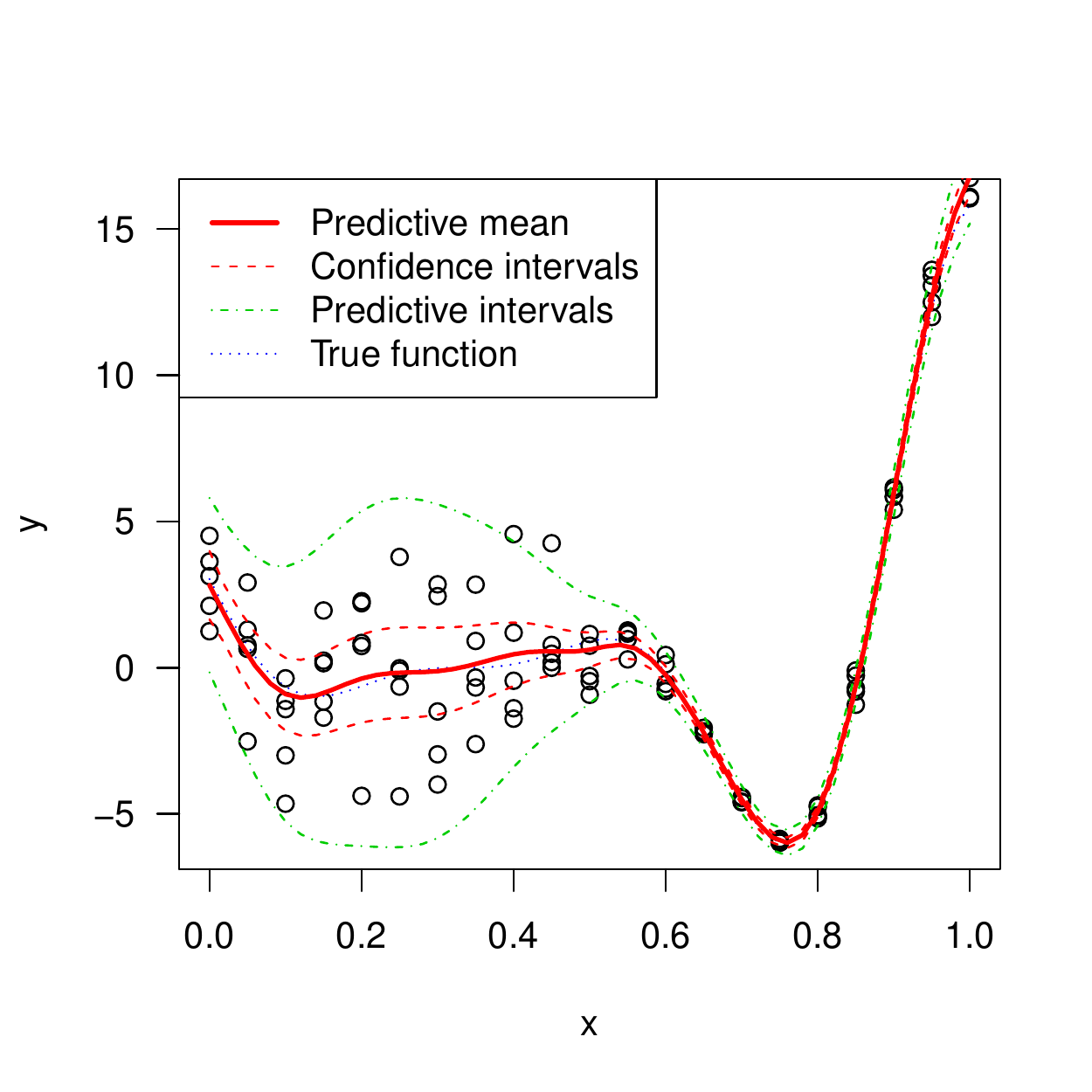}%
	\end{subfigure}%
	  \begin{subfigure}[t]{0.4\textwidth}%
	\centering%
	\includegraphics[width=\textwidth, trim = 0 5 0 10, clip = TRUE]{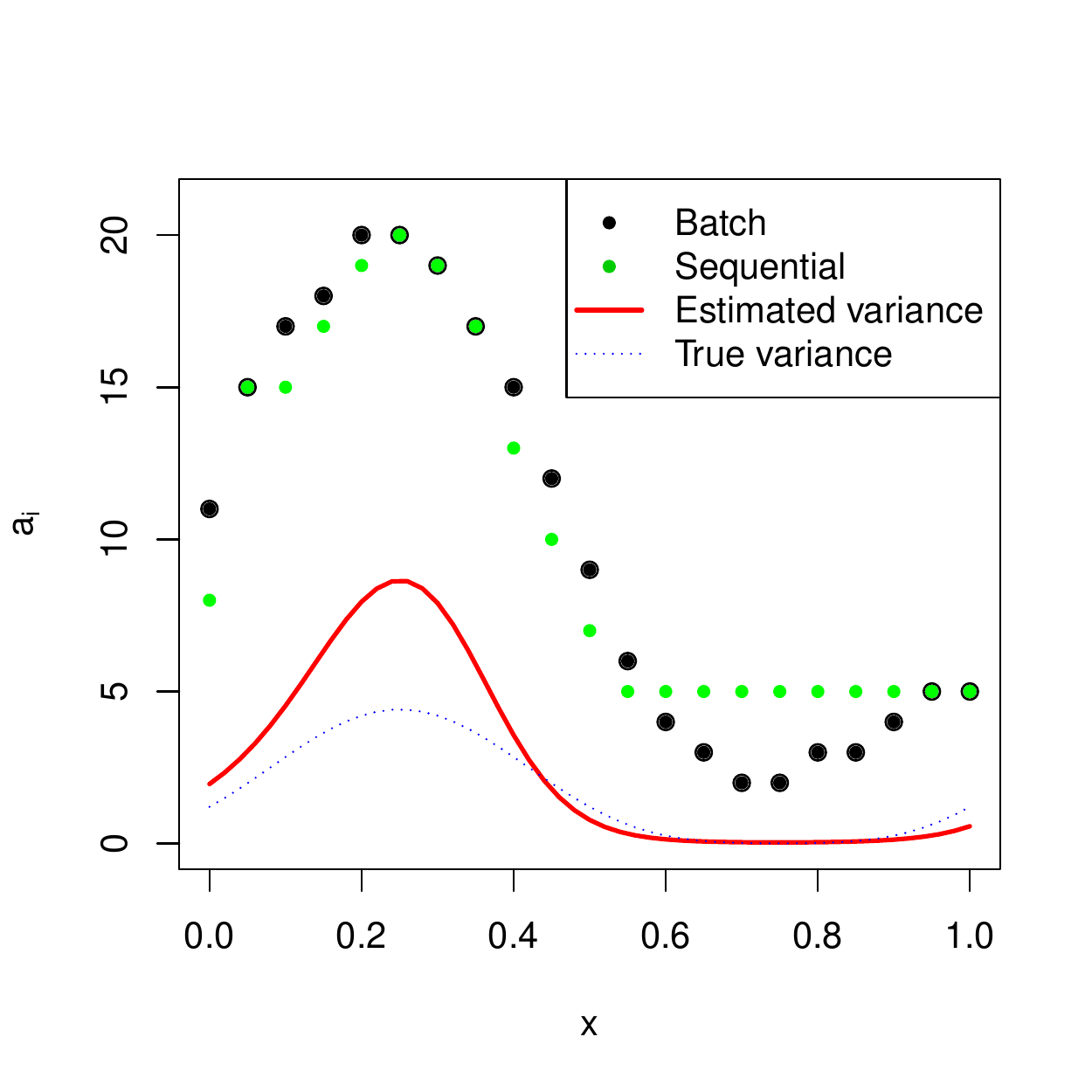}%
	\end{subfigure}%
  \vspace{-0.5cm}
  \caption{Left: toy example with 5 replicates at each of the 21 uniformly
  spaced unique design points. Right: proposed allocation of new 105
  replicates (total 210 observations) based on \eqref{eq:sk-imspe} and a
  greedy sequential approach.}%
  \label{fig:alloc}%
\end{figure}

Instead of such two-stage design, we utilize \eqref{eq:sk-imspe} in a
sequential fashion, by making a comparison between the allocation $a_i^*$ via
\eqref{eq:sk-imspe} (employing the current estimates of the noise $r(\xu)$ at
that particular stage) and the actual $a_i$'s collected so far from the
sequential design. The existing number of replicates $a_i$ is then either too
high, in which case no more replicates should be added, or too low, and could
benefit from more replication. We use this information in the \emph{Adapt}
scheme to adjust the horizon by sampling
\begin{equation}
h_{N+1} \sim \mathrm{Unif} \{ a'_1, \dots, a'_n \} \quad \text{ with }\quad a'_i := \max(0, a_i^* - a_i).
\label{eq:hadapt}
\end{equation}
Hence, if there are locations that require many more replicates according to
\eqref{eq:sk-imspe}, $h_{N+1}$ could be large to encourage replication.

\section{Experiments}
\label{sec:experiment}

Here we illustrate our methods and simpler variants on a suite of examples
spanning synthetic and real data from computer simulation experiments. Our
main metric is out-of-sample root mean-square (prediction) error (RMSE) over
the sequential design iterations, and in particular after the final iteration.
Since accurate estimation of variances over the input space is also an
important consideration (especially in the heteroskedastic context)---even
though our IMSPE criteria does not explicitly target learning variances---we
consider RMSE to the true $\log$ variance, when it is known, and when it is
not we use a proper scoring rule \citep[][Eq.~(27)]{gneiting:raftery:2007}
combining mean and variance forecasts out-of-sample.  Our main comparators are
non-sequential (space-filling) designs, homoskedastic GP predictors, and
combinations thereof.

\subsection{Illustrative one-dimensional example}

We start by reusing the 1d toy example from above [surrounding Figure
\ref{fig:rollout} and \ref{fig:alloc}] to show qualitatively the effect of the
horizon choice on the resulting designs. The underlying function is $f(x) =
(6x -2 )^2 \sin(12x-4)$, from \citet{Forrester2008}, and the noise function is
$r(x) = (1.1 + \sin(2 \pi x))^2$. The experiment starts with an initial maximin
LHS with 10 points, no replicates, and the GPs use a Gaussian kernel.

\begin{figure}[ht!]
  \centering
	\includegraphics[scale=0.38,trim=0 30 44 0,clip=TRUE]{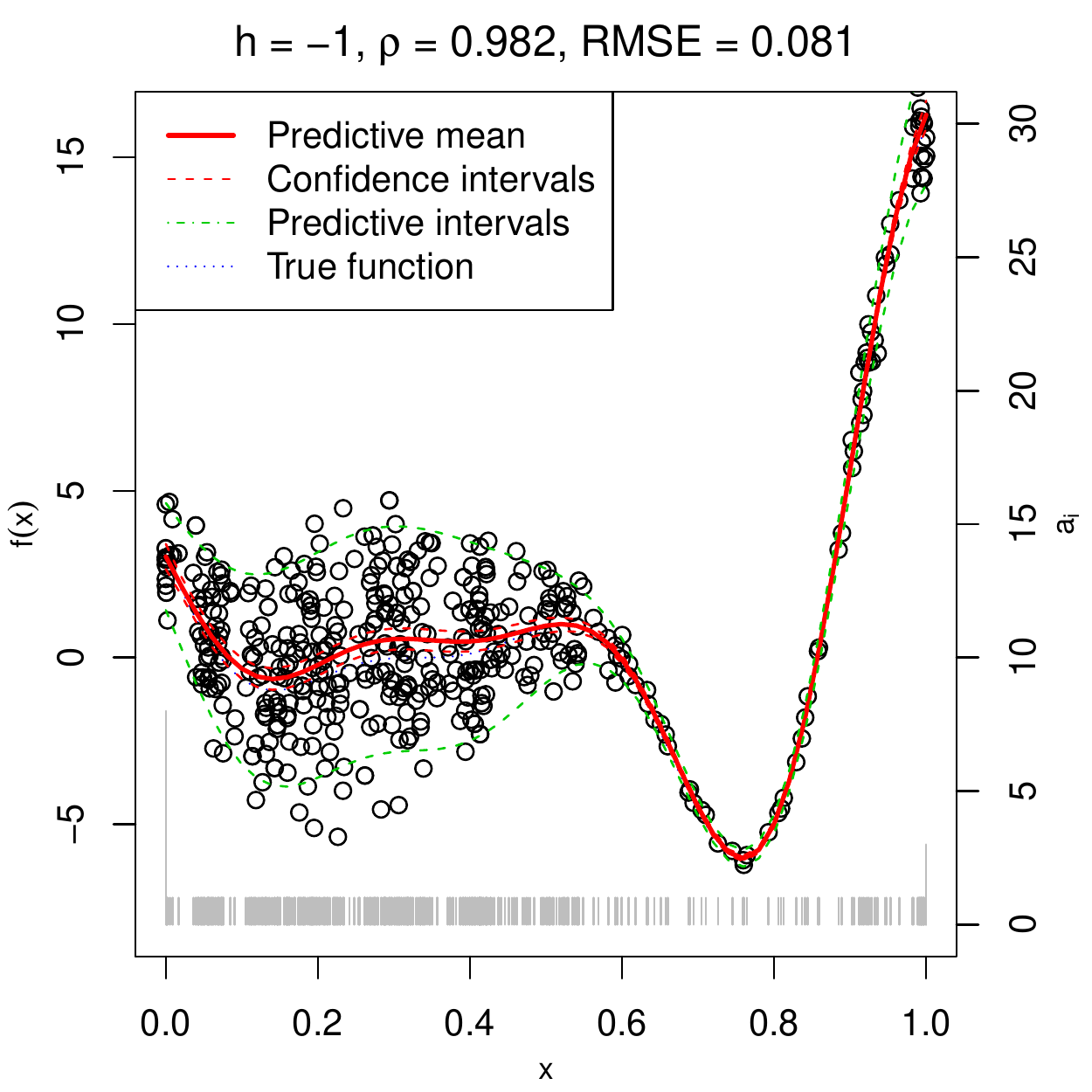}
	\includegraphics[scale=0.38,trim=44 30 44 0,clip=TRUE]{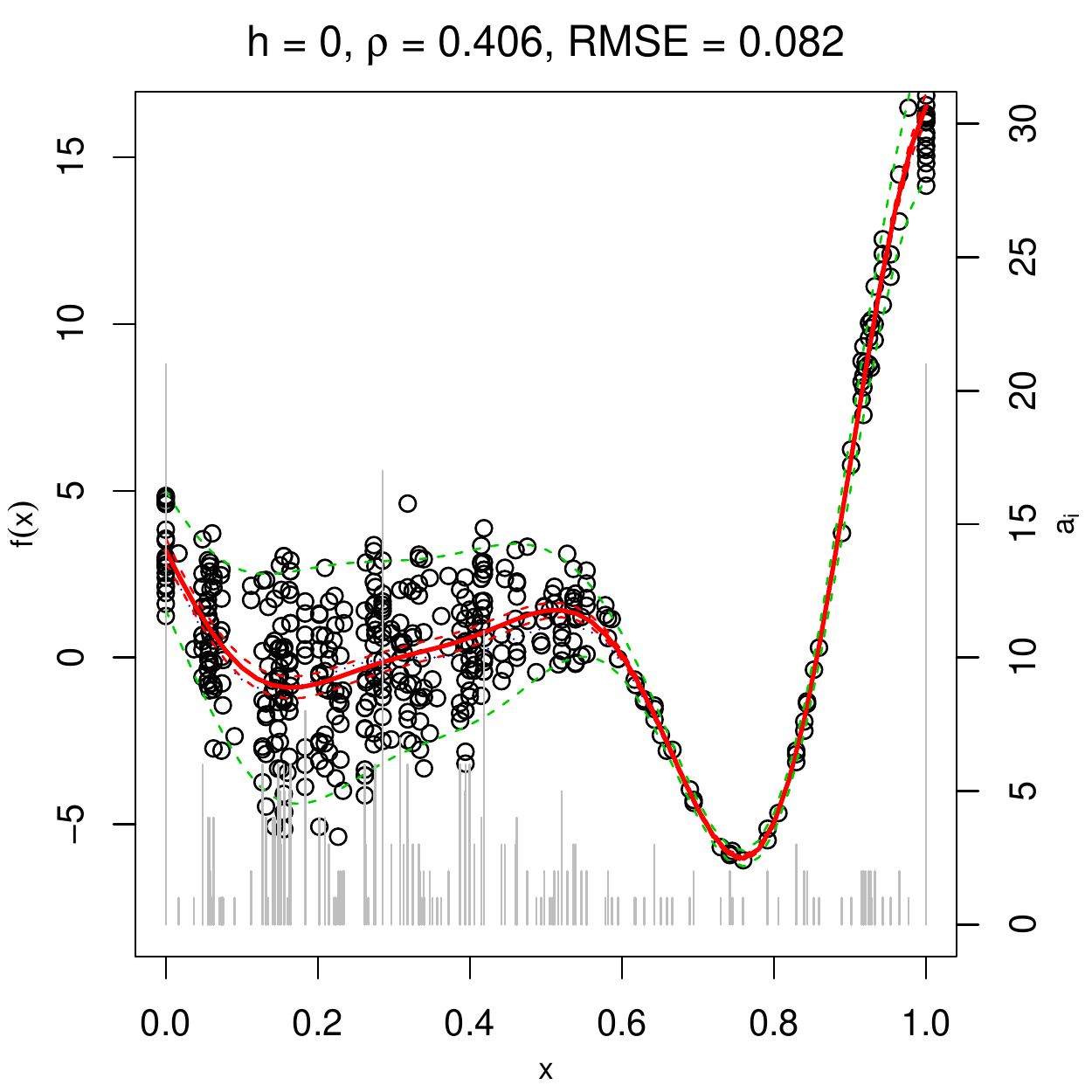}
	\includegraphics[scale=0.38,trim=44 30 44 0,clip=TRUE]{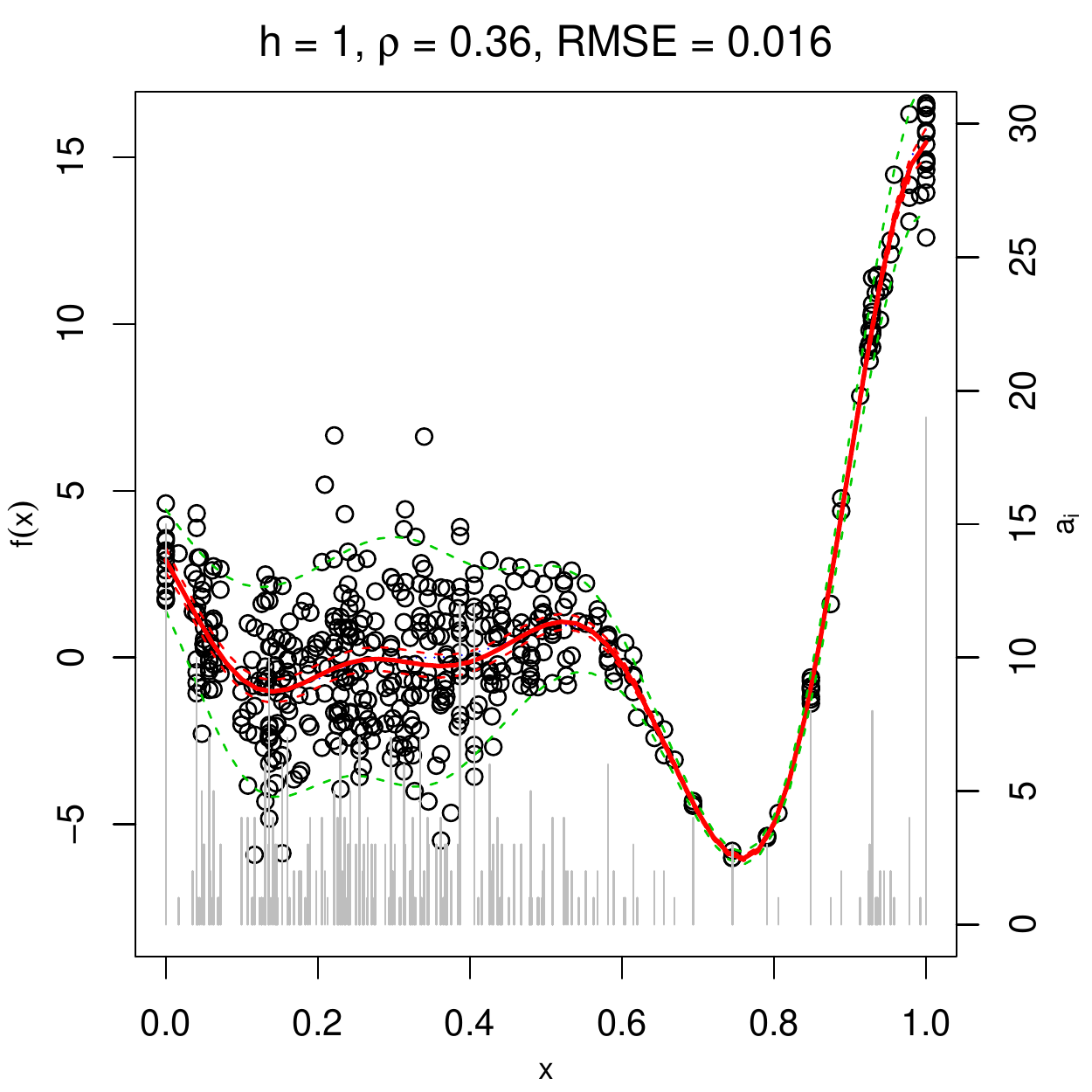}
	\includegraphics[scale=0.38,trim=44 30 0 0,clip=TRUE]{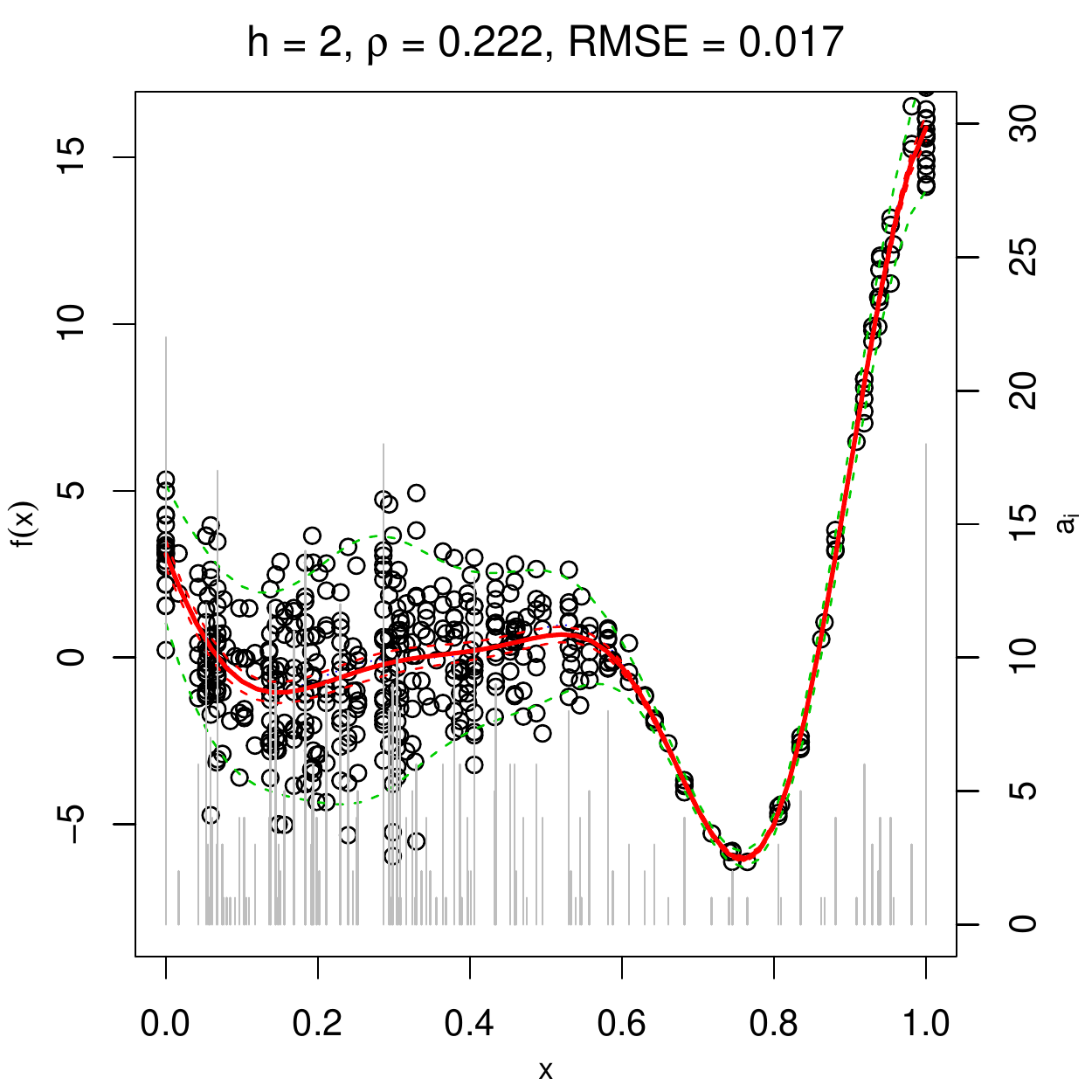}

	\includegraphics[scale=0.38,trim=0 0 44 0,clip=TRUE]{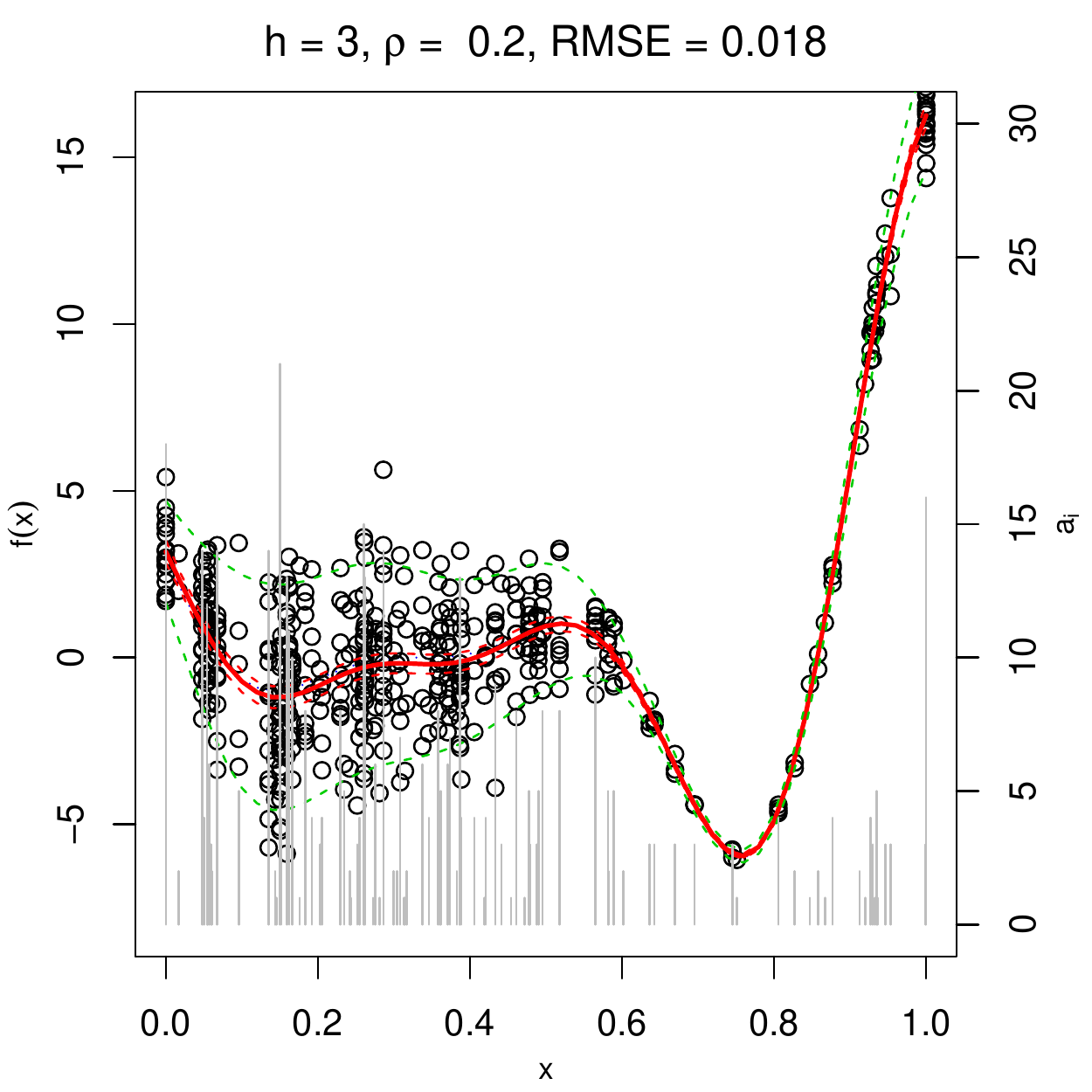}
	\includegraphics[scale=0.38,trim=44 0 44 0,clip=TRUE]{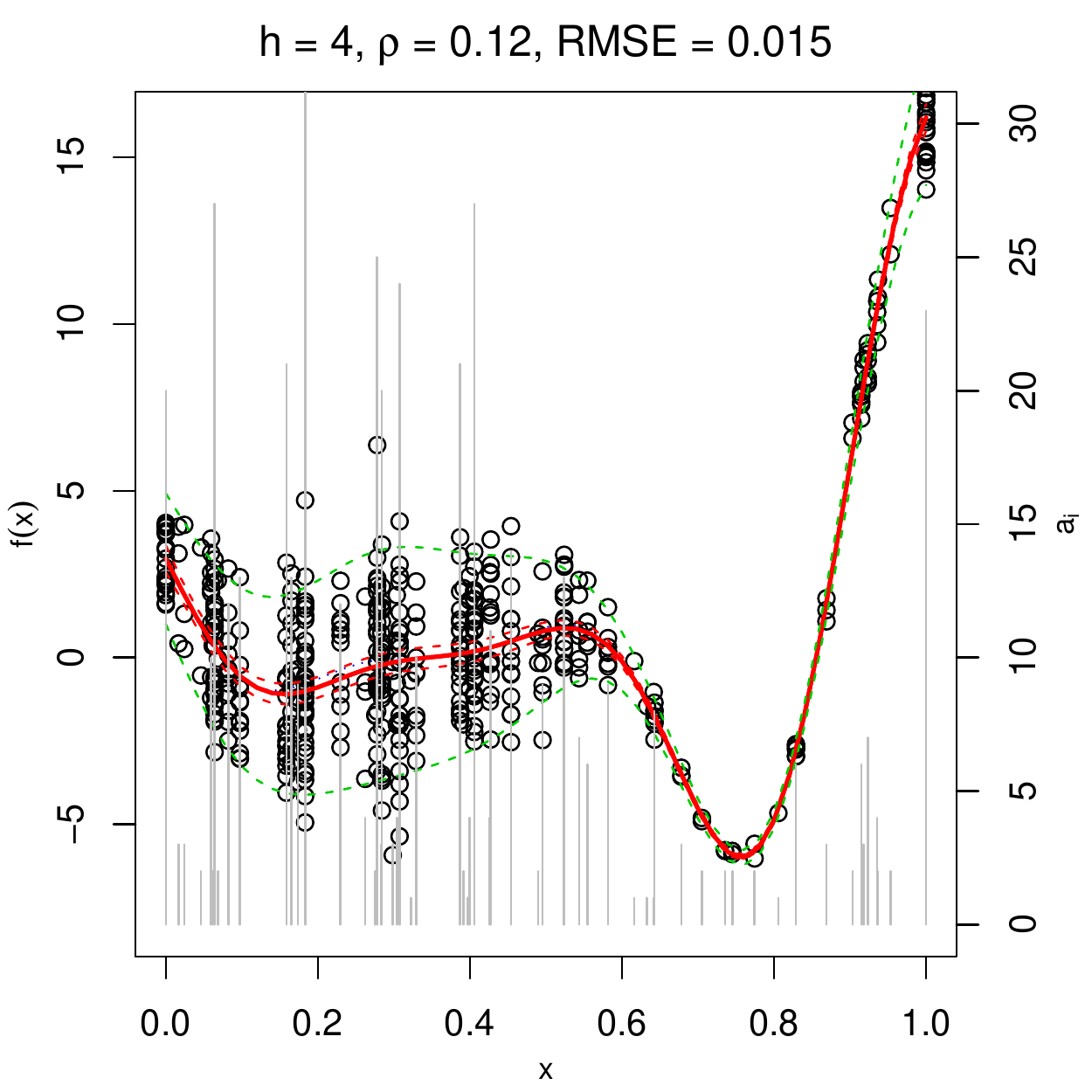}
	\includegraphics[scale=0.38,trim=44 0 44 0,clip=TRUE]{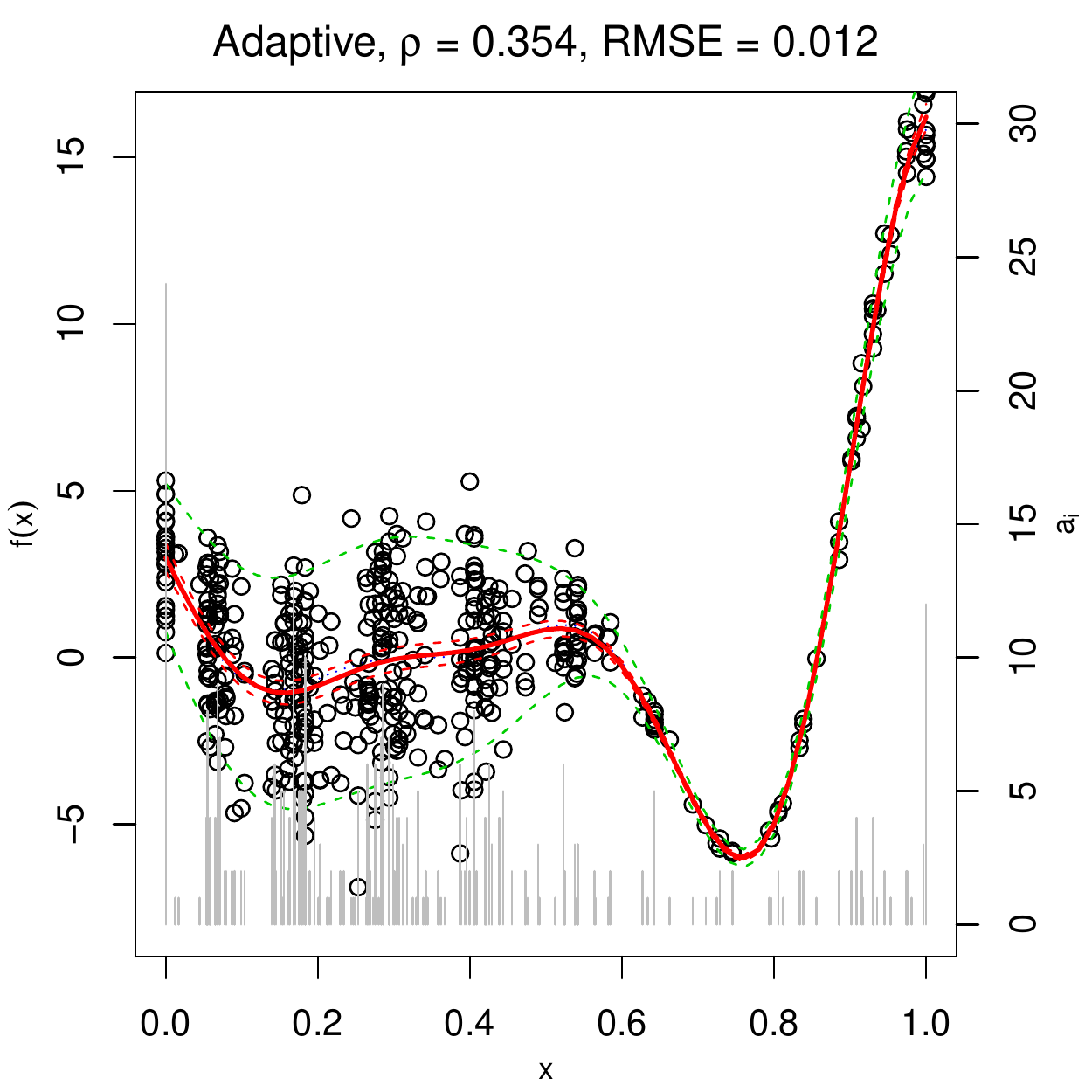}
	\includegraphics[scale=0.38,trim=44 0 0 0,clip=TRUE]{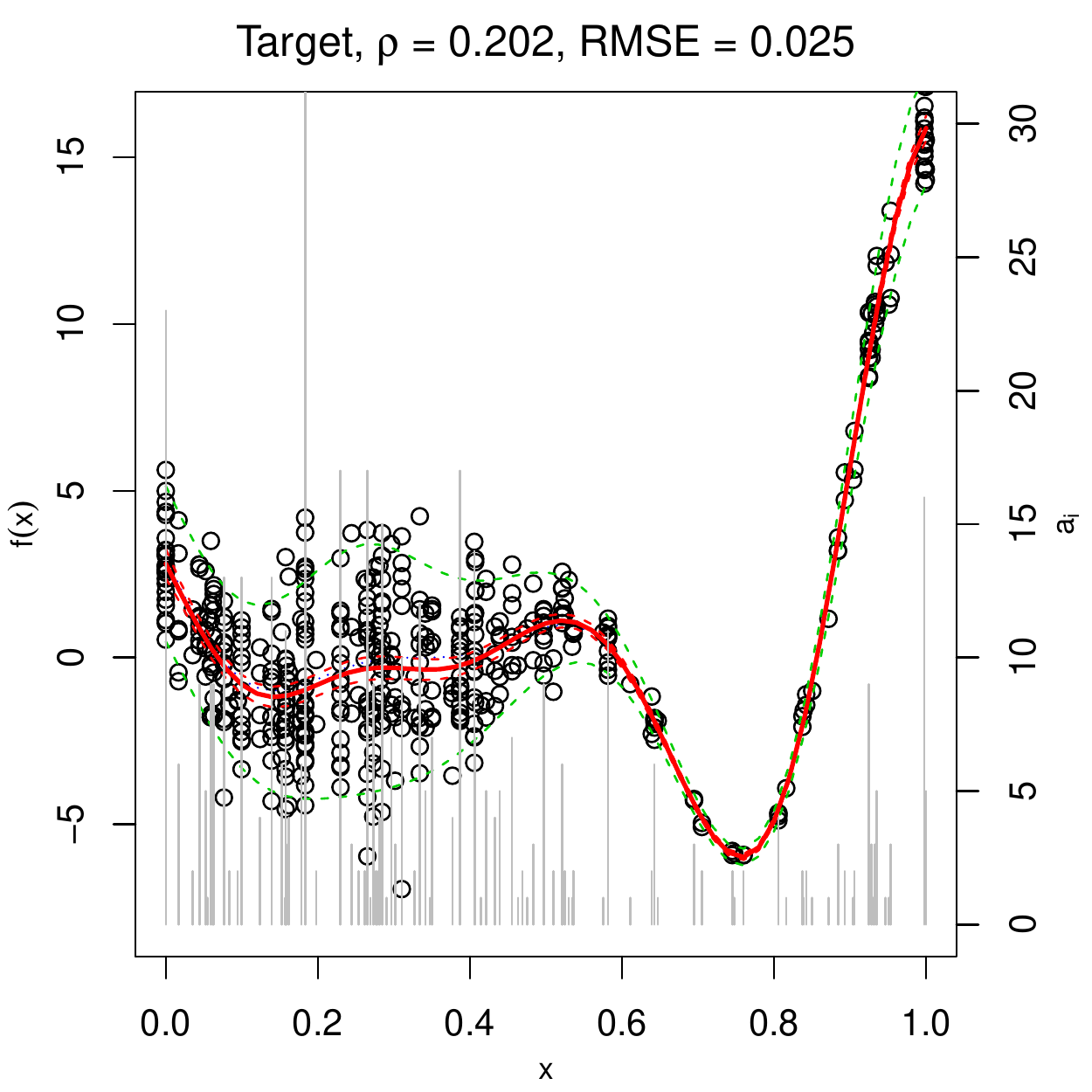}
  \caption{Results on the 1d example by varying the horizon. Grey vertical
  segments indicate the number of replicates at a given location.}%
  \label{fig:1D}%
\end{figure}

Results are presented in Figure \ref{fig:1D} for a total budget of
$N_{\mathrm{max}} = 500$.  Each panel in the figure corresponds to a different
look-ahead horizon $h$, with the final two involving Adaptive and Target
schemes. There are several noteworthy observations. Notice that as the horizon
is increased, more replicates are added.  See $\rho=n/N$ reported in the main
title of each panel. The design density is greatest in the high variance parts
of the space, and that density is increasingly replaced by replication when
the horizon is increased. The effect is most drastic from $h = -1$ to $h = 0$,
with the ratio of unique designs over total designs dropping by more than half
without impact on performance. Notice that replicates are added even with
$h=-1$, at the extremities of the space. Results with high horizons and Adapt
and Target schemes end up having both fewer unique designs and a higher
accuracy.  The very best RMSE results are provided by $h=4$ and the Adapt
(\ref{eq:hadapt}) scheme. In the latter case just 60 unique locations are used
($\rho=0.12$), with some design points replicated as many as thirty times.

\subsection{Synthetic simulation experiment}

Here we expand previous 1-d illustrative example by exploring variation over
data-generating mechanisms via Monte Carlo (MC) with input space $\x \in [0,
1]$. Using the hyperparameter setting outlined in Section \ref{sec:hetGP}, we
consider a process with noise structure $\Lan$ sampled as $\log \Lan
\sim \mathrm{GP}(0, \nu_g \Cg)$, where $\Cg$ is stationary with Mat\'ern $5/2$
kernel $k_{(g)}$. Then observations are drawn via $Y |
\Lan \sim \mathrm{GP}(0, \Kn)$, where $\Kn = \nu(\Cn + \Lan)$ and $\Cn$ is
again Mat\'ern $5/2$. We set $\theta=0.1$, and $\nu=1$ for the mean GP, and
$\theta_{(g)}=0.5$ and $\nu_{(g)}=7^2$ for the noise GP. To manage the MC
variance between runs we normalized the $\Lan$-values thus obtained so
that the average signal-to-noise ratio was one.

We considered a budget of $N=200$ and studied various
strategies for design---comparing one-shot space-filling designs without or
with replication to sequential designs with a lookahead horizon of $h=0$---and
for modeling, testing both homoskedastic and heteroskedastic GPs. These are
enumerated as follows: (i) homoskedastic GP without replication using an
$n=200$ grid design; (ii) {\tt hetGP} without replication, again with an
$n=200$ grid; (iii) {\tt hetGP} with one-shot space-filling design with random
replication on an $n=40$ grid with random $a_i \in
\{1,\dots,10\}$; (iv) sequential learning and design using a homoskedastic GP
initialized with a single-replicate $n=40$ grid, iterating until $N=200$; and
(v) sequential learning and design using {\tt hetGP} initialized with
a single-replicate $n=40$ grid, iterating until $N=200$.

\begin{figure}[ht!]
  \centering
  	\includegraphics[scale=0.45,trim=20 0 0 0]{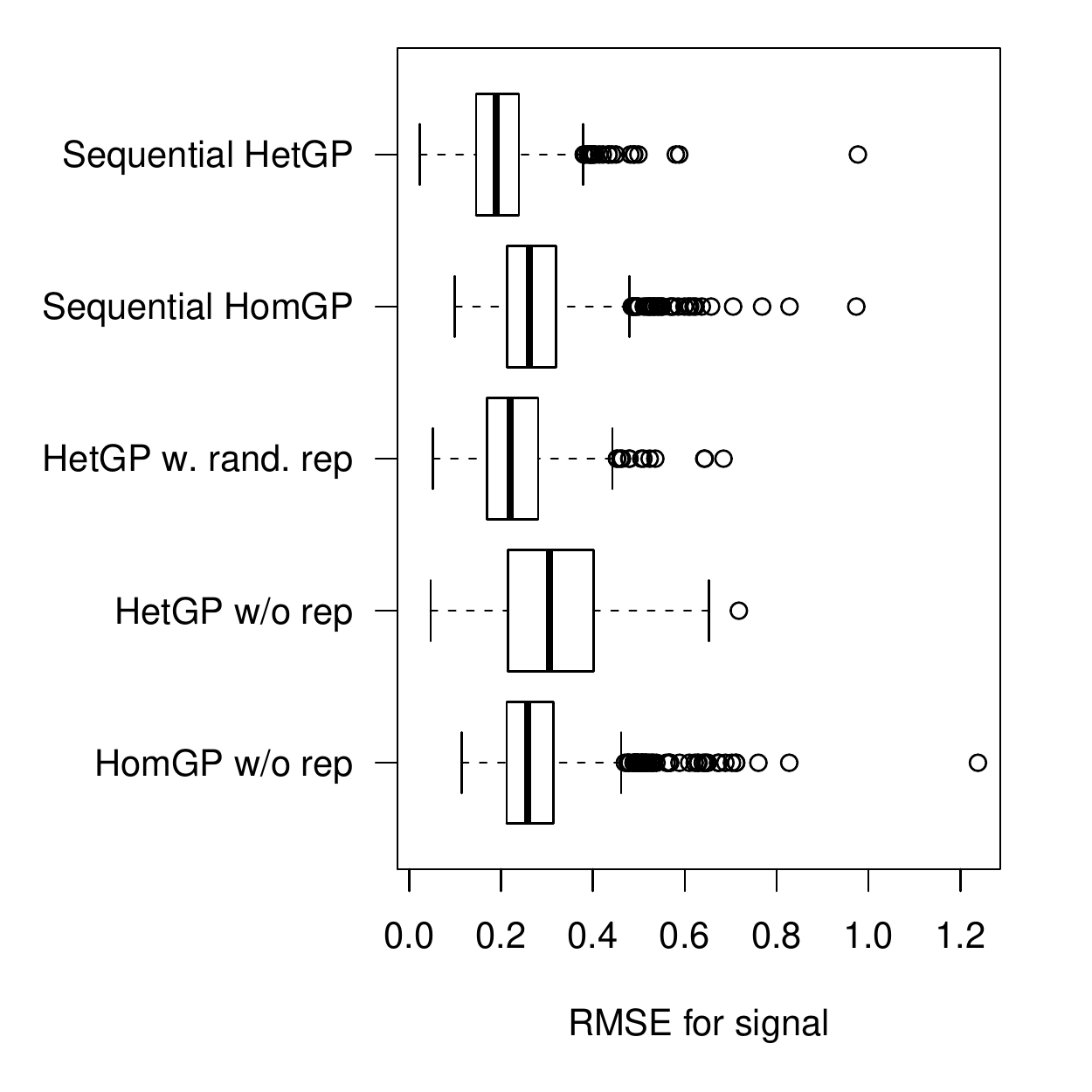}%
	\includegraphics[scale=0.45,trim=130 0 0 0,clip=TRUE]{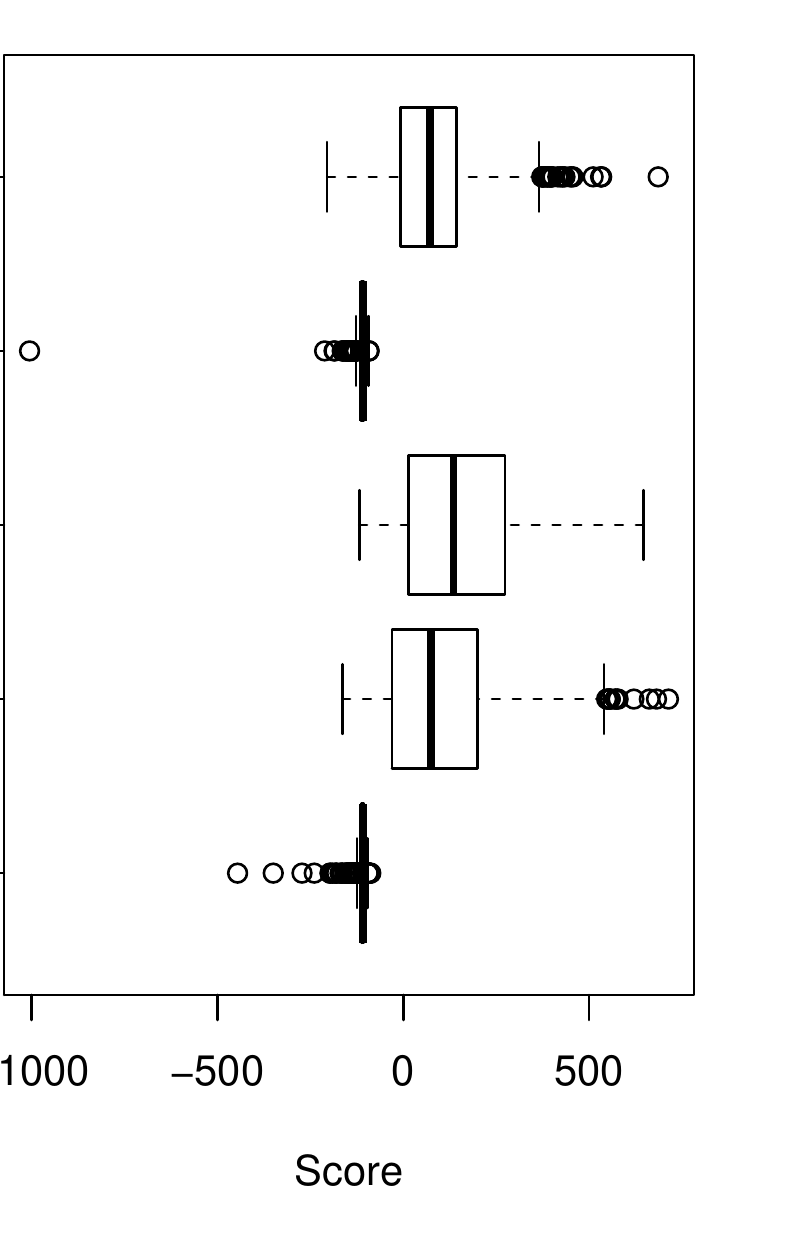}%
	\includegraphics[scale=0.45,trim=130 0 0 0,clip=TRUE]{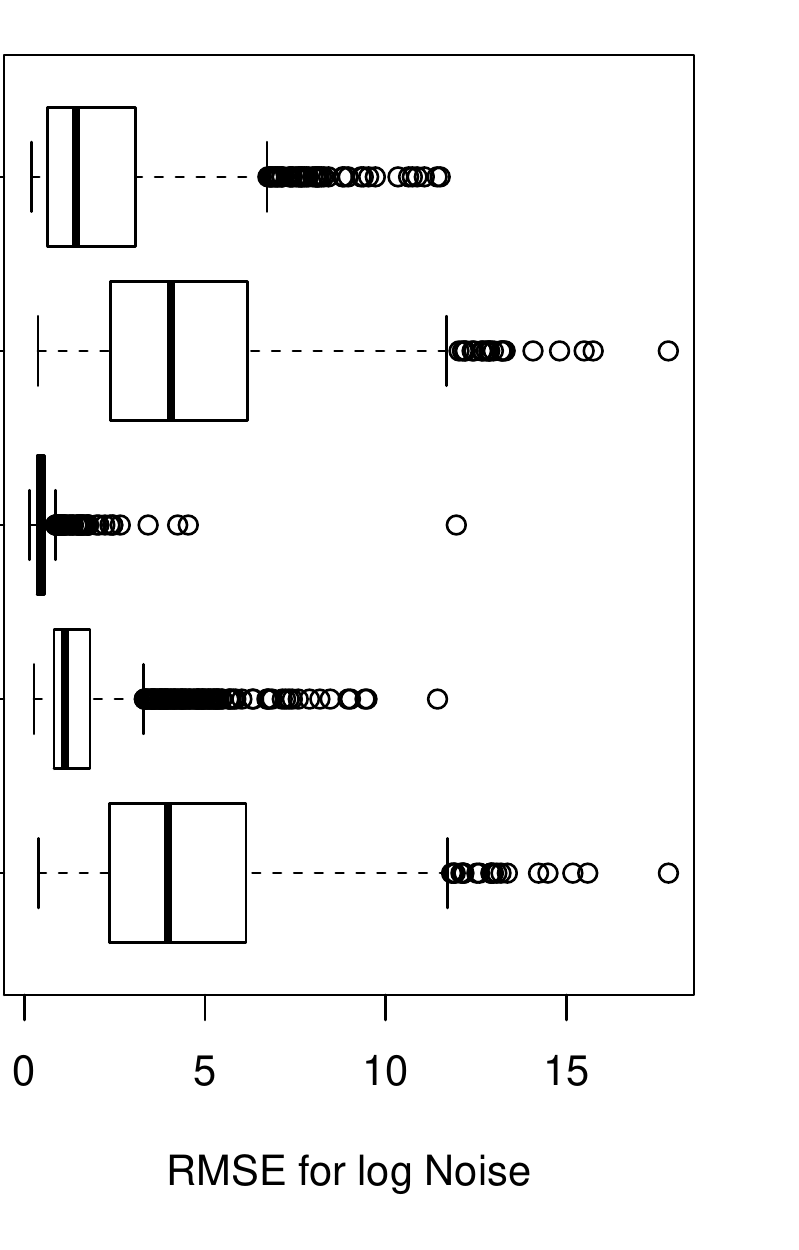}%
	\vspace{-0.25cm}
   \caption{Result from on-dimensional synthetic Monte Carlo experiment in terms
   of RMSE to the true mean (left), proper scores (middle) and RMSE to true log noise (right).
  }%
  \label{fig:test1D}%
\end{figure}

Figure \ref{fig:test1D} summarizes the results from 1000 MC
replicates, illustrating the subtle balance between replication and
exploration. As can be seen in the left panel, our proposed sequential {\tt
hetGP} performs the best in terms of out-of-sample RMSE.  To investigate the
statistical significance of RMSE differences we conducted one-sided matched
Wilcoxon signed-rank tests of adjacent performers (better v.~next-best) with
the order based on median RMSE. The corresponding $p$-values are  $4.80 \times
10^{-24}$, $9.66 \times 10^{-26}$, $0.0956$ and $ 2.66 \times 10^{-12}$. For
example, the test involving our best method, {\tt hetGP} with sequential
design, versus the second best, {\tt hetGP} with random replication, suggests
that the former significantly out-performs the latter. Since our IMSPE design
criteria emphasized mean-squared prediction error, it is refreshing that our
proposed method wins (significantly) on that metric. The only such comparison
which did not ``reject the null at the 5\% level'' involved pitting sequential
versus uniform design with a homoskedastic GP. The value of
proceeding sequentially is much diminished without the capability to learn a
differential noise level.

The center and right panels of the figure show that other design variations
may be preferred for other performance metrics.  Observe that one-shot
space-filling design with random replication using {\tt hetGP} wins when using
proper scores. Apparently, random replication yields better estimates of
predictive variance when comparing to the truth.  See the right panel.
Space-filling and uniform replication are easily achieved in this
one-dimensional case, but may not port well to higher dimension as our later,
more realistic, examples show. Our sequential {\tt hetGP}, coming in second
here on score and log noise RMSE, offers more robustness as the input
dimension increases.

\subsection{Susceptible-Infected-Recovered (SIR) epidemic model}

Our first real example deals with estimating the future number of
infecteds in a stochastic Susceptible-Infected-Recovered (SIR) epidemic model. This
is a standard model for cost-benefit analysis of public health interventions
related to communicable diseases, such as influenza or dengue. For our
purposes we treat it as a 2d input space indexed by the count $I_0 \in
\mathbb{N}$ of initial infecteds and $S_0 \in \mathbb{N}$ of initial
susceptibles (the total population size $M \ge I_0 + S_0$ is pre-fixed; the
rest of the population is viewed as immune to the disease). The pair $(I_t,
S_t) \in \{ S+I \le M \}$ evolves as a continuous-time Markov chain (easily
simulated) following certain non-linear (hence analytically intractable)
transition rates, until eventually $I_t = 0$ and the epidemic dies out.  The
response $f(S,I)$ is the expected aggregate number of infected-days,
$\int_0^\infty I_t \; dt$ averaged across the Markov chain trajectories; determining $f(S,I)$ is a first step towards
constructing adaptive epidemic response policies. It is important to note that
the signal-to-noise ratio is varying drastically over $D$, with a zero
variance at $I = 0$ (where $Y \equiv 0$) and up to $r(\x) \approx 90^2$ on the
left part of the domain, in the critical region where the stochasticity in
infections leads to either a quick burn-out in infecteds or a rapid infection
of a significant population fraction.

Whereas \citet{Binois2016} considered static space-filling designs with random
numbers of replicates, with a favorable comparison to SK, here we focus on
aspects of sequential design, in particular the effect of horizon $h$ in the
IMSPE with lookahead over replication.  We perform a Monte Carlo experiment
wherein designs are initialized with $n = N = 10$ unique design locations
(just one observation each), and grown to size $N=500$ over sequential design
iterations, disregarding  how many unique locations $n$ are chosen along the
way. A Mat\'ern kernel with $\nu = 5/2$ is used. The experiment is repeated
30 times and averages of various statistics are reported in Figures
\ref{fig:SIR_res}, \ref{fig:horizons_SIR} and Table \ref{tab:SIR} based on a
testing set placed on a dense grid with a thousand replications each.

\begin{figure}[ht!]
  \centering
  \vspace{-0.5cm}
  	\includegraphics[scale = 0.47,trim=4 0 20 0,clip=TRUE]{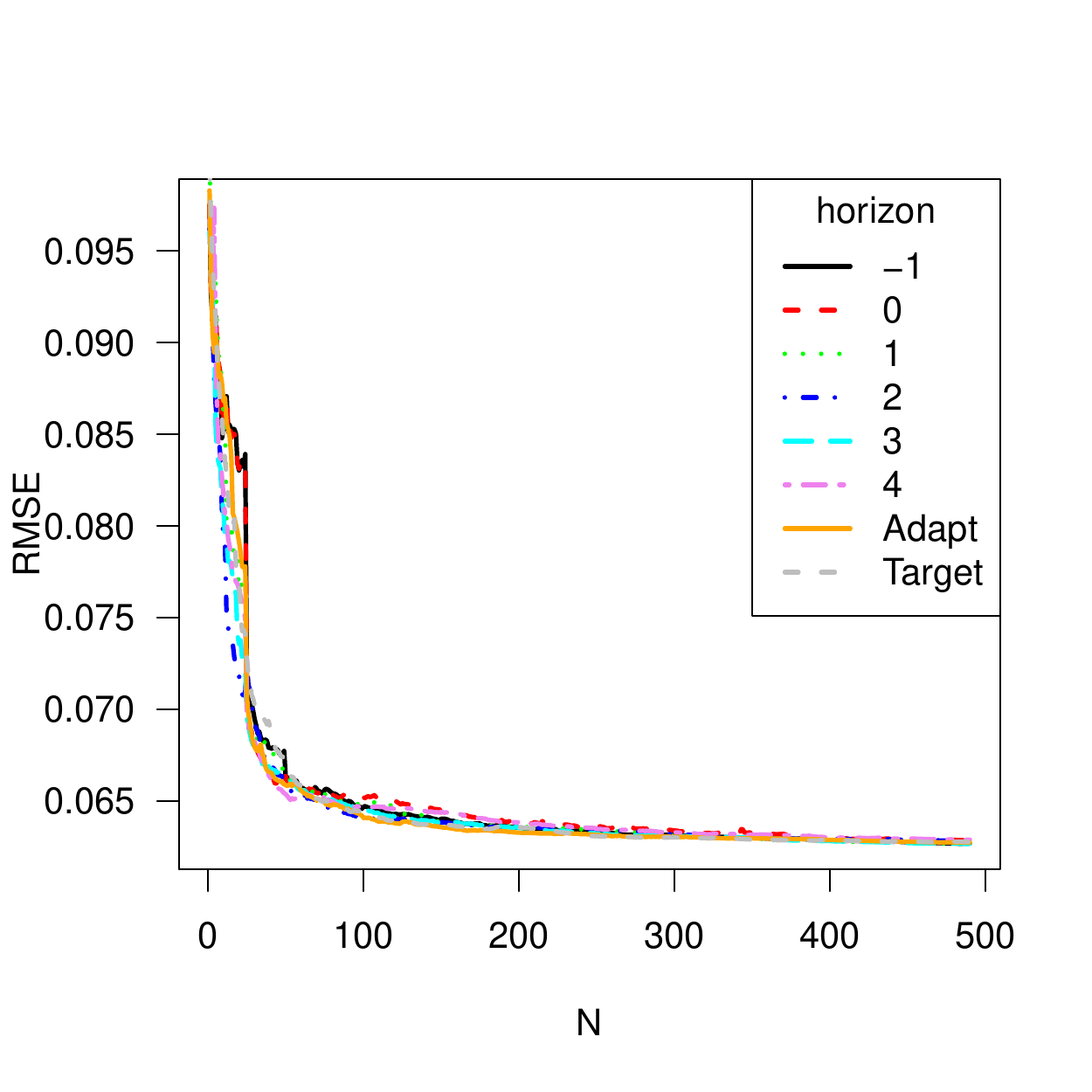}%
	\includegraphics[scale = 0.47,trim=10 0 20 0,clip=TRUE]{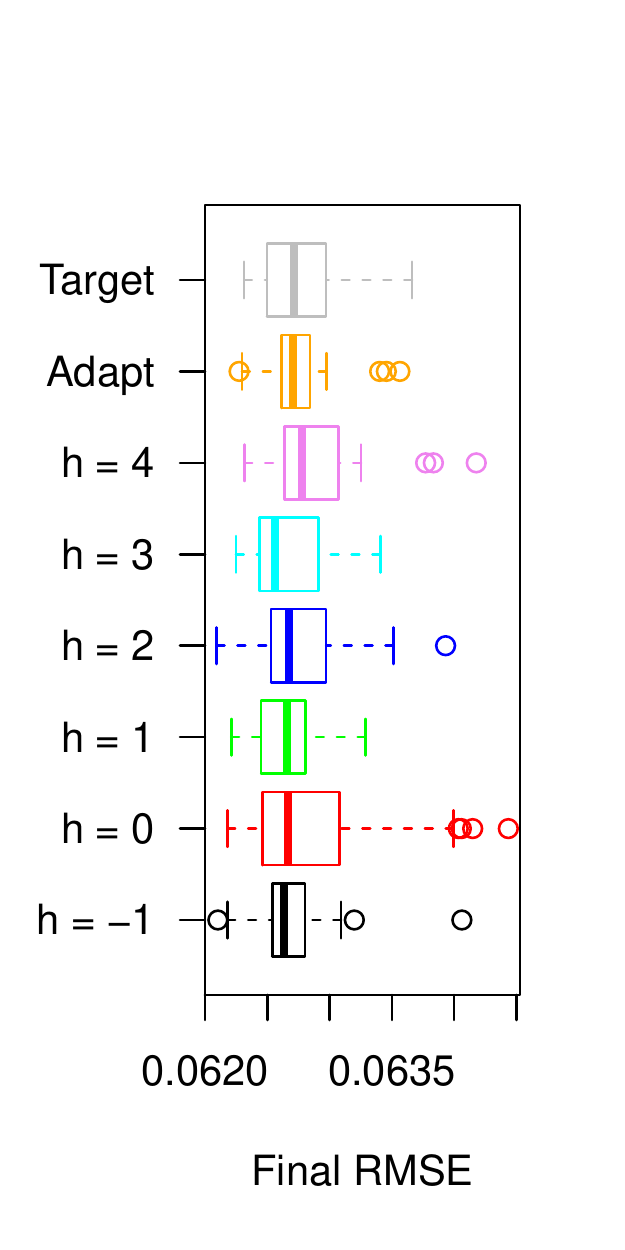}%
	\includegraphics[scale = 0.47,trim=4 0 20 0,clip=TRUE]{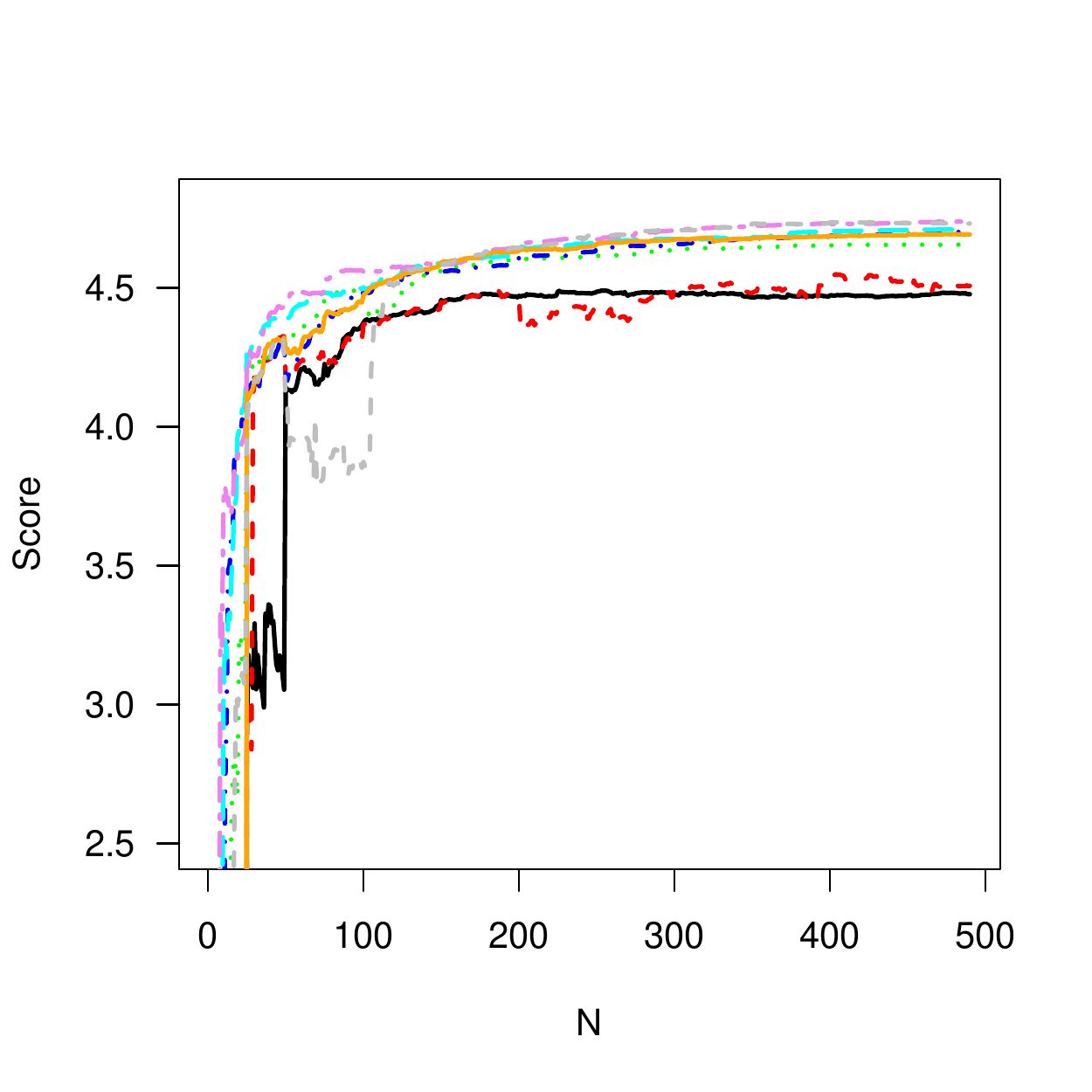}%
	\includegraphics[scale = 0.47,trim=10 0 20 0,clip=TRUE]{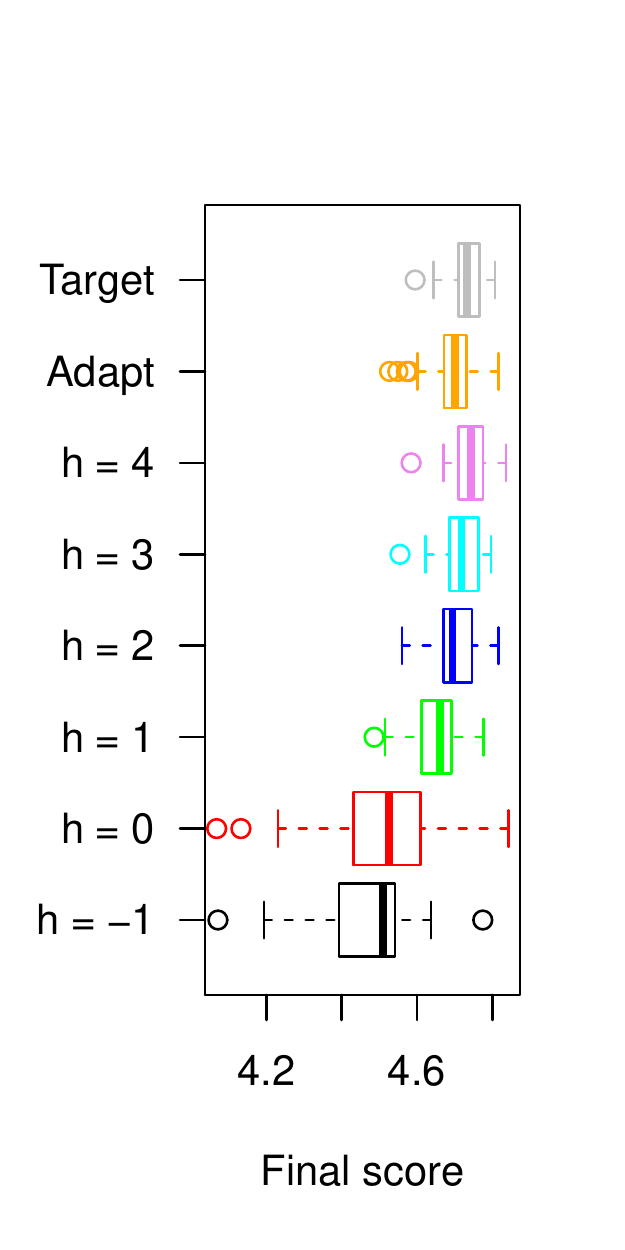}%
	\vspace{-0.25cm}
  \caption{RMSE and score results on the SIR test case over sequential design
  iterations, via summaries 30 MC repetitions. The Target scheme aims for $\rho=n/N = 0.2$.}%
  \label{fig:SIR_res}%
\end{figure}

 In Figure \ref{fig:SIR_res}, the results in terms of RMSEs and score are
 presented. While the RMSEs are barely distinguishable, the scores exhibit
 more spread, and the best results are obtained by the methods leaning the
 most toward replication (i.e., $h = 4$ and Target scheme). Since the signal-to-noise
 ratio is low in some parts of the input space, replication is
 beneficial in terms of RMSE {\em and score}. One reason for the RMSEs not to
 be very different between the alternatives is that the underlying function
 is very smooth. However, the variance surface is more challenging, such that
 having more replicates is helpful in this case, as highlighted by the
 differences in score, shown in the final panel of Figure \ref{fig:SIR_res}.

\begin{figure}[ht!]
  \centering
  \vspace{-0.5cm}
  	\includegraphics[scale = 0.46,trim=4 0 20 0,clip=TRUE]{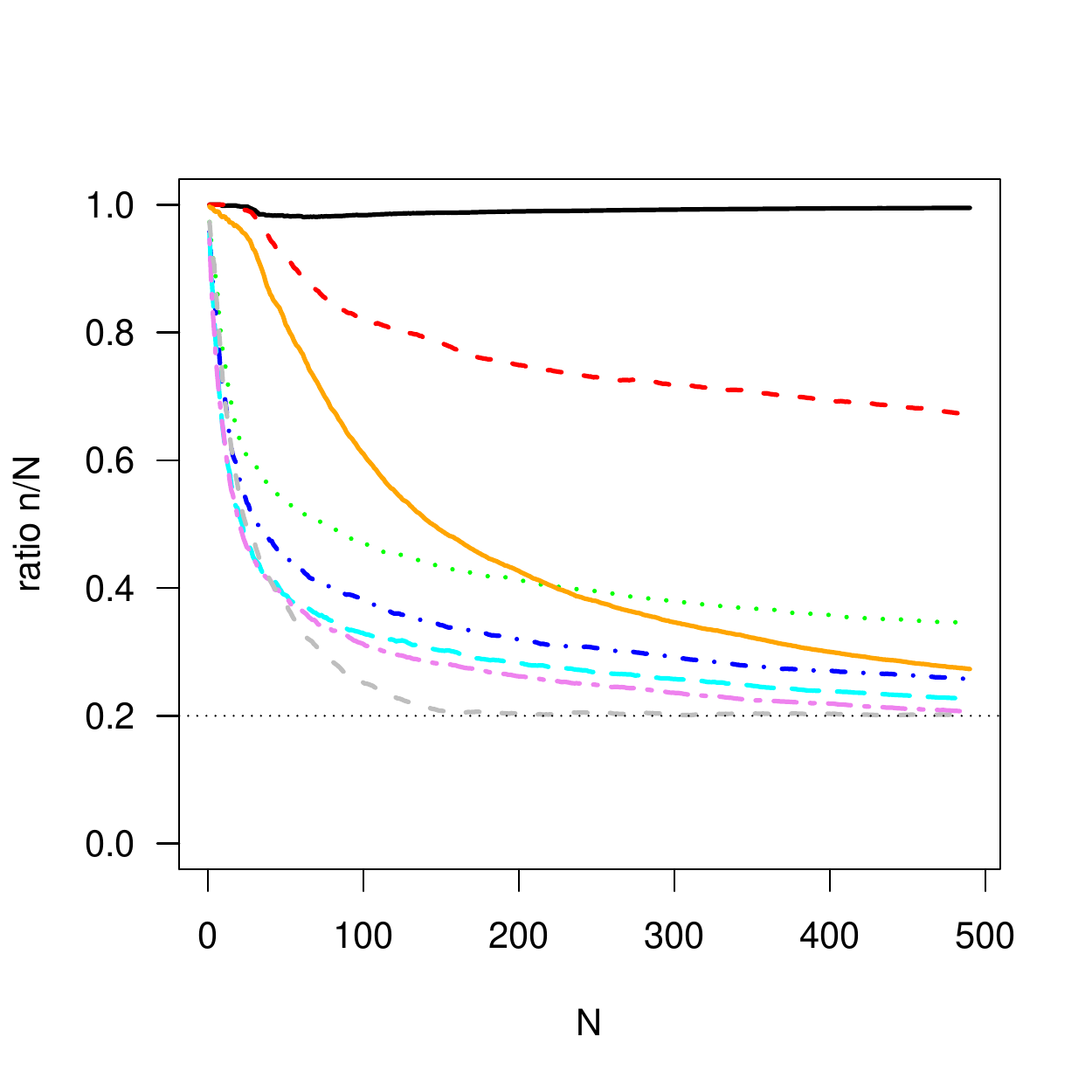}%
	\includegraphics[scale = 0.46,trim=4 0 20 0,clip=TRUE]{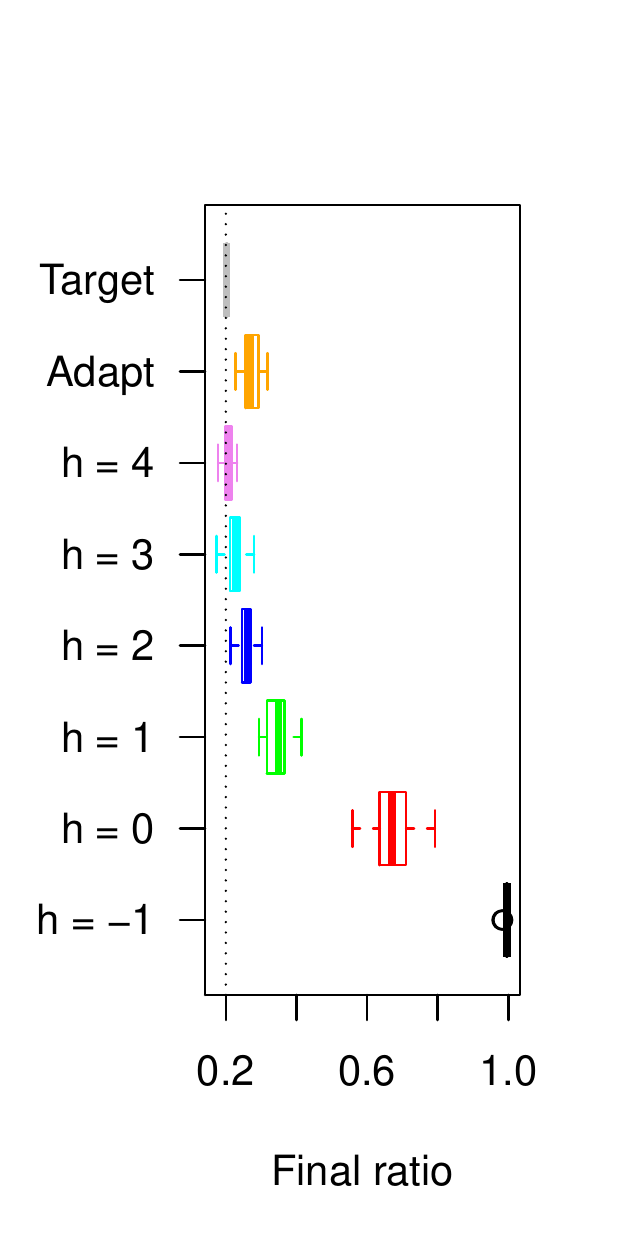}%
	\includegraphics[scale = 0.46,trim=4 0 25 0,clip=TRUE]{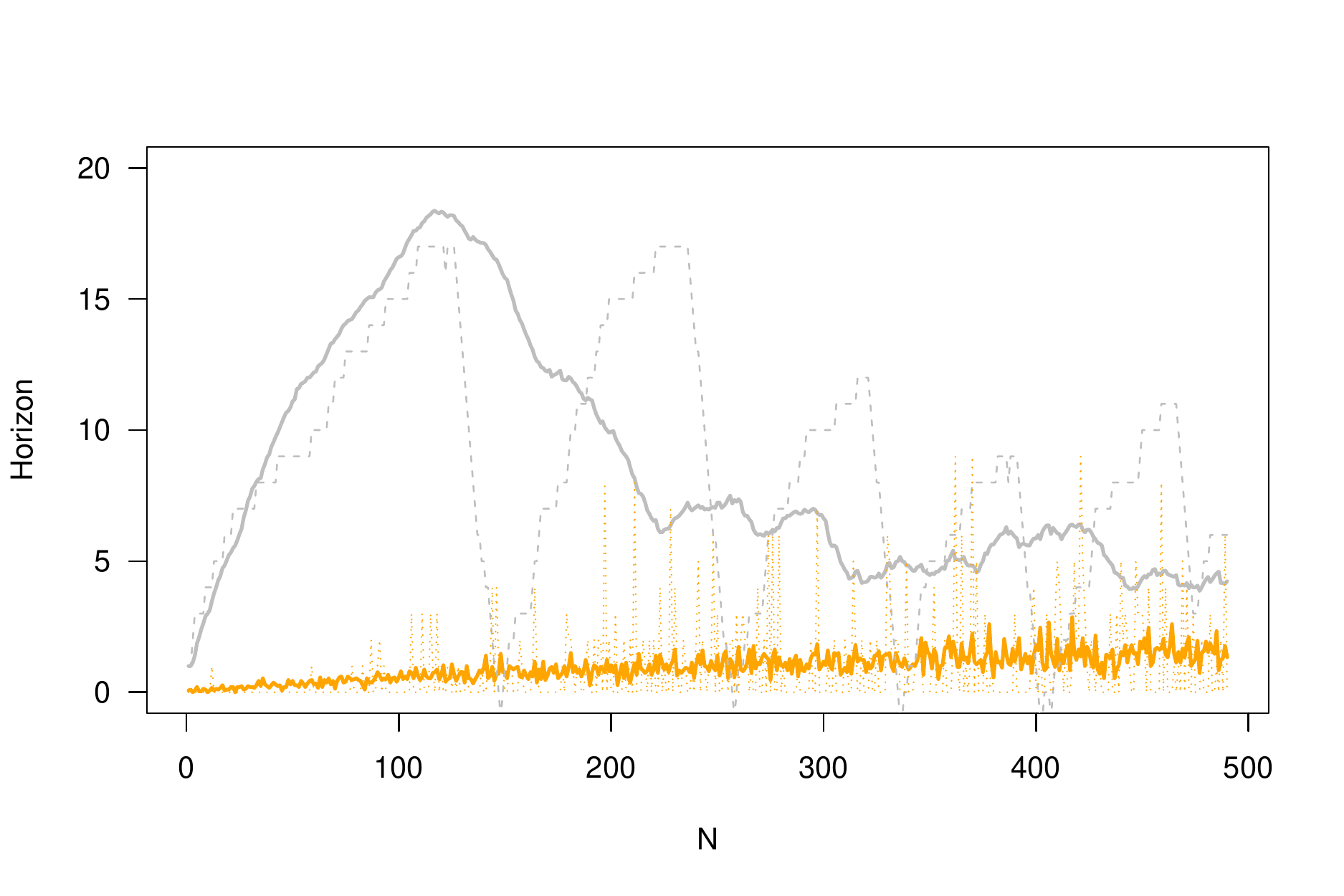}%
	\vspace{-0.25cm}
  \caption{Ratio $n/N$ and horizon evolution on the SIR test case over
  sequential design iterations, via summaries over 30 MC repetitions. The thin
  dotted line indicates the Target ratio of 0.2. Right: thin dotted
  (resp.~thick) lines represent one iteration (resp.~average) of the horizons
  for the Adapt and Target schemes. Colors are the same as in Figure
  \ref{fig:SIR_res}.}
  \label{fig:horizons_SIR}
\end{figure}

\begin{table}[ht!]
\centering
\caption{Average percentage of designs points with no replicates, with more
than five, and running time on the SIR test problem.}
\label{tab:SIR}
\begin{tabular}{l||c|c|c|c|c|c|c|c}
Horizon          & -1   & 0    & 1    & 2   & 3   & 4   & Adapt & Target \\ \hline \hline
Percentage of 1s & 99.2 & 49.6 & 13.6 & 6.9 & 4.8 & 3.5 & 8.8   & 4.1 \\ \hline
Percentage of 5s and more & 0.04 & 1.6 & 5.7 & 6.9 & 7.7 & 7.9 & 6.9 & 7.6 \\ \hline
Time (s)         & 812  & 473	 & 278 & 257 & 259 & 271 & 306 & 288
\end{tabular}
\end{table}

Moving on to Figure \ref{fig:horizons_SIR}, the left and center panels show
the ratio of unique locations over the total design size: $n/N$. As expected,
as the horizon $h$ increases, more replicates are selected. In turn, this
lowers the computation time, as reported in Table \ref{tab:SIR}. In
particular, observe that the computational cost of looking ahead is negligible
next to the cost saved by having smaller $n$ relative to $N$.  The final panel
in Figure \ref{fig:horizons_SIR} shows how the horizon $h$ evolves when fixing
a Target ratio of  $\rho = 0.2$ in (\ref{eq:htarget}), i.e., an average of 5
replicates per unique design location) or learning it with the adaptive scheme
(\ref{eq:hadapt}). Notice that the Target scheme with $\rho = 0.2$ sometimes
utilizes horizons higher than $h \geq 15$, yet the computational cost is never
higher than the high-fixed-horizon results, which offer the best performance
for this problem. Due to its random nature, the Adapt scheme changes abruptly
between algorithm runs, but its horizon $h_N$ is increasing on average in $N$.

\begin{figure}[ht!]
  \centering%
	\begin{subfigure}[t]{0.25\textwidth}%
	\centering%
	\includegraphics[width=0.95\textwidth]{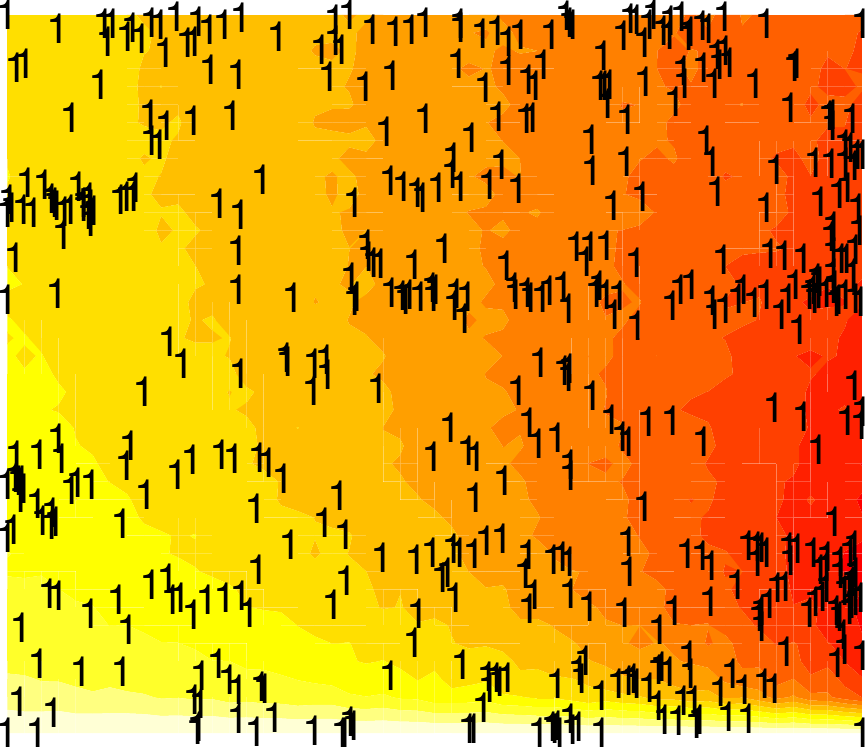}%
	\caption{$h = -1$}
	\end{subfigure}%
		\begin{subfigure}[t]{0.25\textwidth}%
	\centering%
	\includegraphics[width=0.95\textwidth]{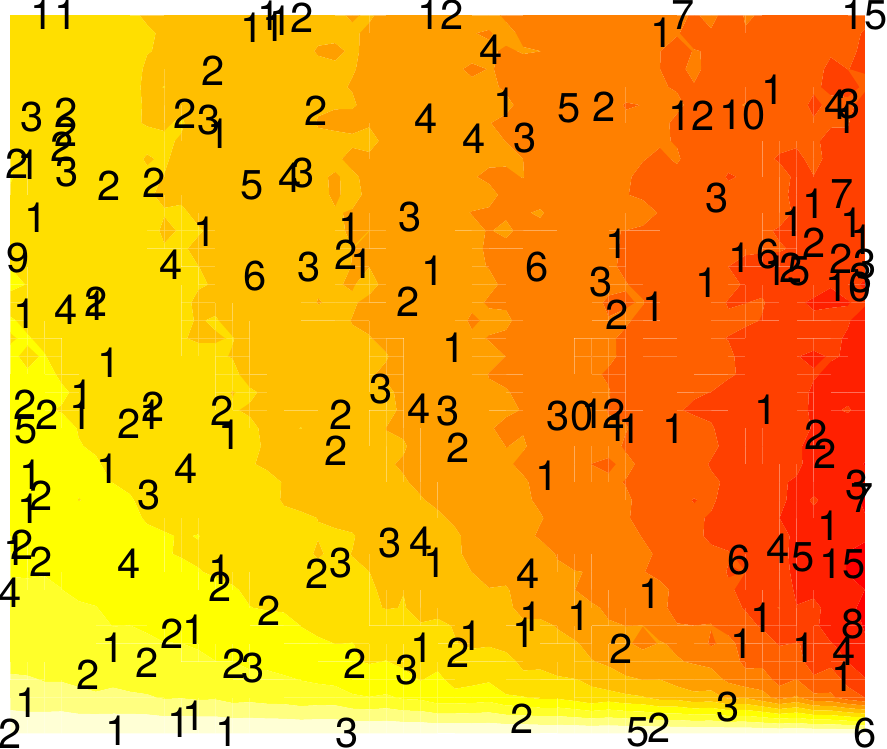}%
	\caption{$h = 1$}
	\end{subfigure}%
	\begin{subfigure}[t]{0.25\textwidth}%
	\centering%
	\includegraphics[width=0.95\textwidth]{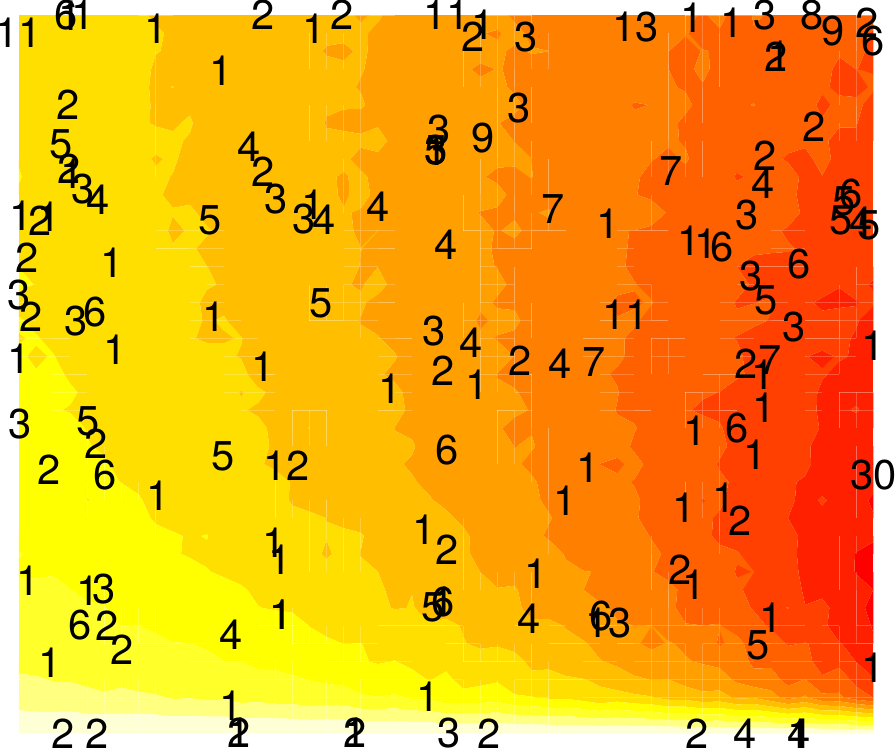}%
	\caption{Adapt}
	\end{subfigure}%
		\begin{subfigure}[t]{0.25\textwidth}%
	\centering%
	\includegraphics[width=0.95\textwidth]{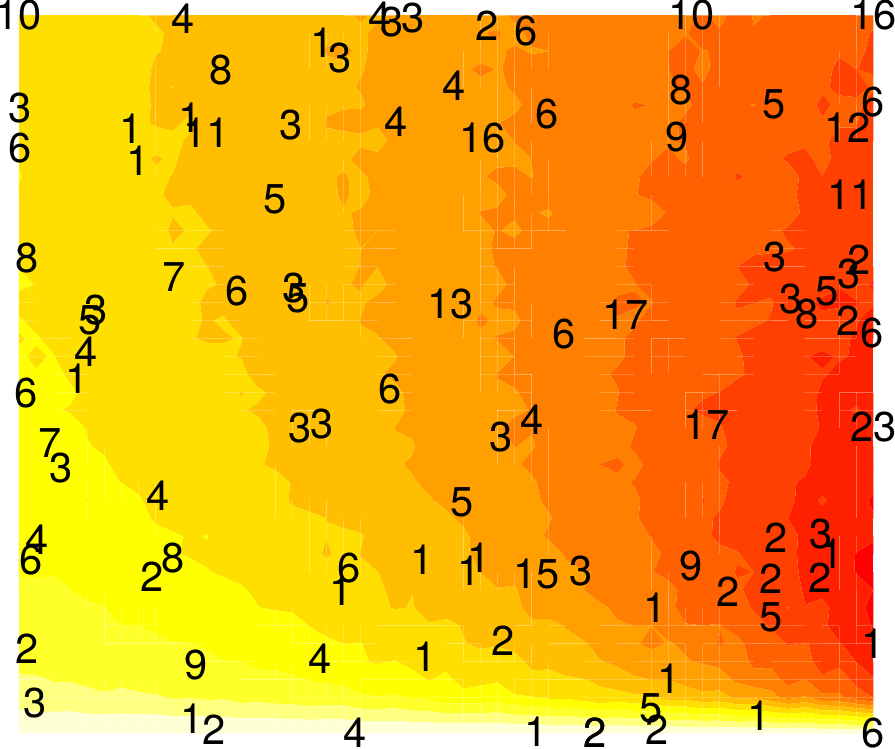}%
	\caption{Target}
	\end{subfigure}%
  \caption{Designs obtained with different strategies for the horizon, where
  numbers indicate how many replicates $a_i$ are performed at a given location
  $\xu_i$. Darker colors indicate higher variance. The x-axis is the number of
  susceptibles, from 1200 to 2000 while the y-axis is the initial number of
  infecteds, from, 0 to 200.}%
  \label{fig:SIR_dens}%
\end{figure}

Figure \ref{fig:SIR_dens} provides a visual indication of the density of
design throughout the input space for fixed and tuned horizons. As expected,
in all panels the density of inputs in the design is higher in high variance
parts of the input space. The numbers in the plot indicate the numbers of
replicates $a_i$. Observe that low-horizon heuristics result in mostly $a_i =
1$, whereas for the longer horizons clusters of tightly grouped unique
locations are replaced with replicates. Table \ref{tab:SIR} demonstrates that
this feature is consistent over MC repetitions. Thus, our heuristic is adept
at capturing the basic logic of amalgamating singleton design locations into
replicates, which apparently maintains essentially the same statistical
efficiency while reducing computational overhead by a factor of more than 3.

\subsection{Inventory management}

The assemble to order (ATO) simulation, first introduced by
\citet{hong:nelson:2006} with implementation in {\sf MATLAB} later provided by
\citet{xie:frazier:chick:2012}, comes from inventory management.  The inputs
determine stocks and replenishment schedules for key items in assembled
products, and the simulator estimates revenue by combining inventory costs
with profits obtained from orders which come in following a compound Poisson
random process. \citet{Binois2016} showed the benefit of heteroskedastic
modeling, versus several homoskedastic alternatives, on random space-filling
designs with $n=1000$ unique locations with a random number of replications
(uniform in $1, \ldots, 10$) so that the average full data size was $N=5000$.
Here, one of our aims is to illustrate that by building  a
better design (sequentially), a much lower $N$ is possible without sacrificing accuracy.
\citeauthor{Binois2016} used a proper scoring rule
\citep[][Eq.~(27)]{gneiting:raftery:2007} as their main metric. Since our
IMSPE criterion targets accuracy via squared-error loss we report RMSEs, but
include scores to facilitate comparison to those space-filling designs.  The
best average score reported in Figure 2 of that paper was 3.3, with a min and
max of 2.8 and 3.6 respectively.

Similar to the SIR experiment, we perform the following variations on
sequential IMSPE design, varying the horizon, $h$, of lookahead and offering
the two adaptive horizon schemes outlined in Section \ref{sec:ahead}.  We
initialize with $n=100$ unique space-filling locations and a random number of
replicates, uniform in $\{1,\ldots,10\}$ so that the  starting size is $N=500$
on average.  Subsequently, sequential design iterations are performed until
$N=2000$ total samples are gathered, irregardless of how many unique
locations, $n$, result. The experiment is repeated in a Monte Carlo fashion,
with thirty repeats.

\begin{figure}[ht!]
  \centering
  \vspace{-0.25cm}
  	\includegraphics[scale = 0.47,trim=4 0 20 40,clip=TRUE]{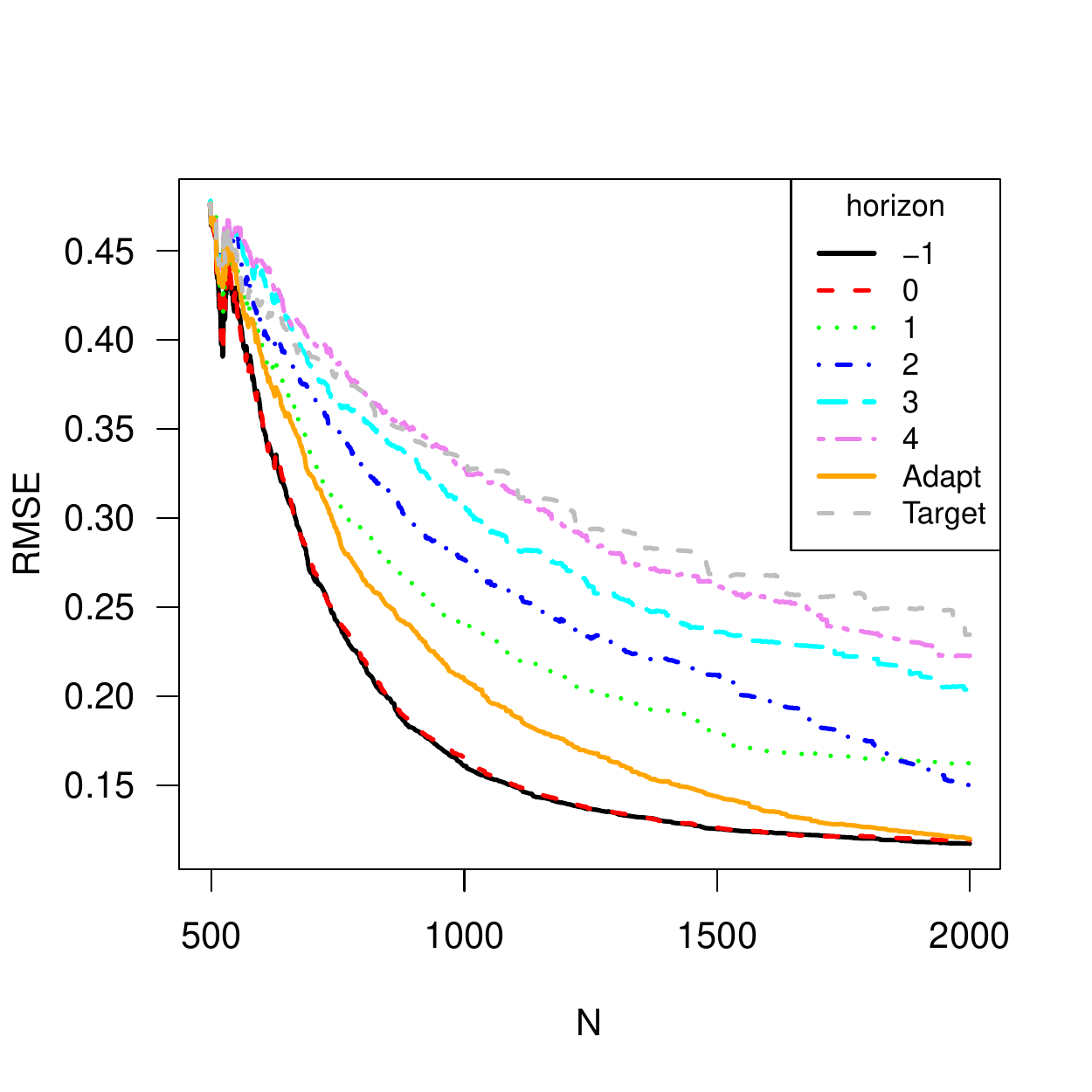}%
	\includegraphics[scale = 0.47,trim=4 0 20 40,clip=TRUE]{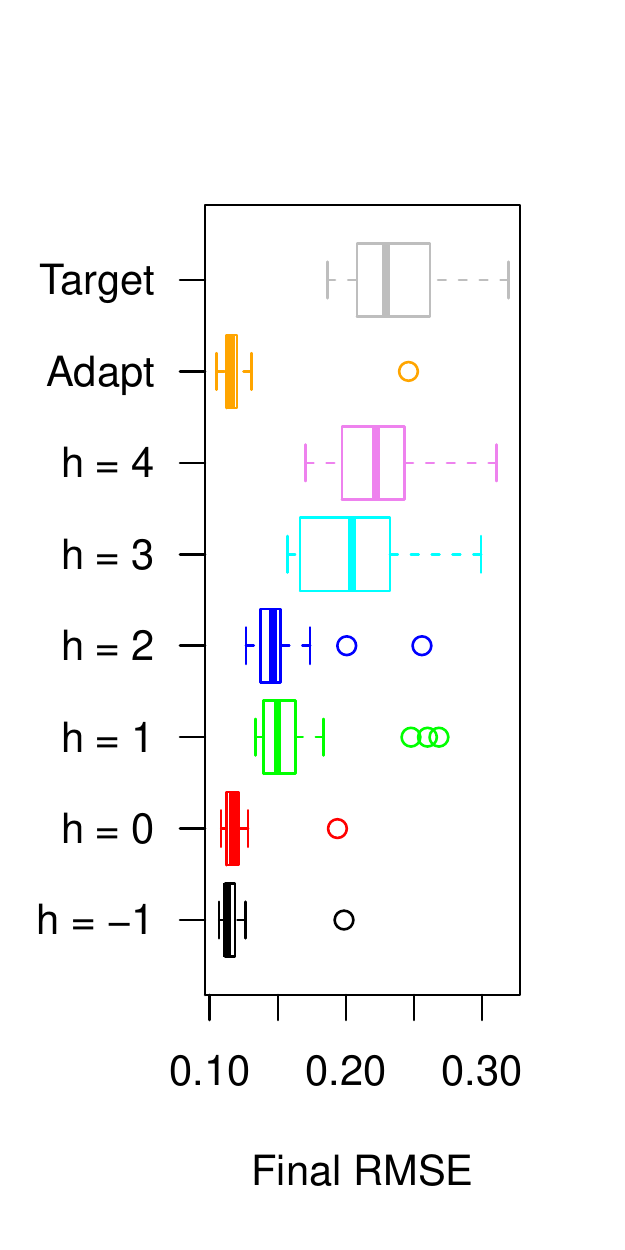}%
	\includegraphics[scale = 0.47,trim=4 0 20 40,clip=TRUE]{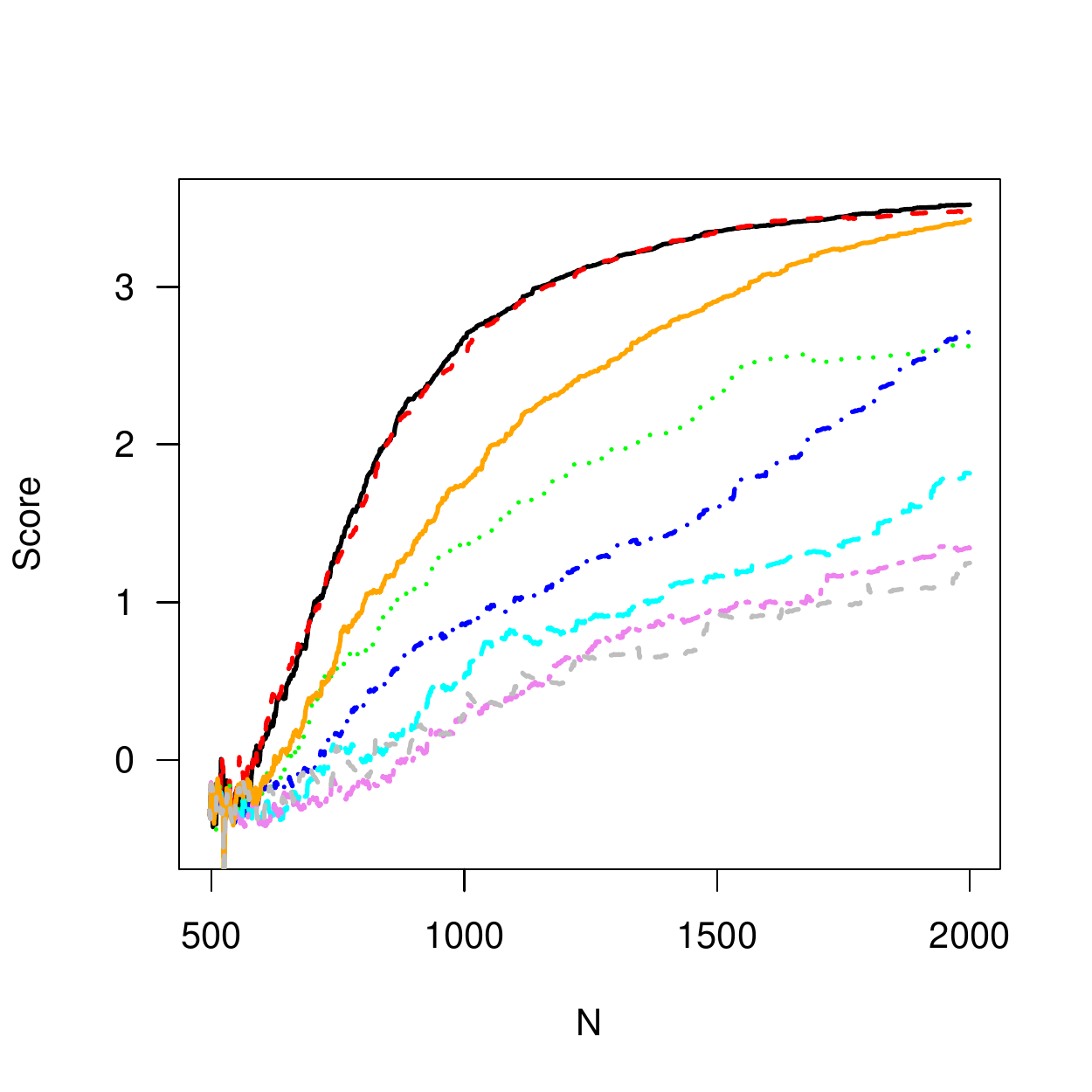}%
	\includegraphics[scale = 0.47,trim=4 0 20 40,clip=TRUE]{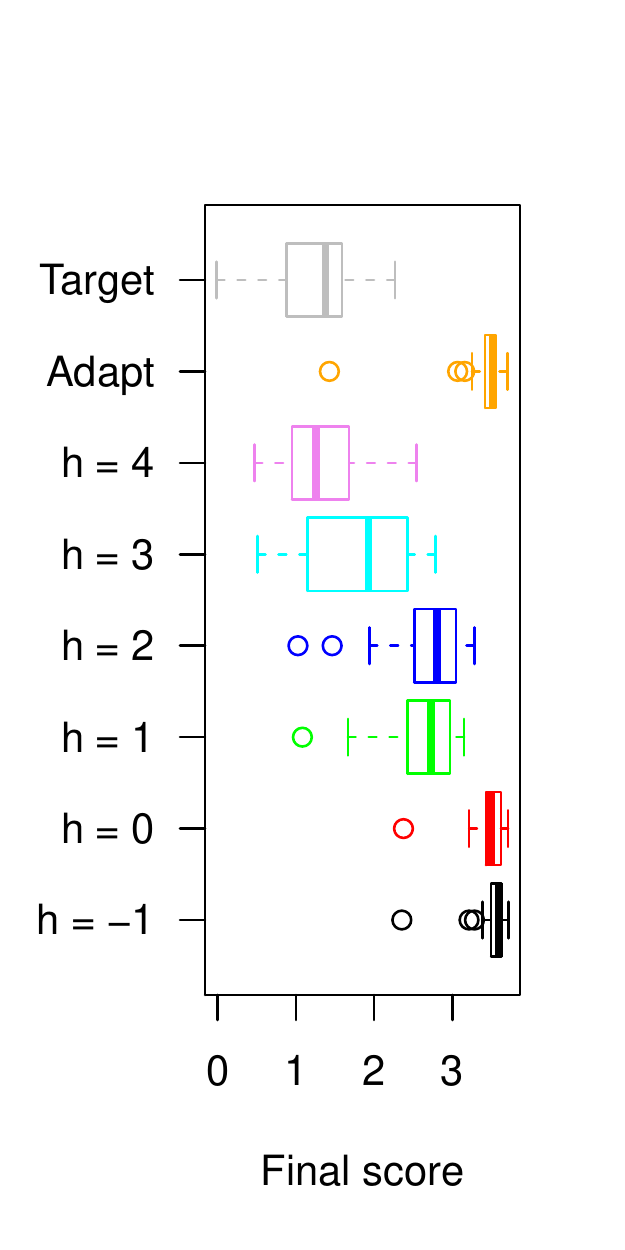}%
	\vspace{-0.25cm}
  \caption{RMSE and score results on the ATO problem, via thirty Monte Carlo iterations, in a format similar to Figure \ref{fig:SIR_res}.}
  \label{f:ato}%
\end{figure}

Figure \ref{f:ato} summarizes the results of the experiment in a format
similar to Figure \ref{fig:SIR_res}.  The take-home message is fairly evident:
in contrast to the SIR example, shorter lookahead horizon is better, owing to
the relatively higher signal-to-noise ratio.  Observe that our average score
is 3.5, and the min and max are 2.4 and 3.7 respectively.  So scores based on
just $N=2000$ samples are higher than in the space-filling $N=5000$
experiment, however the spread is a little wider.  Finally, note that the
Adapt heuristic (\ref{eq:hadapt}) eventually performs as well as the best
horizon ($h=-1$). Targeting $\rho = 0.2$ via (\ref{eq:htarget}), by contrast,
leads to far too little exploration. The Adapt scheme required an average of
682 minutes to build up to a design of size $N=2000$ with an average of
$n=1086$ unique sites (min and max of 465 and 1211 respectively), whereas
Target took only 183 minutes thanks to using $n=400$ locations on average (399
to 405).

\section{Conclusion and perspectives}

This paper addresses a design question which has been around the surrogate
modeling literature, and informally in the community, for many years.  There
is general agreement that the ``fully batched'' version of the problem, of
finding $n$ locations and numbers of replicates $a_1, \dots, a_n$ on each, is
not computationally tractable, although there are some attempts in the recent
literature.  We therefore consider the simpler task of deciding whether the
{\em next} sample should explore or replicate in a sequential design context.
The condition we derive is simple to express, and leads to an intuitive suite
of visualizations.  Proceeding sequentially has merits, not only
computationally but also facilitating ``as you go'' adjustments to help avoid
pathologies arising from feedbacks between design and inference.  However the
procedure is still myopic.  To help correct this we introduced a
computationally tractable lookahead scheme that emphasizes the role of
replication in design. Tuning the horizon of that scheme allows the user to
trade off the dual roles of replication in surrogate modeling design:
computational thriftiness of inference against out-of-sample accuracy,
although as we show these are not always at odds.

Our presentation focused on the integrated mean-squared prediction error
(IMSPE) criteria.  We chose IMSPE because it is popular, but also
because it leads to closed form  
derivatives for optimization, updating equations for search over look-ahead
horizons, and simplifications for entertaining replicates. There is, of
course, a vast literature on model-based design criteria \citep[see,
e.g.,][]{Chen2017,Kleijnen2015} targeting alternative quantities of interest, such
as entropy or information for unknown model parameters. Although designs for
prediction and estimation sometimes coincide, like for linear regression, the
correlation structure for GPs can be a game-changer \citep{Mueller2012}.  It
may well be that other criteria lead to strategies similar to ours, which may be an interesting avenue for future research.

Our implementation and empirical work leveraged a new heteroskedastic Gaussian
process modeling library called {\tt hetGP}, available for {\sf R} on CRAN
\citep{Binois2017}.  Our IMSPE, updates, lookahead procedures, and more are
provided in a recently updated version of the package.  To aid in
reproducibility, our supplementary material contains codes using that library
to reproduce the smaller examples from the paper [Figures
\ref{fig:rollout}--\ref{fig:1D}]. The other examples require rather more
computing, and/or linking between {\sf R} and {\sf MATLAB} for simulation
[ATO], which somewhat challenges ease of replication.  However, we are happy
to provide those codes upon request.
	
Processes (i.e., data generating mechanisms) benefiting from a heteroskedastic
feature bring out the best in our sequential design schemes, demanding a
greater degree of replication in high-noise regions relative to low-noise
ones, confirming the intuition that replication becomes more valuable for
separating signal from noise as the data get noisier \citep[e.g.,][]{Wang2017}.
However, the results we provide are just as valid in the homoskedastic
setting, albeit with somewhat less flair.  In that context, inferring the
right level of replication is a global affair, except perhaps at the edges of
the input space which tend to prefer a slightly higher degree.

Our three sets of examples illustrated that the method both does what it is
designed to do, and that designs with the right trade-off between exploration
and replication perform better than ones which are designed more na\"ively.
These examples span a range of features, from low to high noise (and slow to
rapid change in noise), low to moderate input dimension, and synthetic to real
simulation experiments.  The behavior is diverse but the results are
consistent: sequential design with lookahead-based IMSPE leads to accurate
prediction, and the slight bias toward replication  yields
computationally more thrifty predictors without a compromise on accuracy.
However, sequential design might not always be appropriate.  Sometimes
batching, at least to a small degree, cannot be avoided.  Addressing this
situation represents an exciting avenue for further research.

\subsubsection*{Acknowledgments}

We thank the anonymous reviewers for helpful comments on the earlier version of the paper. All four authors are grateful for support from National Science Foundation grant
DMS-1521702 and DMS-1521743.

\bibliographystyle{jasa}
\bibliography{RepOrExp}

\appendix


\section{Detailed gradient expressions}
\label{ap:grad}

Here we provide expressions for the Section \ref{ss:grad} discussion on the gradient of the IMSPE.
\begin{align*}
\dfrac{ \partial \K^{-1}_{n+1}}{\partial \xnew} & =\dfrac{ \partial} {\partial \xnew} \begin{bmatrix}
\K_{n}^{-1} + \vecg(\xnew) \vecg(\xnew)^ \top \sigma_n^2(\xnew) & \vecg(\xnew) \\
\vecg(\xnew) ^\top & \sigma_n^2(\xnew)^{-1}
 \end{bmatrix}
 = \begin{bmatrix}
 \Hmat(\xnew) & \vech(\xnew) \\
\vech(\xnew)^\top & v_1(\xnew)
 \end{bmatrix} \quad \mbox{as in Eq.~(\ref{eq:dKi})}, \\
\mbox{where} \;\; v_1(\xnew) &:=
\dfrac{ \partial \sigma_n^2(\xnew)^{-1} } {\partial \xnewp}{}
= \dfrac{2\vecd(\xnew)^\top \K_n^{-1} \veck(\xnew) + \frac{\partial r(\xnew)}{\partial \xnewp}}{\left(k(\xnew, \xnew) -\veck(\xnew) ^\top \K_n^{-1} \veck(\xnew) + r(\xnew) \right)^2}, \quad \vecd(\xnew)
:= \dfrac{ \partial \veck(\xnew) } {\partial \xnewp} \nonumber \\
\vech(\xnew) & := \dfrac{ \partial \vecg(\xnew) } {\partial \xnewp}
= - \K_n^{-1} \left( v_1(\xnew) \veck(\xnew)+\sigma_n^2(\xnew)^{-1} \vecd(\xnew) \right) \nonumber \\
\Hmat(\xnew) &:=\dfrac{ \partial \vecg(\xnew) \vecg(\xnew)^ \top \sigma_n^2(\xnew) } {\partial \xnewp}
= \dfrac{ \partial \sigma_n^2(\xnew) } {\partial \xnewp} \vecg(\xnew) \vecg( \xnew)^\top
	\!+\!\sigma_n^2(\xnew) \vech(\xnew) \vecg(\xnew)^\top\!+\!\sigma_n^2(\xnew)  \vecg(\xnew) \vech(\xnew)^\top \nonumber \\
&= v_2(\xnew) \vecg(\xnew) \vecg( \xnew)^\top + \sigma_n^2(\xnew) \left( \vech(\xnew) \vecg(\xnew)^\top
	+ (\vech(\xnew) \vecg(\xnew)^\top)^\top \right), \nonumber
\end{align*}
and $v_2(\xnew) =-2\vecd(\xnew)^\top \K_n^{-1} \veck(\xnew)$.  Similarly, since $\Wn$ does not depend on $\xnew$:
\begin{equation*}
\dfrac{ \partial \W_{n+1}}{\partial \xnewp}= \begin{bmatrix}
\mathbf{0}_{n \times n} & \vecc_1(\xnew) \\
\vecc_1(\xnew)^\top & c_2(\xnew)
 \end{bmatrix}, \quad \quad \mbox{as presented in Eq.~(\ref{eq:dW})}.
 \end{equation*}
Expressions for $\vecc_1(\cdot)$ and $c_2(\cdot)$ for particular kernels may
be found in Appendix \ref{ap:kernel_exp}.

\section{Expressions for common kernels}
\label{ap:kernel_exp}

We consider here four kernels common in practice: Gaussian
(or squared exponential) and Mat\'ern with parameter $\alpha = 5/2, ~3/2, ~1/2$
(the last one being the exponential kernel) and give the
corresponding expressions for $E$, $w$, $\mathbf{d}$, $\mathbf{c}_1$ and $c_2$ as
introduced in Section \ref{sec:imspe}.
Notice that all these kernels are stationary, i.e., $k(\x, \x') = \nu c(\x - \x')$
with $\nu$ the process variance and $c$ the correlation function.
As a consequence, $E = \int_{\x \in D} k(\x,\x) \; d\x = \int_{\x \in D} \nu c(\mathbf{0}) d\x = \nu$.

In their separable form, over $D = [0,1]^d$, these kernel write $k(\x, \x') =
\nu \prod_{p = 1}^d k_i(x_p, x'_p)$ with $k_i$ one of the aforementioned
kernel. By using the separability we get:
$$\nu w(\x_i, \x_j) =
\int \limits_{\x \in D}  k(\x_i, \x) k(\x_j, \x) d\x = \nu
\prod \limits_{p=1}^d \int \limits_{x \in [0,1]} k_i(x_{i,p}, x) k_i(x_{j,p},
x) dx = \nu \prod \limits_{p=1}^d w_p(x_{i,p}, x_{j,p} ).$$

Below we provide
our parameterization of these kernels in univariate form along with the
corresponding expressions for $w$, $\mathbf{d}$, $\mathbf{c}_1$ and $c_2$.

\subsection{Gaussian kernel}

The univariate Gaussian kernel is $k_G(x, x') = \exp \left( -\frac{(x - x')^2}{\theta} \right)$. Therefore:
$$
d_i = \frac{\partial k_G(x_i, x)}{\partial x} = \dfrac{2 (x_{i}-x)}{\theta}  \exp \left(-\dfrac{(x_{i}-x)^2}{\theta}\right)
$$
$$w(x_i, x_j)= \dfrac{\sqrt{2\pi \theta}}{4} \exp \left(-\dfrac{(x_{i}-x_{j})^2}{2 \theta} \right) \left( \erf \left(\dfrac{2-(x_{i}+x_{j})}{\sqrt{2 \theta}} \right)+ \erf \left( \dfrac{x_{i}+x_{j}}{\sqrt{ 2 \theta}} \right) \right), ~ 1 \leq i,j  \leq n$$
with $\erf$ the error function. In addition:
$$c_2=\dfrac{ \partial w(x_i, x_i)}{\partial x_i}=  \exp \left( -\dfrac{2 x_{i} ^2}{\theta} \right) - \exp \left(-\dfrac{(1-2x_i)^2}{\theta} \right)
$$
and, for the vector $\vecc_1$ , $1 \leq i \leq n$:
\begin{multline*}
\frac{\partial w(x, x_i)}{\partial x} = \sqrt{\frac{\pi}{8 \theta}} \exp \left(-\frac{(x - x_i)^2}{2 \theta}\right) \left[ (x -x_i) \left( \erf \left( \frac{x + x_i -2}{\sqrt{2 \theta}}\right) - \erf\left(\frac{x + x_i}{\sqrt{2 \theta} }\right)  \right) \right. \\
\left. + \sqrt{\frac{2 \theta}{\pi}} \left( \exp \left(- \frac{(x + x_i)^2}{2 \theta} \right) - \exp \left( \frac{- (x + x_i - 2)^2}{2\theta} \right) \right) \right].
\end{multline*}
{\bf Remark:} this is the kernel used in Figure \ref{fig:rep}, with hyperparameters $\nu = 1$, $\theta = 0.01$.

\subsection{Mat\'ern kernels with \texorpdfstring{$\alpha = \{1,3,5\}/2$}{nu = \{1,3,5\}/2}}

We use the following parameterization of the Mat\'ern kernel for specific
values of $\alpha$:
$$ k_{M,1/2}(x, x') = \exp \left( -\frac{|x - x'|}{\theta} \right)$$
$$ k_{M,3/2}(x, x') = \left(1 + \frac{\sqrt{3}|x - x'|}{\theta} \right) \exp \left( -\frac{\sqrt{3} |x - x'|}{\theta} \right)$$
$$ k_{M,5/2}(x, x') = \left(1 + \frac{\sqrt{5}|x - x'|}{\theta} + \frac{5 (x - x')^2}{2\theta^2} \right) \exp \left( -\frac{\sqrt{5} |x - x'|}{\theta} \right) $$
The derivatives with respect to $x$, i.e., in $\vecd$ are:
$$ \frac{ \partial k_{M,1/2}(x, x')}{\partial x} = \frac{(-1)^{\delta_{x < x'}}}{\theta} \exp \left( -\frac{|x - x'|}{\theta} \right)$$
$$ \frac{ \partial k_{M,3/2}(x, x')}{\partial x} = \frac{(-1)^{\delta_{x < x'}} \times 3|x - x'|}{\theta^2} \exp \left( -\frac{\sqrt{3}|x - x'|}{\theta} \right)$$
$$ \frac{ \partial k_{M,5/2}(x, x')}{\partial x} = (-1)^{\delta_{x < x'}} \frac{ \left(\frac{10}{3} - 5 \right) |x - x'| - \frac{5 \sqrt{5}}{3 \theta} (x - x')^2}{\theta^2} \exp \left( -\frac{\sqrt{5}|x - x'|}{\theta} \right)$$
To get closed form derivatives of $w(x_i,x_j)$ in Lemma~\ref{lem:IMSPE_closed}, first consider $x_i \leq x_j$ to drop
absolute values, then divide integration into components $p_1$ ($0 \rightarrow
x_i$), $p_2$ ($x_i \rightarrow x_j$), $p_3$ ($x_j \rightarrow 1$). We rely on
symbolic solvers for the most tedious components, see e.g., \url{https://www
.integral-calculator.com/}.
To reduce expression, define $\beta = \exp \left( \frac{2 \sqrt{3}}{\theta} \right)$ and $\gamma = \exp \left( \frac{2\sqrt{5}}{\theta} \right)$.

The first term is given by:
\begin{equation*}
p_{1,1/2} = \int \limits_0^{x_i} \exp \left(-\frac{(x_i - x)}{\theta} \right) \exp \left(-\frac{(x_j - x)}{\theta} \right) dx
= \frac{\theta}{2} \left( \exp \left( \frac{2x_i}{\theta} \right) -1\right) \exp \left(\frac{-x_j-x_i}{\theta} \right),
\end{equation*}
and similarly
\begin{multline*}
p_{1, 3/2} = \frac{1}{12\theta} \left[ \left(\theta\left(5 \sqrt{3}\theta+9 x_j- 9 x_j \right) \exp\left(\frac{2\sqrt{3}x_i}{\theta} \right)-5\sqrt{3}\theta^2-9\left(x_j+x_i\right)\theta \right. \right. \\
\left. \left. -2{\cdot}3^\frac{3}{2}x_i x_j\right) \exp \left(-\frac{\sqrt{3}\left(x_i+x_j\right)}{\theta}\right) \right]
\end{multline*}%
\vspace{-0.5cm}
\begin{multline*}
p_{1, 5/2} \cdot t_1 =\theta^2\left(63 \theta^2+9{\cdot}5^\frac{3}{2}x_j \theta-9{\cdot}5^\frac{3}{2}x_i \theta + 50 x_j^2-100x_ix_j+50 x_i^2\right) \exp \left( \frac{2 \sqrt{5}x_i}{ \theta}\right)    \\
  - 63 \theta^4 -9{\cdot}5^\frac{3}{2}\left(x_j+x_i\right) \theta^3 -10\left(5x_j^2+17x_ix_j+5 x_i^2\right) \theta^2-8{\cdot}5^\frac{3}{2}x_ix_j\left(x_j+x_i\right) \theta-50x_i^2x_j^2,
\end{multline*}
with $t_1 = 36\sqrt{5} \theta^3 \exp \left(\frac{\sqrt{5}\left(x_j+x_i\right)}{\theta}\right)$.

The second term:
\begin{equation*}
p_{2,1/2} =  \int \limits_{x_i}^{x_j} \exp \left(-\frac{(x - x_i)}{\theta} \right) \exp \left(-\frac{(x_j - x)}{\theta} \right) = (x_j - x_i) \exp \left(-\frac{x_j - x_i}{\theta} \right)
\end{equation*}
\begin{equation*}
p_{2,3/2} = \dfrac{\left(x_j-x_i\right)\left(2\theta^2+2 \sqrt{3}\left(x_j-x_i\right)\theta+x_j^2-2x_i x_j+x_i^2\right) \exp \left(-\frac{\sqrt{3}\left(x_j-x_i\right)}{\theta} \right)}{2\theta^2}
\end{equation*}
\begin{multline*}
p_{2,5/2} \cdot t_2 = \left(x_j- x_i\right) 54\theta^4+\left(54\sqrt{5} x_j -54\sqrt{5} x_i \right)\theta^3+\left(105 x_j ^2-210 x_i x_j+105x_i^2\right)\theta^2 \\ + \left(3{\cdot}5^\frac{3}{2}x_j^3-9{\cdot}5^\frac{3}{2} x_i x_j^2+9{\cdot}5^\frac{3}{2}x_i^2 x_j -3{\cdot}5^\frac{3}{2}x_i^3\right)\theta+5x_j^4-20x_i x_j^3+30x_i^2 x_j^2-20 x_i^3 x_j +5x_i^4
\end{multline*}
with $t_2 = 54\theta^4 \exp \left( \frac{\sqrt{5}\left(x_i-x_j\right)}{\theta} \right)$.

The third term:
\begin{multline*}
p_{3, 1/2} = \int \limits_{x_j}^{1} \exp \left(-\frac{(x - x_i)}{\theta} \right) \exp \left(-\frac{(x - x_j)}{\theta} \right)\\
 = \dfrac{\theta}{2} \left( \exp \left( \frac{x_i -x_j}{\theta} \right) - \exp \left( \frac{x_j + x_i -2 }{\theta} \right) \right)
 \end{multline*}
 \vspace{-0.5cm}
 \begin{multline*}
p_{3, 3/2} \cdot t_3 = \theta \left(5 \theta +3^\frac{3}{2}\left(x_j-x_i\right)\right) \beta \\- \left(\theta\left(5 \theta-3^\frac{3}{2}\left(x_j+x_i-2\right)\right)+6\left(x_i-1\right)x_j-6x_i+6\right)  \exp \left( \frac{2 \sqrt{3}x_j}{\theta} \right)
\end{multline*}
\vspace{-0.5cm}
\begin{multline*}
p_{3,5/2} \cdot t_4 = \exp \left( \frac{2\sqrt{5}x_j}{\theta} \right) \cdot \left[\theta \left(\theta \left(9 \theta \left(7 \theta -5^\frac{3}{2}\left(x_j+ x_i -2\right)\right)+10x_j\left(5 x_j+17 x_i -27\right)  \right. \right. \right.  \\
\left. \left. +10\left(5 x_i^2-27 x_i+27\right)\right) -8{\cdot}5^\frac{3}{2}\left(x_i-1\right)\left(x_j-1\right)\left(x_j+x_i-2\right)\right)+50\left(x_i-1\right)^2\left(x_j-2\right)x_j \\
\left. + 50\left(x_i-1\right)^2\right] -\theta^2\left(63\theta^2+9{\cdot}5^\frac{3}{2}x_j \theta-9{\cdot}5^\frac{3}{2}x_i \theta +50 x_j^2-100 x_i x_j+50x_i^2\right) \gamma
\end{multline*}
with $t_3 = 4 \theta \sqrt{3}  \exp \left( \frac{\sqrt{3}\left(x_j-x_i+2\right)}{\theta} \right)$, $t_4 = -36\sqrt{5}\theta^3 \exp \left( \frac{\sqrt{5}\left(x_j-x_i+2\right)}{\theta} \right)$.

The case when $x_i > x_j$ is obtained by swapping $x_i$ and $x_j$ above.
Derivatives with respect to $x_i$ and $x_j$, to account for
both of these cases, are provided as follows:
\begin{equation*}
\frac{\partial w_{1/2}(x_i,x_j)}{\partial x_i} = -\dfrac{\left(2\left(x_i+ \theta- x_j\right) \exp \left( \frac{2 x_i}{\theta} \right) + \theta \exp \left( \frac{2 x_j}{\theta} \right) -\theta\right) \exp \left( -\frac{x_i +x_j}{\theta} \right) }{2 \theta}
\end{equation*}
\begin{equation*}
\frac{\partial w_{1/2}(x_i,x_j)}{\partial x_j} = \dfrac{\left(\theta \exp \left( \frac{2x_j}{\theta} \right) -2 \exp \left( \frac{2 x_i}{\theta} \right) x_j+2\left(\theta+x_i\right) \exp \left( \frac{2 x_i}{\theta} \right) + \theta \right) \exp \left( -\frac{x_j+x_i}{\theta} \right) }{2 \theta}
\end{equation*}
\vspace{-0.5cm}
\small
\begin{multline*}
\frac{\partial w_{3/2}(x_i,x_j)}{\partial x_i} t_5 =
\exp \left( \frac{2\sqrt{3}x_i}{\theta} \right) \left[2 \sqrt{3} \beta x_i^3+\left(-6 \theta -2{\cdot}3^\frac{3}{2}x_j\right) \beta x_i^2 + \right.  \\ \left.
+\left(\left(\left(6x_j-6\right)\theta-3^\frac{3}{2}\theta^2\right) \exp \left( \frac{2 \sqrt{3}x_j}{\theta} \right) +\left(2 \sqrt{3} \theta^2+12 x_j \theta +2{\cdot}3^\frac{3}{2}x_j^2\right) \beta \right)x_i + \right. \\
\left. \left(2\theta^3+\left(4 \sqrt{3}-\sqrt{3} x_j \right)\theta^2+\left(6-6x_j\right) \theta \right) \exp \left(\frac{2\sqrt{3}x_j}{\theta} \right)+\left(-2\sqrt{3}x_j \theta^2 - 6x_j^2 \theta-2\sqrt{3}x_j^3\right) \beta \right] \\
+\left(-3^\frac{3}{2}\theta^2-6x_j x_i\right) \beta x_i + \left(-2s^3-\sqrt{3}x_j \theta^2\right) \beta
\end{multline*}
\vspace{-0.5cm}
\begin{multline*}
\frac{\partial w_{3/2}(x_i,x_j)}{\partial x_j} t_6 =
\theta \left[\left(3^\frac{3}{2} \theta -6 x_i+6\right) x_j- \theta \left(2 \theta - \sqrt{3}\left(x_i - 4\right)\right)+6 x_i-6\right] \exp \left( \frac{2 \sqrt{3} \left(x_j + x_i\right)}{\theta} \right) \\
 -2\sqrt{3} \exp \left( \frac{2 \sqrt{3} \left(x_i+1\right)}{\theta} \right) x_j^3 - 2\left(3 \theta -3^\frac{3}{2}x_i\right)
  \exp \left( \frac{2 \sqrt{3} \left(x_i+1\right)}{\theta} \right) x_j^2- \\
 \beta \left(2 \sqrt{3}\theta^2 \exp \left( \frac{2 \sqrt{3}x_i}{\theta} \right) -12x_i \theta \exp \left( \frac{2 \sqrt{3}x_i}{\theta} \right) +2{\cdot}3^\frac{3}{2} x_i^2 \exp \left( \frac{2 \sqrt{3}x_i}{\theta} \right)-3^\frac{3}{2}\theta^2-6 x_i \theta \right)x_j \\
 + 2 x_i \left(\sqrt{3} \theta^2 - 3 x_i \theta +\sqrt{3} x_i^2 \right) \exp \left( \frac{2 \sqrt{3}\left(x_i+1\right)}{\theta} \right) + \theta^2 \left(2 \theta+\sqrt{3}x_i\right) \beta
\end{multline*}
\normalsize
with  $t_5 = -4\theta^3 \exp\left( \frac{\sqrt{3} (x_i + x_j + 2)}{\theta} \right)$,
$t_6 = - t_5$.

\small
\begin{multline*}
\frac{\partial w_{5/2}(x_i,x_j)}{\partial x_i} t_7=
\left[ 2 \cdot 5^\frac{3}{2} \gamma x_i^5+\left(-100 \theta-2 \cdot5^\frac{5}{2}x_j\right)
\gamma x_i^4 + \left(18 \cdot 5^\frac{3}{2} \theta^2 + 400 x_j \theta+4 \cdot 5^\frac{5}{2} x_j^2\right)  \gamma x_i^3 + \right. \\
 \left( \left(150 \theta^3+\left(24 \cdot 5^\frac{3}{2}-24 \cdot 5^\frac{3}{2} x_j\right) \theta^2+
\left(150 x_j^2-300 x_j+150\right) \theta\right) \exp \left( \frac{2 \sqrt{5} x_j}{\theta} \right) + \right. \\
\left.
 \left(-210 \theta^3-54 \cdot 5^\frac{3}{2} x_j \theta^2-600 x_j^2 \theta - 4 \cdot 5^\frac{5}{2} x_j^3\right)
\gamma\right) x_i^2 +\left(\left(-3 \cdot 5^\frac{5}{2}\theta^4 +\left(270 x_j-570\right) \theta^3+  \right. \right. \\
\left. \left(-12 \cdot 5^\frac{3}{2} x_j^2+72 \cdot 5^\frac{3}{2} x_j-12 \cdot 5^\frac{5}{2}\right) \theta^2 +
\left(-300 x_j^2+600 x_j-300\right)\theta\right) \exp \left( \frac{2 \sqrt{5}x_j}{\theta} \right) \\
\left.
+\left(42 \sqrt{5} \theta^4+420 x_j \theta^3+54 \cdot 5^\frac{3}{2} x_j^2 \theta^2+400 x_j^3 \theta+2 \cdot 5^\frac{5}{2} x_j^4\right) \gamma\right) x_i + \\
\left. \left(54 \theta^5+\left(108 \sqrt{5}-33 \sqrt{5} x_j\right) \theta^4+\left(30 x_j^2-330 x_j+450\right) \theta^3
+\left(12 \cdot 5^\frac{3}{2} x_j^2-48 \cdot 5^\frac{3}{2} x_j+36 \cdot 5^\frac{3}{2}\right) \theta^2 \right. \right. \\
\left.  +\left(150 x_j^2-300 x_j+150\right) \theta\right) \exp \left( \frac{2 \sqrt{5}x_j}{\theta} \right) +\\
\left. \left(-42 \sqrt{5} x_j \theta^4-210 x_j^2 \theta^3-18 \cdot 5^\frac{3}{2} x_j^3 \theta^2-100 x_j^4 \theta-2 \cdot 5^\frac{3}{2} x_j^5\right)
 \gamma\right] \exp\left( \frac{2 \sqrt{5} x_i}{\theta} \right) + \\
 \left(-150 \theta^3-24 \cdot 5^\frac{3}{2} x_j \theta^2-150 x_j^2 \theta\right) \gamma x_i^2+\left(-3 \cdot 5^\frac{5}{2} \theta^4-270 x_j \theta^3-12 \cdot 5^\frac{3}{2} x_j^2 \theta^2\right) \gamma x_i\\
 +\left(-54 \theta^5-33 \sqrt{5} x_j \theta^4-30 x_j^2 \theta^3\right) \gamma
\end{multline*}
\normalsize
with $t_7 = -108 \theta^5 \exp \left( \frac{\sqrt{5} (x_j + x_i + 2)}{\theta} \right)$.
\vspace{-0.5cm}
\small
\begin{multline*}
 \frac{\partial  w_{5/2}(x_i,x_j)}{\partial x_j} t_7 =
 \left( \left( 150 \theta^3+ 24 \cdot 5^\frac{3}{2} \left(1 - x_i \right) \theta^2+ \left( 150x_i^2-300x_i+150 \right)\theta\right) \exp \left( \frac{2 \sqrt{5} x_i}{\theta} \right) x_j^2 + \right. \\
 \left. \left(-3 \cdot 5^\frac{5}{2} \theta^4+ \left(270 x_i - 570\right)\theta^3 - 12 \cdot 5^\frac{3}{2} \left(x_i^2 - 6 x_i +1\right)\theta^2 -300\left(x_i^2 - 2x_i+1\right)\theta\right) \exp \left( \frac{2 \sqrt{5} x_i}{\theta} \right) x_j \right. \\
 + \left(54 \theta^5+\left(108  \sqrt{5}-33  \sqrt{5} x_i\right) \theta^4+\left(30x_i^2-330x_i+450\right)\theta^3+\left(12 \cdot 5^\frac{3}{2}x_i^2-48 \cdot 5^\frac{3}{2}x_i+36 \cdot 5^\frac{3}{2}\right)\theta^2+ \right.\\
  \left. \left. \left(150x_i^2-300x_i+150\right)\theta\right) \exp \left( \frac{2 \sqrt{5}x_i}{\theta} \right)  \right) \exp \left( \frac{2  \sqrt{5}x_j}{\theta} \right) +2 \cdot 5^\frac{3}{2} \exp \left( 2 \sqrt{5}x_i/\theta+2 \sqrt{5}/\theta \right) x_j^5+ \\
 \left(100\theta-2 \cdot 5^\frac{5}{2}x_i\right) \exp \left( \frac{2 \sqrt{5}(x_i + 1)}{\theta} \right) x_j^4+\left(18 \cdot 5^\frac{3}{2}\theta^2-400x_i \theta+4 \cdot 5^\frac{5}{2}x_i^2\right) \exp \left( \frac{2  \sqrt{5} \left(x_i + 1\right)}{\theta} \right)x_j^3\\
 +\left(\left(210\theta^3-54 \cdot 5^\frac{3}{2}x_i\theta^2+600x_i^2\theta-4 \cdot 5^\frac{5}{2}x_i^3\right) \exp \left( \frac{2 \sqrt{5} \left(x_i + 1\right)}{\theta} \right)+ \right.\\
 \left. \left(-150 \theta^3-24 \cdot 5^\frac{3}{2}x_i\theta^2-150x_i^2\theta\right) \gamma \right)x_j^2+ \left( \left(42 \sqrt{5}\theta^4-420x_i\theta^3+54 \cdot 5^\frac{3}{2}x_i^2\theta^2-400x_i^3\theta+ \right. \right. \\
 \left. \left. 2 \cdot 5^\frac{5}{2}x_i^4\right) \exp \left( \frac{2  \sqrt{5} \left(x_i +1 \right) }{\theta} \right) +\left(-3 \cdot 5^\frac{5}{2}\theta^4-270x_i\theta^3-12 \cdot 5^\frac{3}{2}x_i^2\theta^2\right) \gamma \right)x_j+ \\
 \left(-42  \sqrt{5}x_i\theta^4+210x_i^2\theta^3-18 \cdot 5^\frac{3}{2}x_i^3\theta^2+100x_i^4\theta-2 \cdot 5^\frac{3}{2}x_i^5\right) \exp \left( \frac{2 \sqrt{5}\left( x_i +1\right)}{\theta} \right)+ \\
 \left(-54\theta^5-33  \sqrt{5}x_i\theta^4-30x_i^2\theta^3\right) \gamma
\end{multline*}
\normalsize
\noindent Finally, we provide expressions for $c_2$ from \eqref{eq:dW}:
\begin{equation*}
c_{2, 1/2} = \exp \left( -\frac{2x_i}{\theta} \right);
\end{equation*}
\vspace{-1.0cm}
\begin{multline*}
c_{2, 3/2} \cdot t_8 = \left(3x_i^2-2\left(\sqrt{3}\theta+3\right)x_i+\theta^2+2 \sqrt{3}\theta+3\right) \exp \left( \frac{4 \sqrt{3} x_i}{\theta} \right) - 3 \beta x_i^2 \\- 2 \sqrt{3} \theta \beta x_i -\theta^2 \beta;
\end{multline*}
\vspace{-1.0cm}
\begin{multline*}
c_{2, 5/2} \cdot t_9 = \exp \left( \frac{4\sqrt{5}x_i}{\theta} \right) \cdot \left[25 \theta^4-2\left(3{\cdot}5^\frac{3}{2}\theta+50\right) x_i^3+3\left(\theta \left(25 \theta +6{\cdot}5^\frac{3}{2}\right)+50\right)x_i^2 - \right. \\
\left. 2\left(3 \theta \left(\theta \left(3\sqrt{5} \theta +25\right)+3{\cdot}5^\frac{3}{2}\right)+50\right)x_i+9\theta^4+18\sqrt{5} \theta^3+75\theta^2+6{\cdot}5^\frac{3}{2} \theta +25\right]  -\\
25 \gamma x_i^4 - 6{\cdot}5^\frac{3}{2}\theta \gamma x_i^3 - 75\theta^2 \gamma x_i^2 - 18 \sqrt{5} \theta^3 \gamma x_i - 9\theta^4 \gamma
\end{multline*}
with $t_8 = -\theta^2 \exp \left(\frac{2 \sqrt{3}\left(x_i+1\right)}{\theta} \right)$,
$t_9 = -9\theta^4 \exp \left( \frac{2\sqrt{5}\left(x_i+1\right)}{\theta} \right)$.

\end{document}

%% file: figures/Schema_look_ahead.pdf_tex
\begingroup%
  \makeatletter%
  \providecommand\color[2][]{%
    \errmessage{(Inkscape) Color is used for the text in Inkscape, but the package 'color.sty' is not loaded}%
    \renewcommand\color[2][]{}%
  }%
  \providecommand\transparent[1]{%
    \errmessage{(Inkscape) Transparency is used (non-zero) for the text in Inkscape, but the package 'transparent.sty' is not loaded}%
    \renewcommand\transparent[1]{}%
  }%
  \providecommand\rotatebox[2]{#2}%
  \ifx\svgwidth\undefined%
    \setlength{\unitlength}{679.99999039bp}%
    \ifx\svgscale\undefined%
      \relax%
    \else%
      \setlength{\unitlength}{\unitlength * \real{\svgscale}}%
    \fi%
  \else%
    \setlength{\unitlength}{\svgwidth}%
  \fi%
  \global\let\svgwidth\undefined%
  \global\let\svgscale\undefined%
  \makeatother%
  \begin{picture}(1,0.82352943)%
    \put(0,0){\includegraphics[width=\unitlength,page=1]{Schema_look_ahead.pdf}}%
    \put(0.56,0.5825){\color[rgb]{0,0,0}\makebox(0,0)[c]{\footnotesize{\smash{$x_{n+1} = 0.239$}}}}%
    \put(0.56,0.5525){\color[rgb]{0,0,0}\makebox(0,0)[c]{\smash{\footnotesize{$(95719)$}}}}%
    \put(0.4425,0.405){\color[rgb]{0,0,0}\makebox(0,0)[c]{\smash{\footnotesize{$x_{n+2} = 0.205$}}}}%
    \put(0.4425,0.375){\color[rgb]{0,0,0}\makebox(0,0)[c]{\smash{\footnotesize{$(93992)$}}}}%
    \put(0.32275,0.23){\color[rgb]{0,0,0}\makebox(0,0)[c]{\smash{\footnotesize{$x_{n+3} = 0.293$}}}}%
    \put(0.32275,0.2){\color[rgb]{0,0,0}\makebox(0,0)[c]{\smash{\footnotesize{$(92363)$}}}}%
    \put(0.205,0.052){\color[rgb]{0,0,0}\makebox(0,0)[c]{\smash{\footnotesize{$x_{n+4} = 0.175$}}}}%
    \put(0.205,0.022){\color[rgb]{0,0,0}\makebox(0,0)[c]{\smash{\footnotesize{$(\mathbf{90860})$}}}}%
    \put(0.32275,0.5825){\color[rgb]{0,0,0}\makebox(0,0)[c]{\smash{\footnotesize{$k=6$}}}}%
    \put(0.32275,0.5525){\color[rgb]{0,0,0}\makebox(0,0)[c]{\smash{\footnotesize{$(95724)$}}}}%
    \put(0.675,0.405){\color[rgb]{0,0,0}\makebox(0,0)[c]{\smash{\footnotesize{$k=6$}}}}%
    \put(0.675,0.375){\color[rgb]{0,0,0}\makebox(0,0)[c]{\smash{\footnotesize{$(94014)$}}}}%
    \put(0.09,0.23){\color[rgb]{0,0,0}\makebox(0,0)[c]{\smash{\footnotesize{$k=7$}}}}%
    \put(0.09,0.2){\color[rgb]{0,0,0}\makebox(0,0)[c]{\smash{\footnotesize{$(92365)$}}}}%
    \put(0.56,0.23){\color[rgb]{0,0,0}\makebox(0,0)[c]{\smash{\footnotesize{$k=7$}}}}%
    \put(0.56,0.2){\color[rgb]{0,0,0}\makebox(0,0)[c]{\smash{\footnotesize{$(92370)$}}}}%
    \put(0.795,0.23){\color[rgb]{0,0,0}\makebox(0,0)[c]{\smash{\footnotesize{$k=4$}}}}%
    \put(0.795,0.2){\color[rgb]{0,0,0}\makebox(0,0)[c]{\smash{\footnotesize{$(92434)$}}}}%
    \put(0.4425,0.052){\color[rgb]{0,0,0}\makebox(0,0)[c]{\smash{\footnotesize{$k=4$}}}}%
    \put(0.4425,0.022){\color[rgb]{0,0,0}\makebox(0,0)[c]{\smash{\footnotesize{$(90873)$}}}}%
    \put(0.675,0.052){\color[rgb]{0,0,0}\makebox(0,0)[c]{\smash{\footnotesize{$k=4$}}}}%
    \put(0.675,0.022){\color[rgb]{0,0,0}\makebox(0,0)[c]{\smash{\footnotesize{$(90867)$}}}}%
    \put(0.915,0.052){\color[rgb]{0,0,0}\makebox(0,0)[c]{\smash{\footnotesize{$k=7$}}}}%
    \put(0.915,0.022){\color[rgb]{0,0,0}\makebox(0,0)[c]{\smash{\footnotesize{$(90877)$}}}}%
    \put(0.205,0.405){\color[rgb]{0,0,0}\makebox(0,0)[c]{\smash{\footnotesize{$k=5$}}}}%
    \put(0.205,0.375){\color[rgb]{0,0,0}\makebox(0,0)[c]{\smash{\footnotesize{$(93992)$}}}}%
    \put(0.4425,0.755){\color[rgb]{0,0,0}\makebox(0,0)[c]{\smash{\footnotesize{$(97668)$}}}}%
  \end{picture}%
\endgroup%

%% file: RepOrExp.bbl
\begin{thebibliography}{54}
\newcommand{\enquote}[1]{``#1''}
\expandafter\ifx\csname natexlab\endcsname\relax\def\natexlab#1{#1}\fi

\bibitem[\protect\citename{Anagnostopoulos and Gramacy,
  }2013]{anagnos:gramacy:2013}
Anagnostopoulos, C. and Gramacy, R. (2013).
\newblock \enquote{Information-Theoretic Data Discarding for Dynamic Trees on
  Data Streams.}
\newblock {\em Entropy\/}, 15, 12, 5510--5535.
\newblock ArXiv:1201.5568.

\bibitem[\protect\citename{Ankenman et~al., }2010]{Ankenman2010}
Ankenman, B., Nelson, B.~L., and Staum, J. (2010).
\newblock \enquote{Stochastic kriging for simulation metamodeling.}
\newblock {\em Operations research\/}, 58, 2, 371--382.

\bibitem[\protect\citename{Antognini and Zagoraiou, }2010]{Antognini2010}
Antognini, A.~B. and Zagoraiou, M. (2010).
\newblock \enquote{Exact optimal designs for computer experiments via Kriging
  metamodelling.}
\newblock {\em Journal of Statistical Planning and Inference\/}, 140, 9,
  2607--2617.

\bibitem[\protect\citename{Bachoc, }2013]{Bachoc2013}
Bachoc, F. (2013).
\newblock \enquote{Cross Validation and Maximum Likelihood estimations of
  hyper-parameters of Gaussian processes with model misspecification.}
\newblock {\em Computational Statistics \& Data Analysis\/}, 66, 55--69.

\bibitem[\protect\citename{Barnett, }1979]{barnett:1979}
Barnett, S. (1979).
\newblock {\em Matrix Methods for Engineers and Scientists\/}.
\newblock McGraw-Hill.

\bibitem[\protect\citename{Binois and Gramacy, }2017]{Binois2017}
Binois, M. and Gramacy, R.~B. (2017).
\newblock {\em {\sf hetGP}: Heteroskedastic Gaussian Process Modeling and
  Design under Replication\/}.
\newblock R package version 1.0.0.

\bibitem[\protect\citename{Binois et~al., }2016]{Binois2016}
Binois, M., Gramacy, R.~B., and Ludkovski, M. (2016).
\newblock \enquote{Practical heteroskedastic Gaussian process modeling for
  large simulation experiments.}
\newblock {\em arXiv preprint arXiv:1611.05902\/}.

\bibitem[\protect\citename{Boukouvalas, }2010]{Boukouvalas2010}
Boukouvalas, A. (2010).
\newblock \enquote{Emulation of random output simulators.}
\newblock Ph.D. thesis, Aston University.

\bibitem[\protect\citename{Boukouvalas et~al., }2014]{Boukouvalas2014}
Boukouvalas, A., Cornford, D., and Stehl{\'\i}k, M. (2014).
\newblock \enquote{Optimal design for correlated processes with input-dependent
  noise.}
\newblock {\em Computational Statistics \& Data Analysis\/}, 71, 1088--1102.

\bibitem[\protect\citename{Burnaev and Panov, }2015]{Burnaev2015}
Burnaev, E. and Panov, M. (2015).
\newblock \enquote{Adaptive design of experiments based on {G}aussian
  processes.}
\newblock In {\em Statistical Learning and Data Sciences\/},  116--125.
  Springer.

\bibitem[\protect\citename{Chen and Zhou, }2014]{Chen2014}
Chen, X. and Zhou, Q. (2014).
\newblock \enquote{Sequential experimental designs for stochastic kriging.}
\newblock In {\em Proceedings of the 2014 Winter Simulation Conference\/},
  3821--3832. IEEE Press.

\bibitem[\protect\citename{Chen and Zhou, }2017]{Chen2017}
--- (2017).
\newblock \enquote{Sequential design strategies for mean response surface
  metamodeling via stochastic kriging with adaptive exploration and
  exploitation.}
\newblock {\em European Journal of Operational Research\/}, 262, 2, 575--585.

\bibitem[\protect\citename{Chevalier et~al., }2014]{Chevalier2014}
Chevalier, C., Ginsbourger, D., and Emery, X. (2014).
\newblock \enquote{Corrected kriging update formulae for batch-sequential data
  assimilation.}
\newblock In {\em Mathematics of Planet Earth\/},  119--122. Springer.

\bibitem[\protect\citename{Cioffi-Revilla, }2014]{cioffi2014introduction}
Cioffi-Revilla, C. (2014).
\newblock {\em Introduction to computational social science\/}.
\newblock Berlin/New York: Springer.

\bibitem[\protect\citename{Das and Kempe, }2008]{Das2008}
Das, A. and Kempe, D. (2008).
\newblock \enquote{Algorithms for subset selection in linear regression.}
\newblock In {\em Proceedings of the fortieth annual ACM symposium on Theory of
  computing\/},  45--54. ACM.

\bibitem[\protect\citename{Forrester et~al., }2008]{Forrester2008}
Forrester, A., Sobester, A., and Keane, A. (2008).
\newblock {\em Engineering design via surrogate modelling: a practical
  guide\/}.
\newblock John Wiley \& Sons.

\bibitem[\protect\citename{Gauthier and Pronzato, }2014]{Gauthier2014}
Gauthier, B. and Pronzato, L. (2014).
\newblock \enquote{Spectral approximation of the IMSE criterion for optimal
  designs in kernel-based interpolation models.}
\newblock {\em SIAM/ASA Journal on Uncertainty Quantification\/}, 2, 1,
  805--825.

\bibitem[\protect\citename{Ginsbourger and Le~Riche, }2010]{Ginsbourger2010}
Ginsbourger, D. and Le~Riche, R. (2010).
\newblock \enquote{Towards Gaussian process-based optimization with finite time
  horizon.}
\newblock In {\em mODa 9--Advances in Model-Oriented Design and Analysis\/},
  89--96. Springer.

\bibitem[\protect\citename{Gneiting and Raftery, }2007]{gneiting:raftery:2007}
Gneiting, T. and Raftery, A.~E. (2007).
\newblock \enquote{Strictly proper scoring rules, prediction, and estimation.}
\newblock {\em Journal of the American Statistical Association\/}, 102, 477,
  359--378.

\bibitem[\protect\citename{Goldberg et~al.,
  }1998]{goldberg:williams:bishop:1998}
Goldberg, P.~W., Williams, C.~K., and Bishop, C.~M. (1998).
\newblock \enquote{Regression with input-dependent noise: A {G}aussian process
  treatment.}
\newblock In {\em Advances in Neural Information Processing Systems\/},
  vol.~10,  493--499. Cambridge, MA: MIT press.

\bibitem[\protect\citename{Gonzalez et~al., }2016]{Gonzalez2016}
Gonzalez, J., Osborne, M., and Lawrence, N. (2016).
\newblock \enquote{GLASSES: Relieving The Myopia Of Bayesian Optimisation.}
\newblock In {\em Proceedings of the 19th International Conference on
  Artificial Intelligence and Statistics\/},  790--799.

\bibitem[\protect\citename{Gorodetsky and Marzouk, }2016]{Gorodetsky2016}
Gorodetsky, A. and Marzouk, Y. (2016).
\newblock \enquote{Mercer kernels and integrated variance experimental design:
  connections between Gaussian process regression and polynomial
  approximation.}
\newblock {\em SIAM/ASA Journal on Uncertainty Quantification\/}, 4, 1,
  796--828.

\bibitem[\protect\citename{Gramacy and Polson, }2011]{gramacy:polson:2011}
Gramacy, R. and Polson, N. (2011).
\newblock \enquote{Particle Learning of {G}aussian Process Models for
  Sequential Design and Optimization.}
\newblock {\em Journal of Computational and Graphical Statistics\/}, 20, 1,
  102--118.

\bibitem[\protect\citename{Gramacy and Lee, }2009]{gra:lee:2009}
Gramacy, R.~B. and Lee, H. K.~H. (2009).
\newblock \enquote{Adaptive Design and Analysis of Supercomputer Experiment.}
\newblock {\em Technometrics\/}, 51, 2, 130--145.

\bibitem[\protect\citename{Hong and Nelson, }2006]{hong:nelson:2006}
Hong, L. and Nelson, B. (2006).
\newblock \enquote{Discrete optimization via simulation using {COMPASS}.}
\newblock {\em Operations Research\/}, 54, 1, 115--129.

\bibitem[\protect\citename{Horn et~al., }2017]{Horn2017}
Horn, D., Dagge, M., Sun, X., and Bischl, B. (2017).
\newblock \enquote{First Investigations on Noisy Model-Based Multi-objective
  Optimization.}
\newblock In {\em International Conference on Evolutionary Multi-Criterion
  Optimization\/},  298--313. Springer.

\bibitem[\protect\citename{Huan and Marzouk, }2016]{Huan2016}
Huan, X. and Marzouk, Y.~M. (2016).
\newblock \enquote{Sequential Bayesian optimal experimental design via
  approximate dynamic programming.}
\newblock {\em arXiv preprint arXiv:1604.08320\/}.

\bibitem[\protect\citename{Jalali et~al., }2017]{Jalali2017}
Jalali, H., Nieuwenhuyse, I.~V., and Picheny, V. (2017).
\newblock \enquote{Comparison of Kriging-based algorithms for simulation
  optimization with heterogeneous noise.}
\newblock {\em European Journal of Operational Research\/}, 261, 1, 279 -- 301.

\bibitem[\protect\citename{Johnson, }2008]{johnson:2008}
Johnson, L. (2008).
\newblock \enquote{Microcolony and Biofilm Formation as a Survival Strategy for
  Bacteria.}
\newblock {\em Journal of Theoretical Biology\/}, 251, 24--34.

\bibitem[\protect\citename{Kami{\'n}ski, }2015]{Kaminski2015}
Kami{\'n}ski, B. (2015).
\newblock \enquote{A method for the updating of stochastic Kriging metamodels.}
\newblock {\em European Journal of Operational Research\/}, 247, 3, 859--866.

\bibitem[\protect\citename{Kersting et~al., }2007]{kersting:etal:2007}
Kersting, K., Plagemann, C., Pfaff, P., and Burgard, W. (2007).
\newblock \enquote{Most likely heteroscedastic {G}aussian process regression.}
\newblock In {\em Proceedings of the International Conference on Machine
  Learning\/},  393--400. New York, NY: ACM.

\bibitem[\protect\citename{Kleijnen, }2015]{Kleijnen2015}
Kleijnen, J.~P. (2015).
\newblock {\em Design and Analysis of Simulation Experiments\/}, vol. 230.
\newblock Springer.

\bibitem[\protect\citename{Krause and Guestrin, }2007]{Krause2007}
Krause, A. and Guestrin, C. (2007).
\newblock \enquote{Nonmyopic active learning of gaussian processes: an
  exploration-exploitation approach.}
\newblock In {\em Proceedings of the 24th international conference on Machine
  learning\/},  449--456. ACM.

\bibitem[\protect\citename{Krause et~al., }2008]{Krause2008}
Krause, A., Singh, A., and Guestrin, C. (2008).
\newblock \enquote{Near-optimal sensor placements in Gaussian processes:
  Theory, efficient algorithms and empirical studies.}
\newblock {\em Journal of Machine Learning Research\/}, 9, Feb, 235--284.

\bibitem[\protect\citename{Lam et~al., }2016]{Lam2016}
Lam, R., Willcox, K., and Wolpert, D.~H. (2016).
\newblock \enquote{Bayesian Optimization with a Finite Budget: An Approximate
  Dynamic Programming Approach.}
\newblock In {\em Advances In Neural Information Processing Systems\/},
  883--891.

\bibitem[\protect\citename{Law, }2015]{Law2015}
Law, A.~M. (2015).
\newblock {\em Simulation Modeling and Analysis\/}.
\newblock 5th ed. McGraw-Hill.

\bibitem[\protect\citename{Leatherman et~al., }2017]{Leatherman2017}
Leatherman, E.~R., Santner, T.~J., and Dean, A.~M. (2017).
\newblock \enquote{Computer experiment designs for accurate prediction.}
\newblock {\em Statistics and Computing\/},  1--13.

\bibitem[\protect\citename{Liu and Staum, }2010]{Liu2010}
Liu, M. and Staum, J. (2010).
\newblock \enquote{Stochastic kriging for efficient nested simulation of
  expected shortfall.}
\newblock {\em The Journal of Risk\/}, 12, 3, 3.

\bibitem[\protect\citename{Mehdad and Kleijnen, }2018]{Mehdad2018}
Mehdad, E. and Kleijnen, J.~P. (2018).
\newblock \enquote{Stochastic intrinsic Kriging for simulation metamodeling.}
\newblock {\em Applied Stochastic Models in Business and Industry\/},  in
  press.

\bibitem[\protect\citename{M{\"u}ller et~al., }2012]{Mueller2012}
M{\"u}ller, W.~G., Pronzato, L., and Waldl, H. (2012).
\newblock \enquote{Relations between designs for prediction and estimation in
  random fields: an illustrative case.}
\newblock In {\em Advances and Challenges in Space-time Modelling of Natural
  Events\/},  125--139. Springer.

\bibitem[\protect\citename{Petersen et~al., }2008]{Petersen2008}
Petersen, K.~B., Pedersen, M.~S., et~al. (2008).
\newblock \enquote{The matrix cookbook.}
\newblock {\em Technical University of Denmark\/}, 7, 15.

\bibitem[\protect\citename{Picheny and Ginsbourger, }2013]{Picheny:2013}
Picheny, V. and Ginsbourger, D. (2013).
\newblock \enquote{A nonstationary space-time {G}aussian process model for
  partially converged simulations.}
\newblock {\em SIAM/ASA Journal on Uncertainty Quantification\/}, 1, 57--78.

\bibitem[\protect\citename{Plumlee and Tuo, }2014]{Plumlee2014}
Plumlee, M. and Tuo, R. (2014).
\newblock \enquote{Building accurate emulators for stochastic simulations via
  quantile Kriging.}
\newblock {\em Technometrics\/}, 56, 4, 466--473.

\bibitem[\protect\citename{Pratola et~al., }2017]{Pratola2017}
Pratola, M.~T., Harari, O., Bingham, D., and Flowers, G.~E. (2017).
\newblock \enquote{Design and Analysis of Experiments on Nonconvex Regions.}
\newblock {\em Technometrics\/},  1--12.

\bibitem[\protect\citename{Pronzato and M{\"u}ller, }2012]{Pronzato2012}
Pronzato, L. and M{\"u}ller, W.~G. (2012).
\newblock \enquote{Design of computer experiments: space filling and beyond.}
\newblock {\em Statistics and Computing\/}, 22, 3, 681--701.

\bibitem[\protect\citename{Quan et~al., }2013]{Quan2013}
Quan, N., Yin, J., Ng, S.~H., and Lee, L.~H. (2013).
\newblock \enquote{Simulation optimization via kriging: a sequential search
  using expected improvement with computing budget constraints.}
\newblock {\em IIE Transactions\/}, 45, 7, 763--780.

\bibitem[\protect\citename{{R Core Team}, }2017]{cran:R}
{R Core Team} (2017).
\newblock {\em R: A Language and Environment for Statistical Computing\/}.
\newblock R Foundation for Statistical Computing, Vienna, Austria.

\bibitem[\protect\citename{Rasmussen and Williams, }2006]{Rasmussen2006}
Rasmussen, C.~E. and Williams, C. (2006).
\newblock {\em {Gaussian Processes for Machine Learning}\/}.
\newblock MIT Press.

\bibitem[\protect\citename{Sacks et~al., }1989]{Sacks1989}
Sacks, J., Welch, W.~J., Mitchell, T.~J., and Wynn, H.~P. (1989).
\newblock \enquote{Design and analysis of computer experiments.}
\newblock {\em Statistical science\/}, 4, 4, 409--423.

\bibitem[\protect\citename{Santner et~al., }2013]{Santner2013}
Santner, T.~J., Williams, B.~J., and Notz, W.~I. (2013).
\newblock {\em The design and analysis of computer experiments\/}.
\newblock Springer Science \& Business Media.

\bibitem[\protect\citename{Seo et~al., }2000]{seo00}
Seo, S., Wallat, M., Graepel, T., and Obermayer, K. (2000).
\newblock \enquote{Gaussian Process Regression: Active Data Selection and Test
  Point Rejection.}
\newblock In {\em Proceedings of the International Joint Conference on Neural
  Networks\/}, vol. III,  241--246. IEEE.

\bibitem[\protect\citename{Wang and Haaland, }2017]{Wang2017}
Wang, W. and Haaland, B. (2017).
\newblock \enquote{Controlling Sources of Inaccuracy in Stochastic Kriging.}
\newblock {\em arXiv preprint arXiv:1706.00886\/}.

\bibitem[\protect\citename{Weaver et~al., }2016]{Weaver2016}
Weaver, B.~P., Williams, B.~J., Anderson-Cook, C.~M., Higdon, D.~M., et~al.
  (2016).
\newblock \enquote{Computational enhancements to Bayesian design of experiments
  using Gaussian processes.}
\newblock {\em Bayesian Analysis\/}, 11, 1, 191--213.

\bibitem[\protect\citename{Xie et~al., }2012]{xie:frazier:chick:2012}
Xie, J., Frazier, P., and Chick, S. (2012).
\newblock \enquote{Assemble to Order Simulator.}

\end{thebibliography}
